\documentclass[11pt]{article}
\usepackage{amsmath}
\usepackage{amssymb}
\usepackage{amsthm,thmtools,thm-restate}
\usepackage{graphicx} % Required for inserting images
\usepackage{fullpage}
\usepackage{hyperref}
\usepackage{xcolor}
\usepackage{colortbl}
\usepackage{multicol}
\usepackage{url}
\usepackage[utf8]{inputenc}
\usepackage[margin=1in]{geometry}
\usepackage{cleveref}
\usepackage{doi}
\usepackage{lineno}
\usepackage{algorithm}
\usepackage{algpseudocode}

\usepackage{tikz}
\usetikzlibrary{matrix,calc,arrows,positioning,graphs}

\newif\iflong
\newif\ifshort

\longtrue
    
\iflong
\else
\shorttrue
\fi

\usepackage{amsmath,amssymb,amsfonts,amsthm}

\usepackage{graphicx}

\usepackage{mathtools}
\usepackage{xspace}
\usepackage{enumerate}
\usepackage{todonotes}
\usepackage{gosip} % MR- Do we need that??
\usepackage{url}
\usepackage{hyperref}
\usepackage{comment}

\newtheorem{theorem}{Theorem}
\numberwithin{theorem}{section}
\newtheorem{lemma}[theorem]{Lemma}

\newtheorem{observation}[theorem]{Observation}

\newtheorem{proposition}[theorem]{Proposition}
\newtheorem{corollary}[theorem]{Corollary}

\newtheorem*{claim*}{Claim}

\newenvironment{claimproof}[1]{\par\noindent{\emph{Proof.}}\space#1}{\hfill $\blacksquare$}

\theoremstyle{definition}
\newtheorem{definition}[theorem]{Definition}

\newtheorem{example}[theorem]{Example}

\DeclareMathOperator{\var}{var}

\usepackage{graphicx}

\definecolor{ColorM}{RGB}{202,132,72}
\definecolor{IndianRed}{RGB}{205,92,92}

\newcommand{\cb}[1]{\textcolor{blue}{#1}}
\newcommand{\cg}[1]{\textcolor{green}{#1}}
\newcommand{\cy}[1]{\textcolor{yellow!80!black}{#1}}
\newcommand{\crr}[1]{\textcolor{red}{#1}}

\newcommand{\disj}{\textsf{disj}}

\newcommand{\mnu}{\text{N}}

\newcommand{\bigoh}{\mathcal{O}}

\newcommand{\SB}{\{\,}
\newcommand{\SM}{\;{|}\;}
\newcommand{\SE}{\,\}}
\def\hy{\hbox{-}\nobreak\hskip0pt}
\newcommand{\CCC}{\mathcal{C}}

\newcommand{\TTT}{\mathcal{T}}

\newcommand{\vars}{\ensuremath{V}}

\newcommand{\pref}{\ensuremath{\mathcal{Q}}}

\newcommand{\magic}{\textsf{part}}

\newcommand{\pies}{\textsf{RED}}
\newcommand{\piesup}{\textsf{REDup}}
\newcommand{\del}{\textsf{SEL}}

\newcommand{\ariti}{\SA}
\newcommand{\SA}{\Delta}
\newcommand{\xor}{\oplus}
\newcommand{\ass}{\textsf{assign}}
\newcommand{\propT}[1]{#1^{\omega}}

\newcommand{\repl}[1]{[#1]}

\newcommand{\NUMCNF}[1]{\textup{#1-CNF}\xspace}
\newcommand{\TCNF}{\NUMCNF{2}}
\newcommand{\DCNF}{\NUMCNF{\ensuremath{d}}}
\newcommand{\CNF}{\textup{CNF}\xspace}
\newcommand{\DNF}{\textup{DNF}\xspace}
\newcommand{\AFF}{\textup{AFF}\xspace}
\newcommand{\HORN}{\textup{HORN}\xspace}
\newcommand{\THORN}{\textup{2-HORN}\xspace}

\newcommand{\DUALHORN}{\textup{DualHORN}\xspace}

\newcommand{\QBF}{\textsc{QBF}\xspace}

\newcommand{\sizeOf}[1]{\|#1\|}

\newcommand{\cnfpropname}{{\textsf{prop}}}
\newcommand{\cnfprop}[1]{{\textsf{prop}}({#1})}
\newcommand{\CL}{C_\mathcal{L}}

\newcommand{\Asize}{2k^2 + k}
\newcommand{\DDD}{\mathcal{D}}
\newcommand{\primG}{G_{\textsf{P}}}
\newcommand{\primincG}{G_{\textsf{PI}}}

\newcommand{\RU}{U}

\crefname{figure}{figure}{figures}
\Crefname{figure}{Figure}{Figures}
\crefname{theorem}{theorem}{theorems}
\Crefname{theorem}{Theorem}{Theorems}
\crefname{lemma}{lemma}{lemmas}
\Crefname{lemma}{Lemma}{Lemmas}
\crefname{definition}{definition}{definitions}
\Crefname{definition}{Definition}{Definitions}

\newcommand{\drgiht}[3]{
    \draw[->, thick, #3]
    (#1.north east) to[out=30, in=150] (#2.north west);
}
\newcommand{\dleft}[3]{
    \draw[->, thick, #3]
    (#1.south west) to[out=210, in=330] (#2.south east);
}
\newcommand{\dposlit}[2]{
    \node[left=1.5pt, #2] at (#1) {\texttt{+}};
}
\newcommand{\dneglit}[2]{
    \node[left=1.5pt, #2] at (#1) {\texttt{\textbf{-}}};
}
\newcommand{\inittikz}{
    \node (A) at (0,0) {$x_1$};
    \node (B) at (1,0) {$x_2$};
    \node (C) at (2,0) {$x_3$};
    \node (D) at (3,0) {$y_1$};
    \node (E) at (4,0) {$y_2$};
    \node (F) at (5,0) {$y_3$};
    \node (G) at (6,0) {$z_1$};
    \node (H) at (7,0) {$z_2$};
    \node (I) at (8,0) {$z_3$};
}
\newcommand{\inittikzIV}{
    \node (A) at (0,0) {$x_1$};
    \node (B) at (1,0) {$x_2$};
    \node (C) at (2,0) {$x_3$};
    \node (A1) at (3,0) {$a_1$};
    \node (A2) at (4,0) {$a_2$};
    \node (A3) at (5,0) {$a_3$};
    \node (A4) at (6,0) {$a_4$};
    \node (D) at (7,0) {$y_1$};
    \node [fill=yellow!30, inner sep=1pt, rounded corners] at (8,0) {$y_2$};
    \node (E) at (8,0) {$y_2$};
    \node (F) at (9,0) {$y_3$};
    \node (G) at (10,0) {$z_1$};
    \node (H) at (11,0) {$z_2$};
    \node (I) at (12,0) {$z_3$};
}

\title{Backdoors for Quantified Boolean Formulas}

\author{
Leif Eriksson\thanks{Link{\"o}ping University, Sweden, \texttt{leif.eriksson@liu.se}}
\and 
Victor Lagerkvist\thanks{Link{\"o}ping University, Sweden, \texttt{victor.lagerkvist@liu.se}}
\and
Sebastian Ordyniak\thanks{University of Leeds, UK, \texttt{sordyniak@gmail.com}}
\and
George Osipov\thanks{Royal Holloway, University of London, UK, \texttt{george.osipov@pm.me}}
\and 
Fahad Panolan\thanks{University of Leeds, UK, \texttt{fahad.panolan@gmail.com}}
\and
Mateusz Rychlicki\thanks{University of Leeds, UK, \texttt{mkrychlicki@gmail.com}}\\
}

\date{}

\begin{document}

\maketitle

\thispagestyle{empty}

\begin{abstract}
%    \todo{GO: Uncommented the long abstract and commented out the short one.}
    The {\em quantified Boolean formula} problem is a well-known PSpace-complete problem with rich expressive power, and is generally viewed as the SAT analogue for PSpace. Given that many problems today are solved in practice by reducing to SAT, and then using highly optimized SAT solvers, it is natural to ask whether problems in PSpace are amenable to this approach. While SAT solvers exploit hidden structural properties, such as \emph{backdoors} to tractability, backdoor analysis for QBF is comparatively very limited.

    We present a comprehensive study of the complexity of QBF parameterized by backdoor size
    to the largest tractable syntactic classes: HORN, 2-SAT, and AFFINE.
    While SAT is in FPT under this parameterization, we prove that QBF remains PSpace-hard even on formulas with backdoors of constant size.
    Parameterizing additionally by the quantifier depth, we design FPT-algorithms 
    for the classes 2-SAT and AFFINE, and show that 3-HORN is W[1]-hard. 

    As our next contribution, we vastly extend the applicability of QBF backdoors not only for the syntactic classes defined above but also for tractable classes defined via structural restrictions, such as formulas with bounded incidence treewidth and quantifier depth.
    To this end, we introduce \emph{enhanced backdoors}:
    these are separators $S$ of size at most $k$ in the primal graph
    such that $S$ together with all variables contained in any purely universal component of the primal graph minus $S$ is a backdoor. We design FPT-algorithms with respect to $k$ for both evaluation and detection of enhanced backdoors to all tractable classes of QBF listed above and more.
    
    The techniques used in polynomial-time algorithms for QBF as well as
    FPT-results for SAT backdoors are mostly inapplicable in our setting.
    A simple yet elucidating first step for all our results is a reduction from formulas with backdoors of size $k$
    to formulas which are disjunctions of at most $2^k$ members of the tractable base class.
    We refer to the latter as \emph{disjuncts}.
    Our algorithms and lower bounds work by ``trading off'' the number of disjuncts for quantifier depth.
    For the PSpace-hardness results, we develop a \emph{squishing procedure} which can be seen as the ``inverse'' of quantifier elimination and allows us to drastically decrease the number of disjuncts by adding only two quantifier blocks.
    For our algorithms for 2-SAT and AFFINE, we use delicate and tailor-made procedures to replace 
    the variables of one (quantifier) block with a bounded number of novel and existing variables at the expense of adding only a bounded number of new disjuncts. In the case of AFFINE we can achieve this for the innermost quantifier block, whereas in the case of 2-SAT this is only achievable for the innermost (mostly second-to-last) existential block.
   To make enhanced backdoors applicable to QBF, we design the novel procedure of \emph{elimination of guarded universal sets}, 
    which allows us to reduce an arbitrarily large set $Y$ of universal variables to one whose size is bounded by the number of its neighbors in the primal graph plus the maximum number of variables from $Y$ that occur in any clause.
    Lastly, for detecting enhanced backdoors we use a mix of well-known FPT-techniques such as important separators as well as the randomized contraction and graph unbreakability framework.
  \end{abstract}    

\pagebreak
\setcounter{tocdepth}{2} 
\tableofcontents
\thispagestyle{empty}

% \listoftodos

\pagebreak
\setcounter{page}{1}

\section{Introduction}

The {\em quantified Boolean formula} problem (QBF) is a generalization of the well-known NP-complete Boolean {\em satisfiability} problem (SAT). Here, the objective is to find the truth value of a given quantified Boolean formula $Q_1 x_1 \ldots Q_n x_n . \phi(x_1,\ldots,x_n)$, where $Q_i \in \{\forall, \exists\}$ for every $1 \leq i \leq n$.  
QBF is a canonical PSpace-complete problem, playing a similar role for
PSpace as SAT does for NP, by both serving as the starting point for
countless hardness proofs, as well as being a  natural problem to
reduce to. The distinction between NP and PSpace is important since
many problems in logic and artificial intelligence, e.g., model
checking,  non-monotonic reasoning, training or deciding properties of
neural networks, and symbolic planning~\cite{qbfsurvey}, are harder
than NP, and seem better suited to being solved via QBF.

Before digging deeper into QBF let us briefly explain the practical and theoretical importance of SAT. Due to the rapid advances of SAT solvers~\cite{DBLP:conf/lics/FichteGHSO23} there has been a surge of interest in reducing problems {\em to} SAT, and many problems in NP can be solved efficiently in practice with off-the-shelf SAT solvers. 
There has simultaneously been intense theoretical research on solving SAT faster, both via unconditional, exponential algorithms~\cite{WangW16}, as well as faster algorithms in the parameterized complexity setting. Here, the goal is either to prove \emph{fixed parameter tractability} (FPT) or conditional hardness (e.g., W[1]-hardness).
For problems in NP there exists a wealth of FPT-results~\cite{parameterizedalgorithmsbook} but for larger classes, and in particular PSpace, the FPT landscape is comparatively sparse.

A prominent FPT approach for SAT and many other problems is the {\em backdoor} approach where the goal is to extend tractable fragments (e.g., 2-SAT or Horn) by identifying a set of variables $B$ such that one reaches a tractable class once all variables in $B$ are assigned/removed. Thus, $|B|$  intuitively measures the distance to a tractable class, and was 
introduced by Williams et al. \cite{WilliamsGS03} and Crama et al. \cite{CramaEH97}, with the motivation of providing a theoretical justification of how SAT solvers in practice frequently beat the theoretically expected bound of $2^n$.
More specifically, let us fix a {\em base class} $\CCC$ of \CNF formulas which is solvable in polynomial time. For a \CNF formula $\phi$, a subset $B$ of the variables in $\phi$ is a {\em (strong) backdoor} if $\phi[\tau] \in \CCC$
for every $\tau \colon B \to \{0,1\}$, where $\phi[\tau]$ is the reduced formula obtained by simplifying according to the partial truth assignment $\tau$. 
Moreover, a subset $B$ is a {\em (variable) deletion backdoor} if we
instead require that $\phi - B \in \CCC$, where $\phi - B$ is the
formula obtained by removing all occurrences of variables in
$B$. 
For SAT, the {\em backdoor evaluation} problem parameterized by the size of the backdoor $B$ is then trivial since we can determine the satisfiability of $\phi$ by enumerating all $2^{|B|}$ possible assignments, and the theoretically interesting question is rather whether a backdoor exists ({\em backdoor detection}).
This question has been extensively studied with respect to many tractable base classes (see, e.g.,~\cite{FominLMRS15,GaspersOrdyniakRamanujanSaurabhSzeider13,
GaspersSzeider13, LokshtanovP022,RazgonOSullivan09,Szeider07a} and the references within) and there are many positive FPT-results.
Again, we stress that backdoors to tractable classes are widely seen as the theoretical underpinning of the SAT-based revolution for tackling NP-complete problems.

With this in mind we now return to QBF.
Like in the case of SAT, QBF is also solvable in polynomial time when
the Boolean formula is \TCNF, \HORN/\DUALHORN (at most one positive/negative literal), or expressible as a system of linear equations (\AFF)~\cite{DBLP:books/daglib/0004131}. 
Hence, we have some reasonable backdoor classes to target, and to make our question more precise we introduce a little bit of notation. First, we often represent conjunctive formulas by sets, and for a class of formulas $\CCC$ and $d \geq 1$ we write $d$-$\CCC$ for the formulas in $\CCC$ of arity at most $d$. We then write $\QBF(\CCC)$ for the computational problem of deciding whether a QBF formula $\cQ . \phi$ where $\phi \in \CCC$, is true. 
The basic notion of a backdoor for QBF (with respect to some
underlying tractable class of $\QBF(\CCC)$ formulas $\cB$) is then
defined analogously as for SAT, e.g., $B$ is a {\em (strong) backdoor} of
$\Psi = \cQ . \phi$ if $\Psi[\tau] \in \cB$ for every $\tau \colon B
\to \{0,1\}$\footnote{In the sequel we often do not make a sharp
  distinction between the class of formulas $\CCC$ and the class of
  $\QBF(\CCC)$ formulas that it induces.}.
As for SAT, the notions of deletion backdoors and backdoors coincide for the base classes \TCNF and \HORN, but
differ for \AFF as well as all classes defined via
structural restrictions considered in this paper.

Similarly to SAT we then have two fundamental problems: backdoor
evaluation and backdoor detection. However, since our backdoor notion
is purely syntactic, it is easy to see that (deletion) backdoor detection for $\QBF$ typically coincides with (deletion) backdoor detection for SAT (however, interesting differences arise if we consider alternative types of backdoor classes, and we return to this in Section~\ref{ssec:enh}). Here, the
problem is FPT for $\TCNF$ and
$\HORN$~\cite{DBLP:conf/birthday/GaspersS12}, and for $\AFF$ if
additionally parameterized by the maximum arity of any involved
relation~\cite[Theorem 4]{DBLP:conf/dagstuhl/GaspersOS17}. Thus, for
those classes, only the evaluation problem remains, which in contrast to SAT is far from straightforward since QBF, generally, is incompatible with branching, except in very restricted cases. For example, branching on the backdoor variables can be done if the backdoor $B$ completely covers a prefix of the input variables. 
This was also observed by Samer and Szeider~\cite{SamerSzeider07a} who
considered backdoors for QBF to $\TCNF$ and $\HORN$ under the
assumption that the backdoor is closed with respect to ``variable
dependencies''.
This is, for example, applicable to instances
where any given backdoor variable implies that also every variable to
the left of it in the prefix is inside the backdoor, but seems too
restrictive to provide a {\em general} method for solving QBF. We
make this clear in Section~\ref{sec:comp-dep-back}, where we provide a
class of instances (with only one quantifier alternation) with strong backdoors of size one, but where the dependency backdoors have a linear amount of variables.

Hence, it is fair to ask (as done by Williams~\cite{williams2013backdoor}) 
whether backdoors/branching and QBF are incompatible, and if the lack
of FPT-results  explains why QBF is hard in practice. In an early paper, Williams~\cite{williams03} (in a slightly different context) even explicitly states the following.

\begin{quote}
\emph{One hard problem that probably cannot be computed with it [William's framework] is solving quantified Boolean formulas. This is because in QBFs, it seems crucial to maintain the fixed variable ordering given by the quantifiers.}
\end{quote}

However, on a more positive note, Williams also asks whether this can be circumvented in the bounded quantifier alternation setting, but this (neither with respect to upper, nor lower, bounds) has not been systematically addressed in the literature and remains a blind spot.
Thus, could backdoors to e.g. $\TCNF$, $\HORN$, or $\AFF$ provide new FPT-algorithms for QBF? If not, can we prove hardness, and if so, consider additional structural restrictions (e.g., bounded quantifier alternations) which give FPT but which are still milder than variable dependency schemes? If successful, could there exist even broader classes (e.g., defined via graph width parameters) suitable as backdoors?
As we will now see, we successfully resolve all of these questions and give a fresh picture of the applicability of backdoors to QBF.

\subsection{Our Contributions}

We first consider the backdoor evaluation problem for QBF with
$\TCNF$, $\HORN$, and $\AFF$, as tractable base classes. We begin
(summarized in Section~\ref{ssec:disjunct}) by proving that the
existence of a backdoor is far from sufficient to solve QBF
efficiently. In fact, backdoor evaluation is not only unlikely to be
FPT but is para-PSpace-complete, i.e., a backdoor does not
change the complexity classification at all. Given the simplicity of this question and
the ubiquity of backdoor results, one could perhaps expect this
hardness result to be readily available in the literature. However, no
such results exist, and their absence has a natural explanation: as
far as we know it cannot be done by any existing techniques and to
prove this we develop the novel technique of {\em squishing}.

We then explore using additional restrictions to obtain
FPT-algorithms for QBF using backdoors. As a first step in this direction, we show
in~\Cref{ssec:alt} that QBF becomes FPT parameterized by backdoor size
for the base classes \TCNF{} and \AFF{}, if additionally parameterized by the number
of quantifier alternations.
But surprisingly remains W$[1]$\hy hard for
the base class $3$-\HORN. We then develop a novel and very general
solution paradigm for QBF in \Cref{ssec:enh} that informally allows us to
efficiently eliminate an
arbitrarily large set of universal quantifiers as long as the set has a
small separator from the remaining instance. This allows us to introduce
so-called \emph{enhanced backdoors} that can be arbitrarily smaller than
backdoors but can still be applied to the same base classes. We consider both strong and deletion variants of these backdoors but concentrate on deletion backdoors for structural classes such as bounded treewidth. 

As a consequence, we are able to extend our FPT-results from
\Cref{ssec:alt} to \QBF{}s with an unbounded number of quantifier alternations.
More importantly, by
observing that many structural base classes are amenable to the  
deletion backdoor approach, this allows us to greatly generalize many of the
most important known tractable fragments for \QBF.
For instance, using enhanced backdoors we are able to extend the
prominent tractable fragment of QBF formulas with bounded number of quantifier
alternations and bounded incidence treewidth to a fragment that also
contains some QBF formulas with an arbitrary number of quantifier alternations
and unbounded incidence treewidth.
We complement this with novel FPT-algorithms for enhanced backdoor {\em detection} and thus prove that enhanced backdoors can both be applied and detected efficiently.

\subsubsection{Lower Bounds}
\label{ssec:disjunct}

We begin by considering the problems $\QBF(\CNF)$, in the following
denoted by $\QBF$, parameterized by a backdoor to $\CCC$, where
$\CCC\in \{\TCNF, \HORN,\AFF\}$. %, with the goal of proving lower bounds.
%\victor{I'll fix this later.}
Here, a likely explanation for the current absence of lower/upper bounds is that existing techniques in the literature simply fall short. They are either concentrated on the detection problem (since, as discussed, the backdoor evaluation problem is typically trivial for satisfiability and constraint satisfaction problems) or have been developed for problems structurally too far away from QBF (e.g., planning~\cite{Kronegger_Ordyniak_Pfandler_2014,Kronegger_Ordyniak_Pfandler_2015}, linear temporal logic~\cite{meier_et_al:LIPIcs.IPEC.2016.23}, and disjunctive logic programs~\cite{10.1145/2818646}). To prove our negative result we introduce a novel technique referred to as {\em squishing}.
In short, we first reduce the given $\QBF$ instance to a disjunctive problem by quantifier elimination of the backdoor variables, and the basic idea in the squishing technique then allows us to bound, or squish, the number of disjuncts (visualized as ``width'') while simultaneously increasing their size (visualized as ``height''). The details are by no means simple and we describe them in Section~\ref{sec:intro_squishing}.

\begin{theorem} 
\label{thm:HardCNF}
  If $\CCC$ is \TCNF, \HORN, or \AFF then
  $\QBF$ is para-PSpace-hard
  parameterized by the backdoor size to $\QBF(\CCC)$.  %Similarily, \QCSP is para-PSpace-hard parameterized by the deletion backdoor size to $\AFF$.
\end{theorem}

The above theorem implies that $\QBF$ is PSpace-hard even when the
input formula has a backdoor to  \TCNF, \HORN, or \AFF, of constant size. Hence, evaluation of backdoors is not only unlikely to be FPT in the QBF setting but is {\em extremely} far from being so, in the sense that the parameter does not affect the complexity classification at all.

\subsubsection{Backdoors and Number of Quantifier Alternations} \label{ssec:alt}

Because of \Cref{thm:HardCNF} we need to consider further
restrictions than just backdoor size in order to obtain
fixed-parameter tractability. 
A natural additional restriction is the number of quantifier
alternations, which has been considered for $\QBF$ before, e.g.,
when showing that $\QBF$ is
fixed-parameter tractable parameterized by incidence treewidth plus
the number of quantifier alternations~\cite{DBLP:conf/stacs/CapelliM19}.
We show that $\QBF$ is fixed-parameter tractable parameterized by the size
of a backdoor to $\QBF(\TCNF)$ and $\QBF(\AFF)$ plus the number
of quantifier alternations, but surprisingly  remains
W$[1]$\hy hard for the base class of $\QBF(3\hy\HORN)$ formulas.
It is well known that parameterizing $\QBF$ by the number of
quantifier alternations alone is not sufficient because even
restricting this number to a constant results
in complete problems for complexity classes in the polynomial
hierarchy (indeed, even zero quantifier alternations is still NP-hard).
\begin{restatable}{theorem}{thmintrotractcnf}\label{thm:introtractcnf}
  $\QBF$ is FPT parameterized by the number of quantifier alternations plus 
  backdoor size to $\QBF(\CCC)$ for every $\CCC \in
  \{\TCNF,\AFF\}$. However, $\QBF$ is W$[1]$\hy hard parameterized by
  the number of quantifier alternations plus backdoor size to $\QBF(3\hy\HORN)$.
\end{restatable}

See Table~\ref{tab:1} for a summary of our main results for (strong) backdoors.

\begin{table}[h] 
\caption{A summary of our main results for (strong) backdoor
  evaluation. The second column provides our results parameterized by
  only the size of the backdoor and the third column outlines our
  results when additionally parameterizing by the number of quantifier
  alternations.}
\label{tab:1}
\bigskip
\begin{center}
\begin{tabular}{l|lll}
        \textbf{Base Class} &  \textbf{Only size of the backdoor}  & \textbf{With quantifier alternations}  \\
         \hline \\
$\TCNF$         &   para-PSpace-hard
                              (\Cref{thm:PSpace-hard})
                           & FPT
                             (\Cref{the:tcnf-tract}) \\
($d \geq 3$)$\HORN$  &   para-PSpace-hard       (\Cref{thm:PSpace-hard}) & W[1]-hard       (\Cref{thm:bounded_arity_horn}) \\
$\HORN$   &   para-PSpace-hard       (\Cref{thm:PSpace-hard})  &W[1]-hard       (\Cref{thm:bounded_arity_horn}) \\
($d \geq 2$)-$\AFF$    & para-PSpace-hard (\Cref{cor:bounded_arity_aff})& FPT (\Cref{the:tract-aff}) \\   
$\AFF$    &   para-PSpace-hard
                          (\Cref{thm:PSpace-hard})  &
                                                                      FPT (\Cref{the:tract-aff})
\end{tabular}
\end{center}
\end{table}

\subsubsection{Elimination of Guarded Universal Sets and Enhanced Backdoors} \label{ssec:enh}

Here, we develop a very general solution method for $\QBF$, which we coin
\emph{elimination
of guarded universal sets}, 
that has many important algorithmic consequences for
solving $\QBF$ in general.
Let $\Phi$ be a $\QBF$ formula and let $Y \subseteq
V_\forall(\Phi)$ be a subset of all universal variables occurring in $\Phi$.
The \emph{primal graph}
$\primG(\Phi)$ of $\Phi$ is the undirected graph with a vertex for every variable in $\Phi$ and 
an edge between two variables if they occur in a common clause. 
We denote by $\delta_{\Phi}(Y)$, or simply $\delta(Y)$ if $\Phi$ is
clear from the context, the \emph{boundary (or guard)} of $Y$ for $\Phi$, i.e.,
the set of all neighbors of $Y$ in the primal graph. 
We also denote by
$\SA(\Phi,Y)$ the maximum number of variables of $Y$ in any clause $c$
of $\Phi$, i.e., $\max_{c \in \Phi}|V(c)\cap Y|$, where $V(c)$ denotes
the set of variables in clause $c$.
Here, and in the sequel, for two QBF formulas $\Phi$ and $\Phi'$ we write $\Phi \Leftrightarrow \Phi'$ if $\Phi$ and $\Phi'$ are {\em equivalent} in the sense that $\Phi$ is true if and only if $\Phi'$ is true.
The
following theorem is the technical heart of our elimination of guarded
universal sets method. 
\begin{restatable}[Elimination of Guarded Universal
  Sets]{theorem}{thnewunmixedqbf}\label{th:new_unmixed_qbf}
  Let $\Phi$ be a \QBF formula and let $Y \subseteq
  V_{\forall}(\Phi)$. Then, there is an FPT-algorithm for parameter $|\delta(Y)|+\SA(\Phi,Y)$
  that computes a partial assignment $\beta$ of $Y$
  such that 
  $\Phi \Leftrightarrow \Phi[\beta]$ and
  $|Y \setminus V(\beta)|$ is bounded by a function of $|\delta(Y)| +
  \SA(\Phi,Y)$, where $V(\beta)$ is the set of variables assigned
  a value by $\beta$.
\end{restatable}
Informally, \Cref{th:new_unmixed_qbf} allows us to eliminate an arbitrarily large set $Y$ of universal variables as
long as its boundary $\delta(Y)$ and arity $\SA(\Phi,Y)$ are not too
large; the remaining (bounded by a function in the parameter) variables in $Y\setminus V(\beta)$ can then be eliminated
using standard quantifier elimination or similar methods. This result is clearly already interesting in its own
right as it allows one to greatly simplify $\QBF$ formulas, e.g., it
implies that $\QBF$ can be reduced to the propositional satisfiability problem (and
therefore its complexity drops from \textsc{PSpace} to NP)
in FPT-time parameterized by the size of a
smallest separator between existential and universal variables in
the primal graph of $\Phi$ plus the maximum number of universal
variables in any clause of $\Phi$.
More importantly, 
we show that the result can be used to significantly enhance the
applicability of (deletion) backdoors.
To this aim, say that a class $\CCC$ of
$\QBF$ formulas is \emph{(deletion) extendable} if deciding the validity of
any $\QBF$ formula is
fixed-parameter tractable parameterized by the size of a given (deletion)
backdoor $B$ to $\CCC$. Via Theorem~\ref{thm:introtractcnf} it easily follows that \TCNF and \AFF with bounded quantifier alternations can be extended, but we can also effortlessly extend other tractable fragments from the literature.
\newcommand{\VC}{\mathcal{V}}

\begin{restatable}{theorem}{theextendableclasses}\label{the:extendableclasses}
  The following classes of $\QBF$ formulas are extendable:
  \begin{enumerate}
  \item\label{DCo} the class $\DDD_q(\CCC)$ of $\QBF(\CCC)$ formulas with at most $q$ quantifier
    alternations for $\CCC \in \{\TCNF,\AFF\}$,
  \item\label{DCt} the class $\DDD_{\exists}(\HORN)$ of $\QBF(\HORN)$ formulas
    having only existential variables.
  \end{enumerate}
  The following classes of $\QBF$ formulas are deletion extendable:
  \begin{enumerate}
  \setcounter{enumi}{2}
  \item\label{DCh} the class $\TTT_q^\omega$ of $\QBF$ formulas with at most $q$ quantifier
    alternations and incidence treewidth at most $\omega$,
  \item\label{DCf} the class $\VC_s$ of $\QBF$ formulas with incidence vertex
    cover number at most
    $s$.
  \end{enumerate}
\end{restatable}
We now show how to use \Cref{th:new_unmixed_qbf} to significantly enhance the applicability of
(deletion) backdoors to extendable classes. For a subset 
$B \subseteq V(\Phi)$ of the variables of a $\QBF$
formula $\Phi$, we denote by $\RU(\Phi,B)$ the union of all
components in $\primG(\Phi)-B$ that only contain universal variables.
For a class $\CCC$ of $\QBF$
formulas, we say that $B \subseteq V(\Phi)$ is an \emph{enhanced
  (deletion) backdoor} to $\CCC$ for $\Phi$ if $B\cup
\RU(\Phi,B)$ is a (deletion) backdoor to $\CCC$ for $\Phi$. 
\begin{restatable}{theorem}{theextdel}\label{the:extdel}
  Let $\CCC$ be a (deletion) extendable class, $\Phi$ be a $\QBF$ formula, and
  let $B \subseteq V(\Phi)$ be an enhanced (deletion) backdoor to $\CCC$ for $\Phi$. Then, deciding the validity of $\Phi$ is
  fixed-parameter tractable parameterized by
  $|B|+\SA(\Phi,\RU(\Phi,B))$.
  %$|B|+\SA(\Phi,Y)$, where $Y=\RU(\Phi,B)$.
\end{restatable}
Therefore, \Cref{the:extdel} shows that enhanced (deletion) backdoors can be
efficiently evaluated for the same base classes for which (deletion) backdoor
sets can, e.g., the classes provided in
\Cref{the:extendableclasses}. Importantly, \Cref{the:enhresover} greatly enhances the applicability of the
(deletion) backdoor approach, since enhanced (deletion) backdoors can be arbitrarily
smaller than backdoors but still apply for the same base classes.
To actually make use of enhanced
(deletion) backdoors, one still needs to be able to efficiently detect
them and indeed we are able to show that this is the case for all
classes given in \Cref{the:extendableclasses} using a combination of
standard branching, important separators~\cite[Section 8]{DBLP:books/sp/CyganFKLMPPS15}, and 
the randomized contractions and unbreakability framework~\cite{chitnis2016designing,cygan2020randomized,LokshtanovR0Z18}. This allows us to obtain the following FPT-results for \QBF.
\begin{restatable}{theorem}{theenhresover}\label{the:enhresover}
  \QBF{} is fixed-parameter tractable parameterized by either:
  \begin{itemize}
  \item the minimum size of an enhanced backdoor to $\DDD_q(\CCC)$ for
    any $\CCC \in \{\TCNF,d\hy\AFF\}$,
  \item the minimum
    size of an enhanced backdoor to $\DDD_{\exists}(\HORN)$ plus the
    maximum number of universal variables occurring in any clause, or
  \item the minimum
    size of an enhanced deletion backdoor to $\CCC$ for any $\CCC \in
    \{\TTT_q^\omega,\VC_s\}$ plus the maximum number of universal
    variables occurring in any clause.
  \end{itemize}
\end{restatable}
Observe that~\Cref{the:enhresover} provides an immediate
generalization of the
prominent
fixed-parameter tractability result for number of quantifier
alternations and incidence treewidth to $\QBF$ formulas with an unbounded
number of quantifier alternations and unbounded incidence treewidth.

Finally, to illustrate the usefulness of \Cref{th:new_unmixed_qbf} for backdoors to the Schaefer
classes, we provide the following consequence of \Cref{the:enhresover}
that shows that unmixed versions of the Schaefer classes are 
amendable to the backdoor approach; a \QBF formula is \emph{unmixed}
if every constraint either only contains universal or only existential variables. This structural restriction is, to the best of our knowledge, novel, and is an interesting complement to our other positive results since (1) we do not need to assume a bounded quantifier prefix, and (2) it is applicable to bounded arity \HORN.
\begin{restatable}{theorem}{thmbackdoorunmixed}
  \label{thm:backdoorunmixed}
  Let $\CCC \in \{\TCNF, d\textit{-}\HORN, d\textit{-}\AFF\}$ for $d \geq 1$. Then, 
  $\QBF$ is FPT parameterized by the minimum size of a backdoor to 
  unmixed $\QBF(\CCC)$, but $\QBF$ is W$[1]$\hy hard parameterized by the
  backdoor size to unmixed $\QBF(\HORN)$. 
\end{restatable}

\subsection{Technical Highlights}

To pinpoint and to in some extend eliminate the complications resulting
from the backdoor variables, we start by transforming the given $\QBF$
formula $\Phi=\pref.\phi$ with a backdoor set of size $k$ to a base class $\CCC$ into a
$2^k$-\textsc{Disjunct}-$\QBF(\CCC)$ formula, i.e.,
a $\QBF$ formula of the form $\pref.\bigvee_{i \in [2^k]}\phi_i$ such
that every disjunct $\phi_i$ is a formula in
$\CCC$ with shared quantifier prefix $\pref$.
A similar result can be achieved with standard quantifier elimination
techniques but instead with a $2^{2^k}$-\textsc{Disjunct}-$\QBF(\CCC)$ formula as a result in the worst case. Even more severely, the latter does not necessarily preserve the
influence of partial assignments on both formulas which we need in our elimination of guarded universal sets method (introduced in \Cref{ssec:unmixed-overview}).
Hence, we develop a different transformation that results in \emph{equisatisfiable} formulas, i.e.,
formulas with the same quantifier prefix such that
their matrices have the same set of satisfying assignments, which is based on the following transformation of CNF
formulas $\phi$ for a given  set of variables $X$:
\[ 
  \disj(\phi, X) = \bigvee_{\tau : X \to \{0,1\}} \phi[\tau] \land U_\tau,
\]
where $U_\tau$ is a conjunction of
unit clauses defined as follows:
if $\tau(x) = 1$, then $U_\tau$ contains the clause
$(x)$, otherwise it contains $(\lnot x)$. The following lemma now
shows that $\phi$ and $\disj(\phi,X)$ are equisatisfiable.

\begin{restatable}{lemma}{lemequisat} \label{lem:equisat}
  Let $\phi$ be a conjunctive formula.
  Then, for every $X \subseteq V(\phi)$,
  the formulas $\phi$ and $\disj(\phi,X)$
  are equisatisfiable.
\end{restatable}

Our reduction, given in the next lemma, works for all
base classes that are closed under reduction by partial assignments (which
we assume implicitly since it holds for all classes considered in this paper) where 
we may freely add unit clauses while staying inside the class.
\begin{restatable}{lemma}{lemfrombackdoortodisjunct}\label{lem:from-backdoor-to-disjunct}
  Let $\CCC$ be a class of $\QBF$ formulas 
  that allows adding unit clauses.
  Let $\Phi$ be a $\QBF$ formula with a backdoor
  of size $k$ to $\CCC$. Then, in time $\bigoh(2^k||\Phi||)$ we can
  compute a $2^k$-\textsc{Disjunct}-$\QBF(\CCC)$ formula that is
  equisatisfiable with $\Phi$.
\end{restatable}
Interestingly, there is also a simple reduction in the reverse direction. 
\begin{restatable}{lemma}{lemfromdisjuncttobackdoor}\label{lem:from-disjunct-to-backdoor}
  Let $\CCC$ be a class of $\QBF$ formulas.
  There is a linear-time reduction that takes
  an instance of $k$-\textsc{Disjunct}-$\QBF(\CCC)$
  and outputs
  an equivalent instance of $\QBF$ with
  a backdoor of size $\lceil \log k\rceil$ to $\CCC$.
\end{restatable}
By the lemmas above, analyzing $\QBF$ formulas having a backdoor to a
class $\CCC$ of $\QBF$ formulas (with parameter backdoor size) turns
out to be equivalent (for all our intents and
purposes) to the analysis of 
$k$-\textsc{Disjunct}-$\QBF(\CCC)$ formulas (with parameter $k$)
and we will make strong use of
this equivalence for both our algorithmic lower bound and upper bound
results by focusing on the analysis of the structurally simpler
$k$-\textsc{Disjunct}-$\QBF(\CCC)$ formulas.
While this is a
simple connection, it is also crucial for our analysis of QBF formulas
having a backdoor (and $\QBF$ formulas in general) and often makes the analysis
much simpler. For instance, we will
often make use of trading the number of disjuncts for the number of
variables and/or the size of the disjuncts in a delicate manner that
would be much more difficult to do on the original formula.

\subsubsection{Squishing and Lower Bounds} \label{sec:intro_squishing}

We will show that evaluating QBF formulas with small backdoors to 
$\TCNF$, $\HORN$, and $\AFF$ is para-PSpace-hard. We start with an immediate observation that \textsc{Disjunct}-$\QBF(\NUMCNF{1})$ is hard because it is essentially $\QBF(\DNF)$.

\begin{proposition}\label[proposition]{prop:dnf-to-disjunct-intro}
  \textsc{Disjunct}-$\QBF(\NUMCNF{1})$ is PSpace-complete.
\end{proposition}

The formulas which are hard by \Cref{prop:dnf-to-disjunct-intro}
have many disjuncts, but each of them is of constant size.
The main idea of ``squishing'' is to reduce the number
of disjuncts to a constant at the expense of
increasing the sizes of disjuncts.
Imagine the number of quantifier alternations 
as the ``height'' of the instance and
the number of disjuncts as its ``width'' --
we are ``squishing'' the instance,
making it taller and thinner.
In the process, new variables are introduced,
and the quantifier depth increases. 

As an illustration, 
consider a $9$-\textsc{Disjunct}-$\QBF(\THORN)$
formula $\Phi$ with quantifier prefix $\cQ$.
Let us label the disjuncts of $\Phi$
by $\phi_{1,1}, \dots, \phi_{3,3}$
and view it as a $3$-by-$3$ matrix 
\begin{equation*}
  \begin{array}{cc>{\columncolor{green!20}}c}
    \phi_{1,1} & \phi_{1,2} & \phi_{1,3} \\
    \rowcolor{red!20}
    \phi_{2,1} & \phi_{2,2} & \cellcolor{yellow!40!red!30!green!30} \phi_{2,3} \\
    \phi_{3,1} & \phi_{3,2} & \phi_{3,3} \\
  \end{array}
\end{equation*}
We will now create an equivalent formula $\Phi'$ with $6$ disjuncts
$\psi^R_1,\psi^R_2,\psi^R_3$ and $\psi^C_1,\psi^C_2,\psi^C_3$.
Intuitively, $\psi^R_i$ corresponds to the $i$-th row
of the matrix above and $\psi^C_j$ corresponds to the $j$-th column.
The formulas will be constructed so that 
the existential player can satisfy $\phi_{i,j}$ in $\Phi$
if and only if they can satisfy $\psi^R_i$ and $\psi^C_j$ in $\Phi'$.
In the example above we illustrate this by highlighting
$\psi^R_2$ and $\psi^C_3$: if both are satisfied, then
$\phi_{2,3}$ (highlighted with both colors) is satisfied as well.

To achieve this, we ``split'' each formula $\phi_{i,j}$ into two ``parts''
$\phi^R_{i,j}$ and $\phi^C_{i,j}$.
We will use shorthand notation $(a = b)$ for $(a \to b) \land (b \to a)$.
Let us illustrate splitting on an arbitrary formula $\phi$ from $\THORN$.
Let $C_1,\dots,C_m$ be the clauses in $\phi$.
From every clause $C_k$ in $\phi$
choose an arbitrary literal, say $\ell_k$.
Introduce a new variable $y_k$.
Create $\phi^R$ from $\phi$
by replacing $\ell_k$ in clause $C_k$ with $y_k$.
Then let $\phi^C = \bigwedge_{k=1}^{m} (\ell_k = y_k)$.
Consider the quantified formula
\[ \cQ \exists y_1,\dots,y_m. \phi^R \land \phi^C \]
Note that the existential player can easily satisfy
$\phi^R$ or $\phi^C$
because the (existential) variables $y_1,\dots,y_\ell$
are at the end of the prefix and exactly one of them 
is present in every clause.
Moreover, if the existential player can satisfy both $\phi^R$ and $\phi^C$,
then he can reuse the same strategy to win on $\cQ. \phi$.
One can think of this step as ``reverse''
quantifier elimination on the last existential block.

Apply ``splitting'' to all disjuncts $\phi_{i,j}$ of $\Phi$,
and construct the disjuncts of $\Phi'$ as follows:
let $\psi^R_i$ contain all clauses 
of $\phi^R_{i,1},\dots,\phi^R_{i,3}$,
and 
let $\psi^C_j$ contain all clauses
of $\phi^C_{1,j},\dots,\phi^C_{3,j}$.
Let $Y$ be the set of all variables 
introduced to split $\phi_{i,j}$ for all $i,j$.
Consider
\[ 
  \cQ \exists Y. (\psi^R_{1} \lor \psi^R_{2} \lor \psi^R_{3}) \land
  (\psi^C_{1} \lor \psi^C_{2} \lor \psi^C_{3}).
\]
If the existential player has a strategy that satisfies
some row formula $\psi^R_i$ and 
some column formula $\psi^C_j$,
then they can satisfy both $\phi^R_{i,j}$ and $\phi^C_{i,j}$,
and thus $\phi_{i,j}$.
To get back to disjunctive form while forcing the existential player to satisfy
a row formula and a column formula
we introduce a new universal variable $w$,
and modify the formula as follows:
\[
  \Phi' = \cQ \exists Y \forall w. 
  (\psi^R_{1} \land w) \lor (\psi^R_{2} \land w) \lor (\psi^R_{3} \land w) \lor
  (\psi^C_{1} \land \lnot w) \lor (\psi^C_{2} \land \lnot w) \lor (\psi^C_{3} \land \lnot w).
\]
Since $w$ is universal and ultimate,
the only way for the existential player to win is to satisfy
one disjunct with $w$ and one disjunct with $\lnot w$,
i.e., a row formula and a column formula.

This idea naturally generalizes to reducing formulas 
with $k^2$ disjuncts to $2k$ disjuncts.
It is sufficient to apply it roughly $\log \log m$
times, where $m$ is the number of disjuncts in the original formula,
to reduce to a constant number of disjuncts. 
While the number of clauses in the formula doubles every time we squish, 
we need only $\bigoh(\log \log m)$ steps
to reduce from $m$ disjuncts to an arbitrary constant, 
so the size of the resulting instance is within
a polynomial factor of the original.
Using a more delicate approach %With further technical complications\todo{Maybe use another word;
we design an even faster ``squishing procedure'' that 
in a single iteration
reduces a formula with $\binom{k}{2^p}$ disjuncts,
where $p$ can be any integer with $2^p \leq k$,
to a formula with $k$ disjuncts,
while keeping a polynomial increase in instance size. 
We then use this to show
that any formula with $6 = \binom{4}{2}$ disjuncts (and therefore any formula)
can be reduced to an equivalent formula with only $4$ disjuncts. 
Note that the source of hardness is a disjunction of formulas in \NUMCNF{1},
and the only clauses introduced by the squishing procedure are
equalities, which can be expressed in \THORN.
Thus, the result of squishing is in \THORN, and we conclude the following.

\begin{theorem}[Follows also from \Cref{cor:bounded_arity_aff}]
  $4$-\textsc{Disjunct}-$\QBF(\THORN)$ is PSpace-hard.
\end{theorem}

Another reason for introducing the fast squishing procedure (reducing
from $\binom{k}{2^p}$ disjuncts to $k$ in one iteration) is that it will
later allow us to obtain a lower bound (using a combination of
the sparsification lemma of~\cite{impagliazzo_which_2001} and the
Exponential Time Hypothesis~\cite{impagliazzo2001complexity}) for \textsc{Disjunct}-$\THORN$ parameterized by 
the number of disjuncts and the number of quantifier alternations.

\subsubsection{Elimination of Guarded Universal Sets}
\label{ssec:unmixed-overview}

Here, we will provide an overview for the proof of
\Cref{th:new_unmixed_qbf} (restated below).

\thnewunmixedqbf*

We start by reformulating \Cref{th:new_unmixed_qbf} into a question
about \textsc{Disjunct}\hy$\QBF$. In particular, let $\Phi'$ be a
$\QBF$ formula and let $Y \subseteq V_{\forall}(\Phi')$. 
After applying
\Cref{lem:from-backdoor-to-disjunct} to $\Phi'$ for the set $\delta(Y)$ of variables, we
obtain an equisatisfiable $k$-\textsc{Disjunct}\hy$\QBF$ $\Phi$ on
the same variables as $\Phi'$ for some $k$ that is bounded
by a function of $|\delta(Y)|$ such that the set $Y$ is \emph{$\Phi$-closed},
i.e., $Y \subseteq V_\forall(\Phi)$ and every clause $c$ in $C(\Phi,Y)$
has all its variables in $Y$, where $C(\Phi,Y)$ is the set of 
clauses (i.e., a clause occurring in some disjunct of $\Phi$) that contains at least one variable
from $Y$.
Therefore, to
show \Cref{th:new_unmixed_qbf}, it is sufficient to show the following theorem.

\begin{restatable}{theorem}{thnewunmixedqbfdisj}
\label{th:new_unmixed_disj}
  Let $\Phi$ be a $k$-\textsc{Disjunct-QBF} and let $Y\subseteq
  V_\forall(\Phi)$ be $\Phi$-closed.
  Then, there is an FPT-algorithm for parameter $k+\ariti(\Phi,Y)$
  that computes a partial assignment $\beta$ of $Y$
  such that 
  $\Phi \Leftrightarrow \Phi[\beta]$ and
  $|Y \setminus V(\beta)|$ is bounded by a function of $k+\ariti(\Phi,Y)$.
\end{restatable}

We now illustrate the proof of \Cref{th:new_unmixed_disj} and start by
introducing the notion of redundancy.
Let $\Phi$ be a $k$-\textsc{Disjunct}-$\QBF(\CNF)$ formula.
We say that $\Phi' \subseteq$\footnote{Here, and in the following,
we will sometimes view a $k$-\textsc{Disjunct}-$\QBF(\CNF)$ formula
as a set of disjuncts and we also assume that the
quantifier prefix does not change after any set operation is
applied.}$\Phi$
is \emph{$\Phi$-redundant} if and
only if $\Phi \Leftrightarrow\Phi \setminus \Phi'$.
We say that a set $X$ of variables is \emph{$\Phi$-redundant} if and only if 
$\Phi \Leftrightarrow \pref'.\bigvee_{\phi \in \Phi} \phi \setminus C(\phi, X)$,
where $\pref'$ is the prefix of $\Phi$ after removing $X$.
The following lemma now shows that in order to obtain the assignment
$\beta$ from \Cref{th:new_unmixed_disj}, it is sufficient to find a subset $Y_{\mnu}\subseteq Y$ together with a subformula $\Phi_{\mnu}\subseteq
\Phi$ and an
assignment $\beta$ over $Y_{\mnu}$ such that $Y_{\mnu}$ is $\Phi\setminus \Phi_\mnu$-redundant
and $\beta$ falsifies every disjunct of $\Phi_\mnu$.
     
\begin{restatable}{lemma}{betaapplication}
\label{lem:beta_application}
     Let $\Phi$ be a $k$-\textsc{Disjunct-QBF} and let $Y \subseteq
     V_\forall(\Phi)$ be $\Phi$-closed. If $\beta$ is an assignment
     over $Y_{\mnu}\subseteq Y$ and $\Phi_{\mnu} \subseteq \Phi$ such that $Y_{\mnu}$ is $\Phi\setminus \Phi_\mnu$-redundant
     and $\beta$ falsifies every conjunctive formula from $\Phi_\mnu$,
     then $\Phi \Leftrightarrow \Phi[\beta]$.
\end{restatable}
\begin{proof}
    First, observe that $\Phi$ implies $\Phi[\beta]$,  
    since we explicitly fix an assignment of universal variables via $\beta$.  
    Moreover, $\Phi[\beta]$ and $(\Phi \setminus \Phi_{\mnu})[\beta]$  
    are equisatisfiable, because $\beta$ falsifies all formulas in $\Phi_{\mnu}$.  
    For every $\phi \in \Phi \setminus \Phi_{\mnu}$,  
    all constraints in $\phi \setminus C(\phi, V(\beta))$  
    are preserved in $\phi[\beta]$.  
    Hence, $(\Phi \setminus \Phi_{\mnu})[\beta]$ implies  
    $\pref'.\bigvee_{\phi \in \Phi \setminus \Phi_{\mnu}}  
    \phi \setminus C(\phi, Y_{\mnu})$,  
    where $\pref'$ is the prefix of $\Phi$ with $Y_{\mnu}$ removed.  
    Finally, since $Y_{\mnu}$ is $\Phi \setminus \Phi_{\mnu}$-redundant,  
    it follows that  
    $\pref'.\bigvee_{\phi \in \Phi \setminus \Phi_{\mnu}}  
    \phi \setminus C(\phi, Y_{\mnu})$ implies $\Phi$.  
    Altogether, we conclude that $\Phi \Leftrightarrow \Phi[\beta]$.
\end{proof}

Let $\Phi$ be a $k$-\textsc{Disjunct-QBF} and  
$Y \subseteq V_\forall(\Phi)$ be a $\Phi$-closed set of universal variables.
Let $\Phi_{\mnu}$ be the subset of $\Phi$ containing 
all disjuncts $\phi$ with $\ariti(\phi,Y)>1$, i.e., $\Phi_{\mnu} =
\SB\phi \SM \phi \in \Phi \land \ariti(\phi,Y)>1\SE$.
We first generate $Y_\mnu \subseteq Y$ from the formula $\Phi_1 = \Phi
\setminus \Phi_\mnu$ as follows.
Note that $\ariti(\Phi_1,Y)\leq1$ and any disjunct in $\Phi_1$ that contains at least one clause over a
variable in $Y$ can be falsified using a single variable from $Y$.  
Moreover, if a formula $\phi \in \Phi_1$ contains at least $k$  
unit clauses over $Y$, then the universal player can always  
falsify $\phi$. This is because the universal player can use any unit
clause in $C(\phi,Y)$ to falsify $\phi$ and the only reason not to use a
particular variable is to
use it to falsify another of the at most $k-1$ disjuncts in
$\Phi_1\setminus \{\phi\}$.
This leads to a useful bound: there are at most $k^2$  
universal variables in $Y$ that are not $\Phi_1$-redundant.  
The following lemma further improves on this bound,  
showing that only $k$ universal variables in $Y$  
need to be considered. 
\begin{restatable}{lemma}{lemunitclauses}\label{lem:unit_clause}
  Let $\Phi$ be a $k$-\textsc{Disjunct-QBF}
  and let $Y \subseteq V_\forall(\Phi)$ be $\Phi$-closed such that
  $\ariti(\Phi,Y) \leq 1$.
  Then, we can compute a set $Y_{\mnu} \subseteq Y$
  such that $|Y \setminus Y_{\mnu}| \leq |\Phi|\leq k$
  and $Y_{\mnu}$ is $\Phi$-redundant in time $\bigoh(k^3 + \sizeOf{\Phi})$.
\end{restatable}

Our next goal is to develop a greedy procedure on $\Phi$ that finds an assignment $\beta$
over $Y_{\mnu}$ (computed in \Cref{lem:unit_clause}) that falsifies as many of the disjuncts in $\Phi_{\mnu}$ as
possible. Before doing so, we need the following notion of a hitting
set of a formula.  
Let $\phi$ be a CNF formula and let $Y_\phi \subseteq Y$ be a set of variables.  
We say that $Y_\phi$ is a \emph{hitting set} of $\phi$ with respect to $Y$  
if every clause in $\phi$ that contains a variable from $Y$  
also contains at least one variable from $Y_\phi$,  
i.e., $C(\phi, Y_\phi) = C(\phi, Y)$.

The greedy procedure starts with the empty assignment $\beta'=\emptyset$ and as long as there is a disjunct $\phi \in
\Phi_{\mnu}$ that is not yet falsified by $\beta'$ and a 
clause $c$ in $C(\phi[\beta'],Y_{\mnu})$
such that $c$ only contains variables in $Y_{\mnu}$,
% does not contain a variable from $(Y\setminus Y_\mnu)\cup V(\beta')$,
it extends $\beta'$ with the unique assignment that falsifies
$c$. 
Then, either the greedy procedure succeeds by producing a
partial assignment $\beta'$ of $Y_{\mnu}$ that falsifies all disjuncts in $\Phi_{\mnu}$ or there is a 
disjunct $\phi \in \Phi_{\mnu}$ that is not yet falsified by $\beta'$,
such that $V(\beta')$ is a hitting set of
$\phi \setminus C(\phi, Y \setminus Y_\mnu)$ with respect to $Y_\mnu$.
In the former case, we can just obtain a (complete) assignment $\beta$
over $Y_\mnu$ by extending $\beta'$ arbitrarily.
In the latter case, we define $Y_\phi$ as $V(\beta') \cup (Y \setminus
Y_{\mnu})$, since $Y\setminus Y_\mnu$ is a hitting set of $\phi$ with
respect to $Y\setminus Y_\mnu$.
Moreover, since the greedy procedure assigns at most
$\ariti(\Phi,Y)$ new variables at every step, we obtain that $|V(\beta)| \leq |\Phi_\mnu|
\cdot \ariti(\Phi,Y)$ and $|Y_\phi| \leq |\Phi_\mnu| \cdot
\ariti(\Phi,Y) + |Y \setminus Y_{\mnu}| \leq |\Phi_\mnu| \cdot
\ariti(\Phi,Y) + |\Phi \setminus \Phi_{\mnu}| \leq k \cdot
\ariti(\Phi,Y)$. This allows us to obtain the following lemma.

\begin{restatable}{lemma}{lemgreedylong}\label{lem:greedy_long}
  Let $\Phi$ be a $k$-\textsc{Disjunct-QBF}
  and let $Y \subseteq V_\forall(\Phi)$ be $\Phi$-closed.
  Let $\Phi_\mnu \subseteq \Phi$ and $Y_{\mnu} \subseteq Y$, 
  such that $\ariti(\Phi \setminus \Phi_\mnu, Y) \leq 1$ and $|Y\setminus Y_{\mnu}| \leq |\Phi \setminus \Phi_\mnu|$.
  In $\bigoh(\sizeOf{\Phi})$ time
  we can find one of the following:
  \begin{enumerate}[(1)]
  \item an assignment $\beta$ over $Y_{\mnu}$ that falsifies all disjuncts in $\Phi_{\mnu}$ or
  \item a conjunctive formula $\phi \in \Phi_{\mnu}$ and a subset $Y_\phi \subseteq Y$ of size at most $k \cdot \ariti(\Phi,Y)$, 
    such that
    $Y_\phi$ is a hitting set of $\phi$ with respect to $Y$.
  \end{enumerate}
\end{restatable}

If the greedy algorithm succeeds to produce an
assignment $\beta$ that falsifies all disjuncts in $\Phi_{\mnu}$,
i.e. case (1) in \Cref{lem:greedy_long}, then we are done, since we can apply \Cref{lem:beta_application}. Otherwise, we will now show how to iteratively modify $\Phi$
(using the hitting set obtained in case (2)) until we end up in (1) case of~\Cref{lem:greedy_long}. To achieve
this we define formulas $\Phi^i$ iteratively by setting
$\Phi^0=\Phi$ and for every $i$, we obtain $\Phi^{i+1}$ from $\Phi^i$
as follows. 
If the greedy algorithm succeeds on $\Phi^i$, then we stop, i.e.,
we set $i_{\text{last}}=i$ and we obtain $\beta$.
Otherwise, we obtain $\phi \in\Phi^i$ and $Y_{\phi}$ from \Cref{lem:greedy_long} such that 
$Y_{\phi}$ is small and is a hitting set of $\phi$ w.r.t. $Y$.
Informally, $\Phi^{i+1}$ is then obtained from $\Phi^i$
by
replacing $\phi$
with a bounded number of new disjuncts 
each having a smaller arity, i.e., 
% $\Phi^i \Leftrightarrow \Phi^{i+1}$ and
$\ariti(\phi,Y)> \ariti(\Phi^{i+1}\setminus \Phi^i,Y)$.
We can achieve this using \Cref{lem:equisat} by replacing
$\phi$ with $\disj(\phi,Y_{\phi})$, i.e.,
$\Phi^{i+1} = \Phi^i \setminus\{\phi\} \cup \disj(\phi,Y_{\phi})$.
While this is sufficient to obtain an FPT-algorithm,  
we later present a different approach that improves  
both the runtime and space complexity of the algorithm.
In both cases we can show the following lemma.
% , which is
% the main challenge of this subsection.

\begin{lemma}\label{lem:new_limit}
  The number of steps $i_{\text{last}}$ for the procedure outlined above is bounded by
  a function of $k$ and $\SA(\Phi,Y)$.
\end{lemma}
To prove~\Cref{lem:new_limit},
we introduce (the progress measure) $f(\Phi^i)$, which is 
the $(d+1)$-dimensional
vector that provides the number of disjuncts of a given arity,
where $d =\ariti(\Phi,Y)$,
i.e.,
for every $0 \leq a \leq d$, $f(\Phi^i)[a] = |\SB \phi \SM \phi \in \Phi^i \land \ariti(\phi,Y)=a \SE|$.
Clearly, the sum of all values in $f(\Phi^i)$ is equal to $|\Phi^i|$.
Note that if $f(\Phi^i)$ contains zero on all indices larger than $1$, then $\Phi_\mnu^i = \emptyset$ and \Cref{lem:greedy_long} generates $\beta$, i.e., $i = i_{\text{last}}$.
Since
the size of $Y_\phi$ is smaller then $d|\Phi^i|$ for the formula
$\phi \in \Phi^i$ returned by \Cref{lem:greedy_long},
the size of $\Phi^{i+1}$ (and the sum of all values in $f(\Phi^{i+1})$) depends only on
$d$ and $|\Phi^i|$, i.e.,
$|\Phi^{i+1}| \leq |\Phi^i| + 2^{d|\Phi^i|}$ (or $|\Phi^{i+1}| \leq (2d+1)|\Phi^i|$ if we use the improved replacement).
Moreover,
$f(\Phi^{i+1})$ is lexicographically smaller than $f(\Phi^{i})$,
i.e.,
there exists $d' \in [d]$ such that $f(\Phi^{i+1})[d'] < f(\Phi^{i})[d']$ and for every $d' < a  \leq d$, it holds that $f(\Phi^i)[a] = f(\Phi^{i+1})[a]$.
Those properties are sufficient to conclude the existence of an upper bound for $i_{\text{last}}$ that depends only on $|\Phi|$ and $d$.
In the later section we calculate the exact number.

Thanks to \Cref{lem:new_limit}, we know that
the assignment $\beta$ is produced after at most $i_{\text{last}}$
steps, which is bounded by a function of $k$ and $\ariti(\Phi,Y)$.
To conclude the proof of \Cref{th:new_unmixed_disj},
observe that %the size of $\Phi^{i+1}$ grows at most by
$|\Phi^{i+1}| \leq (2d+1)|\Phi^i|$, i.e., we
add only a bounded number of disjuncts at every step,
and the number of variables in $Y$ not assigned by the partial
assignment $\beta$ is at most
$|Y \setminus V(\beta)| \leq |\Phi^{i_{\text{last}}}|$.

\subsubsection{Backdoors and Number of Quantifier Alternations for \TCNF}
\label{ssec:2cnf}

Here, we will provide an overview about the techniques
employed in \Cref{thm:introtractcnf} for the case of $\TCNF$. Using \Cref{lem:from-backdoor-to-disjunct}, it is
sufficient to show the following.
\begin{restatable}{theorem}{thetcnftract}\label{the:tcnf-tract}
  There is an algorithm that takes
  a $k$-\textsc{Disjunct-QBF}$(\TCNF)$ formula $\Phi$
  with $q$ quantifier alternations,
  and checks its validity in FPT-time
  parameterized by $k + q$.
\end{restatable}
Note that finding a minimum size
backdoor to \TCNF is well-known to be
FPT~\cite{DBLP:conf/birthday/GaspersS12} and it is therefore
sufficient to show evaluation.
Interestingly, the
algorithm works by eliminating the existential quantifier blocks iteratively
starting with the innermost one.
This is
fairly easy if the innermost quantifier block of the
formula is existential
(using a slight generalization of the well known $Q$-resolution from
CNF to disjuctive formulas formalized in \Cref{lem:prop_2cnf}), but becomes much more
challenging if the innermost quantifier block is universal. The main
technical challenge behind our approach is to show that the 
variables of the innermost existential quantifier block can be
replaced with a small set of variables, whose size is
bounded by a function of $k$. We can then
employ standard quantifier elimination to eliminate the remaining
variables (since it now only results in a bounded increase of the
number of disjuncts).

We begin by introducing the notion of propagation. We call a $\TCNF$ formula $\phi$ {\em propagated}
if it is closed under propositional resolution and does not contain any
superfluous clauses, i.e., clauses that contain other
clauses. Crucially, because
$\phi$ is a $\TCNF$ formula, we can compute an equivalent and
propagated \TCNF{} formula, denoted by $\cnfprop{\phi}$, in
polynomial-time by adding all resolvends and removing superflous
clauses (see also \Cref{lem:prop_comp}).
We say that a disjunctive formula $\Phi$ is
\emph{propagated} if so is every disjunct of $\Phi$. It is easy to see
that $(\pref. \bigvee_{\phi \in \Phi}\phi)$ and $(\pref. \bigvee_{\phi
  \in \Phi}\cnfprop{\phi})$ are equisatisfiable.

Let $\Phi$ be a $k$-\textsc{Disjunct-QBF}(\TCNF),  
where the innermost quantifier block $X_q$ is universal.
We will construct an equivalent 
$2^{2k^2 + k}k$-\textsc{Disjunct-QBF}(\TCNF) formula  
$\pies(\Phi)$, where the
variables $X_{q-1}$ of the innermost
existential quantifier block have become redundant and their function
is replaced by a set $A_\Phi$ of at most $2k^2+k$ fresh existential
variables (inserted into the position of $X_{q-1}$ within the
quantifier prefix). For illustrative purposes, it
will be convenient to think of $\pies(\Phi)$ (and all intermediate
formulas leading to $\pies(\Phi)$) as the disjunction of at most $2^{2k^2 + k}$ many  
$k$-\textsc{Disjunct-QBF}(\TCNF) formulas $\Phi_L$, where $L$ is a
non-contradictory subset of literals over the set
$A_\Phi=\{a_1,\dotsc,a_{2k^2+k}\}$ of fresh existential variables, i.e.,
$\pies(\Phi) \Leftrightarrow \bigvee_L\Phi_L$, where we use
$\bigvee_L\Phi_L$ for the formula $\pref.\bigvee_L (\bigvee_{\phi \in
  \Phi_L}\phi)$ assuming that $\pref$ is the common prefix all formulas $\Phi_L$.  

Intuitively, the set $A_\Phi$ is a set of fresh existential variables
that will later replace all variables in $X_{q-1}$ and the existential
player (in the Hintikka game)
uses the variables in $A_\Phi$ to
select which of the formulas $\Phi_L$ to play on by assigning values for the
variables in $A_\Phi$; importantly, every assignment of the variables
in $A_\Phi$ will falsify every $\Phi_L$ such that $L$ is falsified by
the assignment. Every formula $\Phi_L$ will have the following properties:
\begin{itemize}
\item[(A1)] $\Phi_L$ is propagated,
\item[(A2)] every two (possibly equal) disjuncts $\phi,\phi'\in \Phi_L$ 
  share the same set of clauses that do not contain variables from
  $X_q$, i.e., $\phi \setminus C(\phi, X_q) = \phi' \setminus
  C(\phi', X_q)$, and
\item[(A3)] for every variable $x \in X_{q-1}$, either the literal $x$
  or the literal $\neg x$
  does not occur in any clause (of any disjunct) of $\Phi_L$ with a variable in $X_q$.
\end{itemize}
Crucially (as we show in \Cref{lem:bqa-var-redundant}) it holds
that if each $\Phi_L$ satisfies (A1)--(A3), then $X_{q-1}$ is
$\pies(\Phi)$-redundant, which allows us to
safely ignore all variables in $X_{q-1}$, which have now been
replaced by the at most $2k^2+k$ variables in $A_\Phi$. Note that (thanks to
\Cref{lem:tcnf_pies_redundant}), we can apply
\Cref{lem:bqa-var-redundant} already if (A3) only applies to all
disjuncts that do not contain $2k$ unit clauses with variables in
$X_{q}$, i.e., we can replace (A3) with the following weaker condition:
\begin{itemize}
\item[(A3')] The subset $\Phi'$ of $\Phi_L$ containing all disjuncts
  that contain fewer than $2k$ unit clauses with variables in $X_q$ satisfies:
  for every $x \in X_{q-1}$, either the literal $x$ or the literal $\neg x$
  does not occur in a clause with a variable in $X_q$.
\end{itemize}
Intuitively, this holds because all formulas that contain $2k$ unit
clauses with variables in $X_q$ can be safely ignored since they can
always be falsified by the universal player (see also
\Cref{lem:tcnf_pies_redundant}). Note that replacing (A3) with (A3')
is crucial to show that $\pies(\Phi)$ can be obtained from $\Phi$ in a bounded
number of steps.

The main
ingredient of \Cref{lem:bqa-var-redundant} is to show that if a
formula $\Phi_L$ satisfies (A1)--(A3), then $X_{q-1}$ is
$\Phi_{L}$-redundant with respect to $\Phi^i$. Note that
$\Phi^i$ and 
$(\Phi^i \setminus \Phi_L) \cup \Phi_L'$ are equivalent, 
where $\Phi_L'$ is obtained from $\Phi_L$ after removing all clauses
containing a variable in $X_{q-1}$.
Informally, to show this consider an existential variable $x$ from $X_{q-1}$.  
Clearly, if $x$ appears only as the literal $\neg x$  
in all clauses of $\Phi_L$, then $x$ is $\Phi_L$-redundant,  
because the existential player may simply assign $x$ to $0$  
to satisfy all such clauses. Now the crucial idea of
\Cref{lem:bqa-var-redundant}
is that $x$ remains $\Phi_L$-redundant even if $x$ merely satisfies (A3), i.e.,
if $\neg x$ is the only literal of $x$ appearing in any clause in
$C(\Phi_L,X_q)$. Towards arguing this, consider the case that $x$
satisfies (A3), say for the literal $\neg x$, and there is a clause $(l \lor x)$ for some literal $l$
of a variable in $V(\Phi_L) \setminus X_q$ and assume for simplicity
that there is only one such clause.
We can assume that the variable of $l$ is before $x$ in the prefix,
because we can always reorder elements inside one quantifier block, i.e., $X_{q-1}$. 
When, during the Hintikka game on $\Phi_L$, the existential player has  
to decide on the assignment of $x$,  
it is sufficient to assign $x$ complementary to $l$, i.e., we claim
that if all clauses containing $l$ are satisfied, then all clauses
containing $x$ are satisfied in this way.
This is because every formula in $\Phi_L$ contains  
the clause $(l \lor x)$ (A2) and all formulas are propagated (A1).  
For instance, let $\phi \in \Phi_L$ and $z \in V(\Phi_L)$ 
such that $(l \lor x), (\neg x \lor z) \in \phi$.  
Then, by propagation, we obtain $(l \lor z) \in \phi$.  
If an assignment $\beta$ satisfies $(l \lor z)$  
and assigns $x$ complementary to $l$,  
then $\beta$ also satisfies $(\neg x \lor z)$.  
More generally, suppose that instead of a single literal $l$,  
we have a set $L'$ of literals from $V(\Phi_L) \setminus X_q$,  
such that only literals from $L'$ appear in clauses with $x$.
Then, $x$ can be assigned to $0$  
if all literals in $L'$ are satisfied by the current assignment,  
and $1$ otherwise. This presents an informal proof of
\Cref{lem:bqa-var-redundant}, i.e., the statement that if $x$ satisfies (A3), then it is
$\Phi_L$-redundant.

Our goal is now to build such a disjunction  
of formulas $\Phi_L$ that satisfy (A1), (A2), and (A3').  
We begin by introducing the function $\del(\Phi)$,  
which can be seen as the initialization step of the algorithm.  
Intuitively, $\del(\Phi)$ produces an equivalent formula  
consisting of the disjunction of $2^k$  
distinct
$k$-\textsc{Disjunct-QBF}(\TCNF) formulas  
that satisfy properties (A1) and (A2).
We then create the formula $\pies(\Phi)$, which
satisfies property (A3'), as well as (A1), and (A2), iteratively starting from $\del(\Phi)$.

Let $\Phi = \pref.\bigvee_{i \in [k]} \phi_i$ be a 
$k$-\textsc{Disjunct-QBF}(\TCNF{}) with 
$\pref = Q_1 X_1 \dots Q_q X_q$, and let 
$A_\Phi = \{a_i : i \in [\Asize]\}$ be a set of fresh (existential) variables.
We define the function $\del(\Phi)$, which intuitively acts 
as a selection mechanism that forces the existential player 
to explicitly choose a subset of disjuncts (by choosing an assignment of $A_\Phi$) to continue the game on.
For every non-empty set $I \subseteq [k]$, let 
$\phi^{\del}_I$ denote the conjunction of all clauses from 
$\bigwedge_{i \in I} \phi_i$
that do not contain variables from 
$X_q$, i.e.,
$
\phi^{\del}_I = \bigwedge_{i \in I} 
\left( \phi_i \setminus C(\phi_i, X_q) \right)
$.
Then, the formula $\del(\Phi)$ is given as:
\begin{equation}\label{en:del_def}
\del(\Phi) =  Q_1 X_1 \dots\exists (X_{q-1} \cup A_\Phi)\forall X_{q}.\bigvee_{\emptyset\neq I \subseteq [k]}
\Phi^\del_I,
\end{equation}
 where

$$
\Phi^\del_I =  \bigvee_{i \in I} 
\cnfpropname{}
\biggl( 
\phi_i \land \phi^{\del}_I
\land (\bigwedge_{j \in I} a_j) \land (\bigwedge_{j \in [k] \setminus I} \neg a_j)\biggr)
$$

\begin{example}\label{ex:del}
  Consider the following $3$-\textsc{Disjunct-QBF}(\TCNF):
  $\Phi =
  \forall_{x_1,x_2,x_3}\exists_{y_1,y_2,y_3}\forall_{z_1,z_2,z_3}. \phi_1
  \lor \phi_2 \lor \phi_3$, where:
  $\phi_1 = (x_2 \rightarrow x_1) \land (y_2 \rightarrow x_3)
  \land (y_2 \rightarrow z_2) \land (y_3 \rightarrow z_1)$,
  $\phi_2 = (x_2) \land (x_3 \rightarrow y_1) \land (y_2 \rightarrow
  y_3) \land (y_3 \rightarrow z_3)$, and
  $\phi_3 = (x_3 \rightarrow x_2) \land (y_1 \rightarrow z_3)
  \land (z_2 \rightarrow y_2)$, which together with the formulas
  $\phi_{\{1,2,3\}}^\del$ and
  $\cnfprop{\phi_{\{1,2,3\}}^\del\land \phi_3}$ are also illustrated in Figure~\ref{fig:del}.

  Then, 
  the formula $\phi_{\{1,2,3\}}^\del$ is derived from $\phi_1$, $\phi_2$, and $\phi_3$ 
  by keeping only the clauses that do not involve variables from the
  final $\forall$-block, i.e.:
  \[\phi_{\{1,2,3\}}^\del=(x_2 \rightarrow x_1) \land (y_2
    \rightarrow x_3)\land (x_2) \land (x_3 \rightarrow y_1) \land (y_2 \rightarrow
    y_3)\land (x_3 \rightarrow x_2).
  \]
  Moreover, the formula $\cnfprop{\phi_{\{1,2,3\}}^\del}$ and  $\cnfprop{\phi_{\{1,2,3\}}^\del\land \phi_3}$
  are given as
  \[\cnfprop{\phi_{\{1,2,3\}}^\del}=(x_1)\land(x_2)\land(x_3\rightarrow
    y_1)\land(y_2\rightarrow y_3)\land(y_2\rightarrow
    y_1)\land(y_2\rightarrow x_3),\]
    and
\begin{align*}
\cnfprop{\phi_{\{1,2,3\}}^\del\land \phi_3} = \cnfprop{\phi_{\{1,2,3\}}^\del} 
  &\land (x_3\rightarrow z_3)\land (y_1\rightarrow z_3)\land(y_2\rightarrow z_3)\land (z_2\rightarrow z_3) \\
  &\land (z_2\rightarrow y_3)\land(z_2\rightarrow y_2)\land(z_2\rightarrow y_1)\land(z_2\rightarrow x_3)
\end{align*}
\end{example}

\begin{figure}[htb]
    \centering
\begin{tikzpicture}[node distance=1.5cm]
    % Define color styles
    \tikzset{
        cred/.style={red},
        cblue/.style={blue},
        cgreen/.style={green}
    }
    \draw[gray, dashed] (2.5,0.8) -- (2.5,-9);
    \draw[gray, dashed] (5.5,0.8) -- (5.5,-9);
    \node[gray] at (1, 0.7) {$\forall$}; 
    \node[gray] at (4, 0.7) {$\exists$}; 
    \node[gray] at (7, 0.7) {$\forall$}; 
    \begin{scope}[shift={(0,0)}]
        \inittikz
        \drgiht{F}{G}{cred}
        \dleft{E}{C}{cred}
        \dleft{B}{A}{cred}
        \drgiht{E}{H}{cred}
        \node[anchor=east] at (-0.5,0) {$\phi_1:$};
    \end{scope}
    \begin{scope}[shift={(0,-1.5)}]
        \inittikz
        \drgiht{E}{F}{cblue}
        \dleft{I}{F}{cblue}
        \dposlit{B}{cblue}
        \drgiht{C}{D}{cblue}
        \node[anchor=east] at (-0.5,0) {$\phi_2:$};
    \end{scope}
    \begin{scope}[shift={(0,-3)}]
        \inittikz
        \dleft{H}{E}{cgreen}
        \drgiht{D}{I}{cgreen}
        \dleft{C}{B}{cgreen}
        \node[anchor=east] at (-0.5,0) {$\phi_3:$};
    \end{scope}
    \begin{scope}[shift={(0,-4.5)}]
        \inittikz
        \drgiht{E}{F}{cblue}
        \dleft{E}{C}{cred}
        \dposlit{B}{cblue}
        \dleft{B}{A}{cred}
        \dleft{C}{B}{cgreen}
        \drgiht{C}{D}{cblue}
        \node[anchor=east] at (-0.5,0) {$\phi^{\del}_{\{1,2,3\}}:$};
    \end{scope}
    \begin{scope}[shift={(0,-6)}]
        \inittikz
        \dleft{H}{E}{cgreen}
        \drgiht{D}{I}{cgreen}
        \dleft{C}{B}{cgreen}
        \drgiht{E}{F}{black}
        \dleft{E}{C}{black}
        \dposlit{B}{black}
        \dleft{B}{A}{black}
        \dleft{C}{B}{black}
        \drgiht{C}{D}{black}
        \node[anchor=east] at (-0.5,0) {$\phi^{\del}_{\{1,2,3\}} \land \phi_3:$};
    \end{scope}
    \begin{scope}[shift={(0,-8)}]
        \inittikz
        \dleft{H}{E}{cgreen}
        \dleft{H}{F}{cgreen}
        \drgiht{D}{I}{cgreen}
        \drgiht{E}{I}{cgreen}
        \drgiht{H}{I}{cgreen}
        \drgiht{C}{I}{cgreen}
        \drgiht{E}{F}{black}
        \dleft{E}{C}{black}
        \dleft{E}{D}{black}
        \dposlit{B}{black}
        \dposlit{A}{black}
        \drgiht{C}{D}{black}
        \node[anchor=east] at (-0.5,0) {$\cnfprop{\phi^{\del}_{\{1,2,3\}} \land \phi_3}:$};
    \end{scope}
\end{tikzpicture}
\caption[Illustration of the construction of $\phi^\del_I$]{
      The figure illustrates the construction of $\del(\Phi)$ presented in \Cref{en:del_def}.
In each row, the variables appear in the order given by the quantifier prefix of $\Phi$.
Dashed lines separate variables from different quantifier blocks.  
Each arrow represents an implication between positive literals,  
while each plus sign ($+$) or minus sign ($-$) next to a variable indicates that the formula contains a unit clause with the corresponding positive or negative literal.  
      The highlighted black elements in the final two segments represents
      the influence of $\phi_{\{1,2,3\}}^\del$.
    }
    \label{fig:del}
\end{figure}
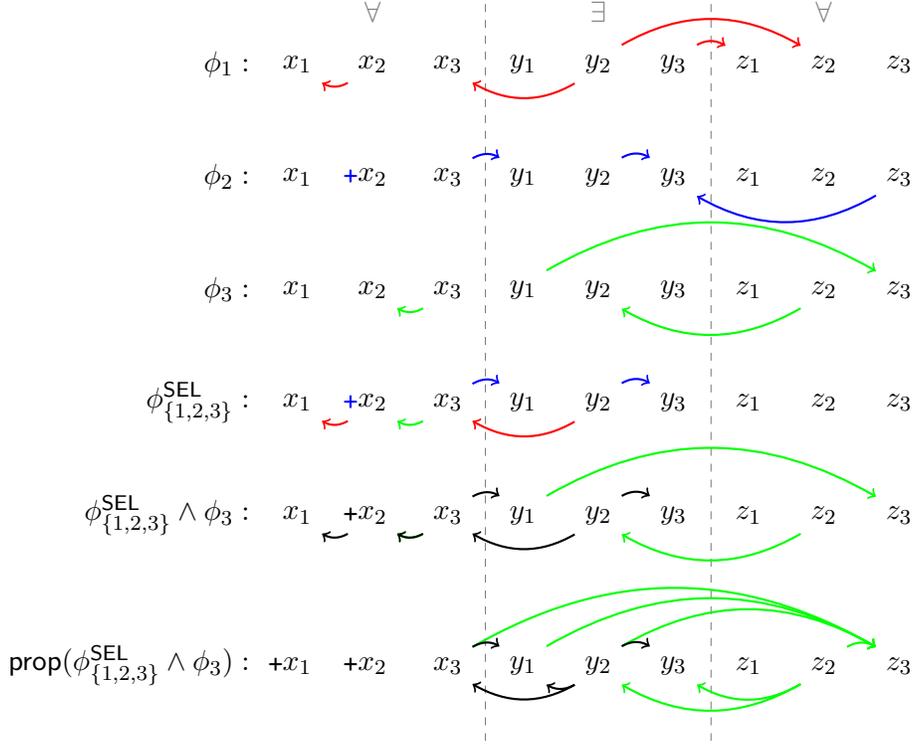

Note that $\del(\Phi)$ does not contain any clauses with variables
$a_i$ for $i > k$, which will later be used to obtain $\pies(\Phi)$.

The following corollary (shown in
\Cref{lem:bqa-dog-prop,lem:bqa-dog-eq})
establishes the equivalence between $\del(\Phi)$ 
and $\Phi$, along with properties that corresponds to (A1) and (A2).
The equivalence holds because any existential winning
strategy for $\del(\Phi)$ is also an existential winning strategy for
$\Phi$. Moreover, given an existential winning strategy for $\Phi$, we
obtain an existential winning strategy for $\del(\Phi)$ by assigning
$a_i$ to $1$ for every $i$ such that $\phi_i$ is not yet falsified;
here we assume without loss of generality that the variables in
$X_{q-1}$ are before the variables in $A_\Phi$ in the quantifier prefix.

\begin{corollary}\label{cor:bqa-dog}
  Let $\Phi$ be a propagated $k$-\textsc{Disjunct-QBF}(\TCNF), whose
  innermost quantifier block is universal.
  Then, $\Phi \Leftrightarrow \del(\Phi)$, $\del(\Phi)$ is a
  $k2^{k}$-\textsc{Disjunct-QBF}(\TCNF) that satisfies (A1) and (A2)
  and that can be
  computed in $\bigoh(k2^{k}\propT{|V(\Phi)|})$ time,
  where $\omega < 3$ is the matrix multiplication exponent.
\end{corollary}

We now define the operation $\pies(\Phi)$ that 
achieves (A3') for every $\Phi_L \in\pies(\Phi)$.
We start by setting $\Phi^0 = \del(\Phi)$ and obtain $\pies(\Phi)$ as
the fixed-point of formulas $\Phi^{i}$ that are iteratively obtained
as follows. % with $\pref$ being the prefix of $\del(\Phi)$. 
Note that
for every set $L$ of literals (that contains exactly one literal of
every variable in $\{a_1,\dotsc,a_k\}$), it holds that
$\Phi_L=\Phi^{\del}_{I(L)}$ and $\del(\Phi)=\pref.\bigvee_L\Phi_L$, where $I(L)$ is the subset
of $[k]$ containing $i \in [k]$ if and only if $a_i \in L$.

We say that a variable $x \in X_{q-1}$ is \emph{reducible} with
respect to $\Phi_L$, if there are two
(possibly equal) formulas $\phi$ and $\phi'$ in $\Phi_L$ such that
both of them contain less than $2k$ unit clauses from $X_q$
and
$x$ appears positively in a clause of $\phi$  
and negatively in a clause of $\phi'$ that includes variables from $X_q$.
% Note that, $\Phi_L$ does not satisfy (A3') if and only if there exists $x\in X_q$, such that $x$ is reducible.
Now for every $i\geq 0$, we obtain $\Phi^{i+1}$ from $\Phi^i$ iteratively as follows.
If $\Phi^i$ contains a $k$-\textsc{Disjunct-QBF}(\TCNF) $\Phi_L^i$ that does
not satisfy (A3'),
then $\Phi_L^i$ contains a variable $x \in X_{q-1}$ that is reducible.
Informally, we then obtain $\Phi^{i+1}$ by replacing $\Phi_L^i$ in $\Phi^i$ 
with two modified copies of $\Phi_L^i$, i.e., $\Phi_{L\cup\{a_{|L|+1}\}}'$ and
$\Phi_{L\cup\{\neg a_{|L|+1}\}}'$, which correspond to the cases where $x$ is
assigned to $1$ or to $0$ in $\Phi_L^i$, respectively. More formally:
\begin{equation}\label{eq:pies_update_short}
\Phi^{i+1}  = \pref.
\bigl(
\bigvee_{\phi \in (\Phi^i\setminus \Phi_L^i)} \phi 
\bigr)
\lor \Phi_{L\cup\{a_{|L|+1}\}}' \lor \Phi_{L\cup\{\neg a_{|L|+1}\}}'
\end{equation}
Here, $\Phi_{L\cup\{a_{|L|+1}\}}'=\bigvee_{\phi \in
  \Phi_L^i}(\cnfprop{\phi[x=1]}  \land x \land a_{|L|+1})$ and
$\Phi_{L\cup\{\neg a_{|L|+1}\}}'=\bigvee_{\phi \in
  \Phi_L^i}(\cnfprop{\phi[x=0]}  \land \neg x \land \neg a_{|L|+1})$.
\begin{example}\label{ex:piesup}
Let $\Phi$ be the $3$-\textsc{Disjunct-QBF}(\TCNF) from \Cref{ex:del}
and let $\del(\Phi) =\pref. \bigvee_L\Phi_L$. 
Then, $\Phi_{\{a_1, a_2, a_3\}}$ contains the propagation of the following three formulas:
$\phi_1' = \phi'' \land \crr{(}y_2 \crr{\rightarrow} z_1\crr{)}  \land \crr{(}y_2 \crr{\rightarrow} z_2\crr{)} \land \crr{(}y_3 \crr{\rightarrow} z_1\crr{)}$,
$\phi_2' = \phi'' \land \cb{(}y_2 \cb{\rightarrow} z_3\cb{)}$, and
$\phi_3' = \phi'' \land \cg{(}y_1\cg{\rightarrow}
z_3\cg{)}\land\cg{(}y_2\cg{\rightarrow} z_3\cg{)}\land\cg{(}z_2\cg{\rightarrow}
z_3\cg{)}\land\cg{(}z_2\cg{\rightarrow} y_3\cg{)}\land\cg{(}z_2\cg{\rightarrow}
y_2\cg{)}\land\cg{(}x_3\cg{\rightarrow} z_3\cg{)}\land\cg{(}z_2\cg{\rightarrow}
y_1\cg{)}\land\cg{(}z_2\cg{\rightarrow} x_3\cg{)}$,
where $\phi'' = \cnfprop{\phi_{\{1,2,3\}}^{\del}} \land (a_1)\land (a_2)\land (a_3)$.
%\todo{MR:new colors}

The formulas $\Phi_{\{a_1, a_2, a_3, a_4\}}$ and $\Phi_{\{a_1, a_2, a_3, \neg a_4\}}$  
are derived from $\Phi_{\{a_1, a_2, a_3\}}$ by applying the update rule in \eqref{eq:pies_update_short} to the variable $y_2$.  
These formulas are illustrated in Figure~\ref{fig:piesup} and given below:
$$\Phi_{\{a_1,a_2,a_3,a_4\}} = x_1\land x_2 \land x_3 \land a_1\land a_2 \land a_3 \land a_4 \land y_1 \land \cy{(}y_2\cy{)} \land y_3 \land \big((\crr{(}z_1\crr{)} \land \crr{(}z_2\crr{)}) \lor \cg{(}z_3\cg{)}\big)$$
$$    \Phi_{\{a_1,a_2,a_3,\neg a_4\}} = x_1\land x_2 \land (x_3 \rightarrow y_1)\land a_1\land a_2 \land a_3 \land \neg a_4 \land \cy{(}\neg y_2\cy{)} \land 
\big(
\crr{(}y_3 \crr{\rightarrow} z_1\crr{)} \lor
\cb{(}z_3 \cb{\rightarrow} y_3\cb{)} \lor
(\cg{(}y_1\cg{\rightarrow} z_3\cg{)} \land \cg{(}\neg z_2\cg{)}
\big)$$
\end{example}
\begin{figure}
\label{fig:piesup}
\centering
\begin{tikzpicture}[node distance=1.5cm]
    % Define color styles
    \tikzset{
        cred/.style={red},
        cblue/.style={blue},
        cgreen/.style={green}
    }
    % \draw[gray, dashed] (0.5, 1) -- (0.5,-4);
    % \draw[gray, dashed] (7.5, 1) -- (7.5,-4);
    % \node[gray] at (0, 0.9) {$\forall$}; 
    % \node[gray] at (4, 0.9) {$\exists$}; 
    % \node[gray] at (9, 0.9) {$\forall$}; 
    \draw[gray, dashed] (2.5, 1.5) -- (2.5,-4);
    \draw[gray, dashed] (9.5, 1.5) -- (9.5,-4);
    \node[gray] at (1, 1.4) {$\forall$}; 
    \node[gray] at (6, 1.4) {$\exists$}; 
    \node[gray] at (11, 1.4) {$\forall$}; 
    \begin{scope}[shift={(0,0)}]
        \inittikzIV
        
        \draw[->, thick, green]   (C.north east) to[out=30, in=145, looseness=0.6] (I.north west);
        \draw[->, thick, black]   (C.north east) to[out=30, in=150, looseness=0.6] (D.north west);
        \draw[->, thick, black]   (E.south west) to[out=210, in=330, looseness=0.6] (C.south east);
        \dposlit{A}{black}
        \dposlit{B}{black}
        
        \dposlit{A1}{black}
        \dposlit{A2}{black}
        \dposlit{A3}{black}
        \drgiht{F}{G}{cred}
        \drgiht{E}{H}{cred}
        \drgiht{E}{G}{cred}
        \drgiht{E}{I}{cgreen}
        \dleft{H}{E}{cgreen}
        \dleft{H}{F}{cgreen}
        \drgiht{D}{I}{cgreen}
        \drgiht{H}{I}{cgreen}
        \dleft{I}{F}{cblue}
        \drgiht{E}{F}{black}
        \dleft{E}{D}{black}
        \node[anchor=east] at (-0.5,0) {$\Phi_{\{a_1,a_2,a_3\}}:$};
    \end{scope}
    \begin{scope}[shift={(0,-1.5)}]
        \inittikzIV
        % \draw[->, thick, black]   (C.north east) to[out=30, in=150, looseness=0.6] (D.north west);
        % \draw[->, thick, black]   (E.south west) to[out=210, in=330, looseness=0.6] (C.south east);
        \dposlit{A}{black}
        \dposlit{B}{black}
        \dposlit{C}{black}
        
        \dposlit{A4}{black}
        \dposlit{A1}{black}
        \dposlit{A2}{black}
        \dposlit{A3}{black}
        \dposlit{E}{black}
        \dposlit{D}{black}
        \dposlit{F}{black}
        \dposlit{G}{cred}
        \dposlit{H}{cred}
        \dposlit{I}{cgreen}
        \node[anchor=east] at (-0.5,0) {$\Phi_{\{a_1,a_2,a_3,a_4\}}:$};
    \end{scope}
    \begin{scope}[shift={(0,-3)}]
        \inittikzIV
        \draw[->, thick, black]   (C.north east) to[out=30, in=150, looseness=0.6] (D.north west);
        \dposlit{A}{black}
        \dposlit{B}{black}
        \draw[->, thick, green]   (C.north east) to[out=30, in=145, looseness=0.6] (I.north west);
        
        \dneglit{A4}{black}
        \dposlit{A1}{black}
        \dposlit{A2}{black}
        \dposlit{A3}{black}
        \dneglit{E}{black}
        \drgiht{F}{G}{cred}
        \dneglit{H}{cgreen}
        \drgiht{D}{I}{cgreen}
        \dleft{I}{F}{cblue}
        \node[anchor=east] at (-0.5,0) {$\Phi_{\{a_1,a_2,a_3,\neg a_4\}}:$};
    \end{scope}
\end{tikzpicture}
\caption{
An illustration of the formulas given in \Cref{ex:piesup}. Every row corresponds to a $3$-\textsc{Disjunct-QBF}(\TCNF),
where colors green, red and blue represent individual clauses of each formula and black color represent clauses that are shared.
Note that $y_2$ ($y_3$)
appears negatively with $z_2$ in red ($z_1$ in red) 
and positively with $z_2$ in green ($z_3$ in blue).
The formulas $\Phi_{\{a_1,a_2,a_3,a_4\}}$ and $\Phi_{\{a_1,a_2,a_3,\neg a_4\}}$ 
presents the steps of iteration of \Cref{eq:pies_update_short},
where $y_2$ is chosen to be the reducible.}
\end{figure}
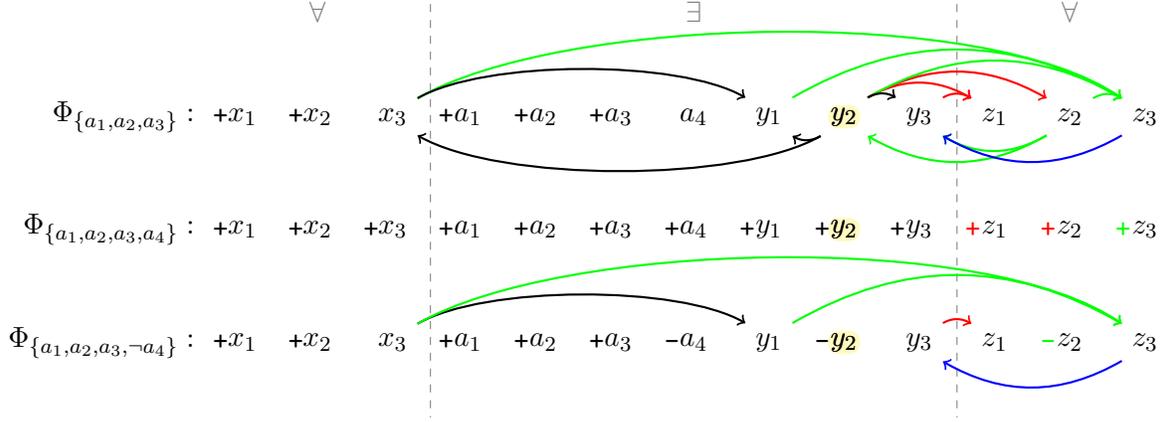

We first show that there will be at most $2^{2k^2+k}$ updates, i.e.,
there is some $i\leq 2^{2k^2+k}$ such that $\Phi^i$ fulfills (A3')
and we will define $\pies(\Phi)$ to be equal to $\Phi^i$.
Let $f(\Phi_L)$ denote the sum $\sum_{\phi \in \Phi_L} \min(2k, u(\phi))$,  
where $u(\phi)$ is the number of unit clauses in $\phi$  
that contain variables from $X_q$.  
Assume that $\Phi_L$ does not satisfy (A3').  
Then, there is a reducible variable $x \in X_{q-1}$
and $x$ appears positively in a clause of $\phi$  
and negatively in a clause of $\phi'$ both containing variables from $X_q$.
In this case, $f(\Phi_L) < 2k \cdot |\Phi_L|$,  
and assigning a value to $x$ increases the number  
of unit clauses in either $\phi$ or $\phi'$,  
because $f(\{\phi\}) < f(\{\phi[x=0]\})$ and $f(\{\phi'\}) < f(\{\phi'[x=1]\})$.  
Therefore, we have  
$f(\Phi_L) < f(\Phi_{L \cup \{a_{|L| + 1}\}})$ and  
$f(\Phi_L) < f(\Phi_{L \cup \{\neg a_{|L| + 1}\}})$.  
This implies that each $\Phi_L$ can be updated at most $2k^2$ times (because,
as we show in \Cref{lem:pies_A_assigment}, each $\Phi_L$ contains at most $k$ disjuncts),  
and in the end, we obtain a formula satisfying (A3').
Finally, it holds that $\Phi \Leftrightarrow \pies(\Phi)$,  
because each update from $\Phi^i$ to $\Phi^{i+1}$ produces an
equivalent formula. We therefore obtain the following (which can
formally be obtained from \Cref{lem:pies_A_assigment,lem:tcnf_pies_limit,lem:phi_pies,lem:computecat}).
\begin{corollary}\label[corollary]{cor:phi_pies_short}
    Let $\Phi $ be a propagated $k$-\textsc{Disjunct-QBF}(\TCNF),
    then 
    $\Phi \Leftrightarrow \pies(\Phi)$,
    $\pies(\Phi)$ satisfies (A1), (A2), and (A3'), and
    $\pies(\Phi)$ is a $2^{2k^2 + k}k$-\textsc{Disjunct-QBF}$(\TCNF)$.
\end{corollary}
As mentioned above, the variables in the innermost existential
quantifier block can now be removed from $\pies(\Phi)$ and the at most
$\Asize$ fresh variables in $A_\Phi$ can then be removed using
standard quantifier elimination, which results in a formula with one
less existential quantifier block. Repeating this process then results
in a universal $k'$-\textsc{Disjunct-QBF}$(\TCNF)$ with $k'$ being
bounded by a function of $k+q$, which can be solved
using~\cite[Corollary 7]{DBLP:conf/ijcai/ErikssonLOOPR24} in FPT-time
parameterized by $k'$.

\subsection*{Organization}
After having introduced the basic notions in
Section~\ref{sec:preliminaries} we begin the paper with
Section~\ref{sec:disjunct} where we relate backdoors to disjunctive
formulas, and Section~\ref{sec:hardness}, where we use this connection
in order to prove lower bounds. In Section~\ref{sec:algorithms} we
begin the algorithmic part of the paper with FPT-algorithms for
backdoor evaluation into $\TCNF$ and $\AFF$ with bounded quantifier
alternations. We continue with the introduction of our method for
eliminating guarded universal sets and its applications for enhanced
(deletion) backdoors in Section~\ref{sec:enhanced} and provide a brief
comparison to the backdoors based on variable dependencies and
introduced by Samer and Szeider~\cite{SamerSzeider07a} in
Section~\ref{sec:comp-dep-back}. Finally, we wrap up the paper in Section~\ref{sec:discussion} with a brief summary and some interesting directions for future research.

\section{Preliminaries} \label{sec:preliminaries}
We begin the paper by introducing the notation necessary for the rest of the paper.

\subsection{Clauses, Relations and Formulas}

Boolean expressions, formulas, variables, 
literals, conjunctive normal form (CNF), clauses and atoms are
defined in the standard way~(cf.~\cite{HandbookOfSat}).
We treat $1$/$\top$ and $0/\bot$ as the truth values
``true'' and ``false'', respectively, 
and formulas in \CNF as sets of clauses where each clause is a set of literals. 
More generally, it is convenient to use a relational perspective where atoms are expressions of the form $R(x_1,  \ldots, x_r)$ where $x_1, \ldots, x_r$ are variables and $R \subseteq \{0,1\}^r$ is an $r$-ary relation. A function $f \colon \{x_1, \ldots, x_r\} \to \{0,1\}$ then satisfies an atom $R(x_1, \ldots, x_r)$ if $(f(x_1), \ldots, f(x_r)) \in R$. Note that clauses can be defined in this way, e.g., the clause $(x \lor y \lor \neg z)$ is equivalent to an atom $R(x,y,z)$ where $R = \{0,1\}^3 \setminus \{(0,0,1)\}$.

Given a Boolean formula $\phi$, we write $\vars(\phi)$ for the set of variables in $\phi$. 
For a (partial) assignment $\tau : V' \rightarrow \{0,1\}$, where $V'\subseteq V(\phi)$, 
and an atom $R(x_1, \ldots, x_r)$, 
we let $R(x_1, \ldots, x_r)[\tau]$ be the atom obtained by (1) removing any tuple $(b_1, \ldots, b_r) \in R$ if there is an $1 \leq i \leq r$ where $x_i \in V'$ and $\tau(x_i) \neq b_i$, and then (2) projecting away any $x_i$ where $x_i \in V'$. 
This easily extends to conjunctions/sets of atoms, and for a formula $\phi$ we thus let 
$\phi[\tau]$ be the formula obtained by simplifying each atom according to $\tau$.
If $V' = \{x\}$ and $\tau(x) = b$, then we simply write $\phi\repl{x = b}$.
In the special case when $\phi$ is in CNF, $\phi[\tau]$ is the formula obtained from $\phi$ by removing all clauses 
satisfied by $\tau$ and removing all literals from clauses that are set to $0$ by $\tau$. 
For a conjunctive $\phi$ and a set of variables $S\subseteq V(\phi)$,
we denote by $C(\phi,S)$ the set of all atoms that contain at least
one variable from $S$. If $\phi$ is in CNF we write $\phi - S$ for the
formula obtained from $\phi$ after removing the variables in $S$ from
every clause of $\phi$.

Let us now define the following classes of conjunctive Boolean formulas.

\begin{enumerate}
\item
 \DCNF for the class of formulas in \CNF where each clause has at most $d$ literals.
\item
 \HORN for the class of formulas in \CNF where each clause has at most one positive literal.
\item
\AFF for the class of conjunctive formulas with atoms of the form
  $(\ell_1 \oplus \ell_2 \oplus \dots \oplus \ell_k) = b$, where
  $\oplus$ is the XOR operator, $b
  \in \{0,1\}$, and each $\ell_i = b_i x_i$ for a variable $x_i$ and
  coefficient $b_i \in \{0,1\}$. Without loss of generality, we will
  represent the atoms of affine formulas by a pair $(A,b)$, where $A$ is
  a set of variables and $b \in \{0,1\}$. Moreover, for a partial
  assignment $\tau \colon X \rightarrow \{0,1\}$ and an (affine) equation
  $C=(A,b)$, we denote by $C[\tau]$ the atom $(A\setminus
  V(\tau),b+A[\tau])$, where $A[\tau]=\sum_{v \in V(\tau)\cap A}\tau(v)$. We let $d$-$\AFF \subseteq \AFF$ where each involved equation has arity at most $d \geq 1$.
\end{enumerate}
We say that a conjunctive formula is in \emph{clausal form} if all
its atoms are clauses.
We refer to clauses with a single literal as \emph{unit clauses},
and slightly abusing terminology,
refer to single-variable atoms $(\ell) = b$ 
of affine formulas as unit clauses as well.

    The {\em satisfiability} problem for $\phi$ is then simply to determine whether there exists a satisfying assignment $f \colon \vars(\phi) \to \{0,1\}$. It is worth remarking that the satisfiability problem for propositional formulas in $2$-CNF, Horn and affine, is in P, while  $k$-CNF for  $k \geq 3$ yields an NP-hard satisfiability problem~\cite{sch78}.

For a pair of formulas $\phi_1$ and $\phi_2$
on the same set of variables $V$, we say that
$\phi_1$ and $\phi_2$ are \emph{equisatisfiable}
if, for every assignment $\alpha : V \to \{0,1\}$,
we have that $\alpha$ satisfies
$\phi_1$ if and only of
$\alpha$ satisfies $\phi_2$.
We say that two QBF formulas are equisatisfiable
if they have the same quantifier prefix
and their matrices are equisatisfiable.

\subsection{The Quantified Boolean Formula Problem}

A \emph{quantified Boolean formula (QBF)} 
is of the form $\cQ. \phi$, where 
$\cQ = Q_1 x_1 \ldots Q_n x_n$ with
$Q_i \in \{\forall, \exists\}$ for all $1 \leq i \leq n$
is the \emph{(quantifier) prefix},
$x_1, \dots, x_n$ are variables, and
$\phi$ is a Boolean formula 
%$\phi$
on these variables
called the
\emph{matrix}.
If $Q_i = \forall$, we say that
$x_i$ is a \emph{universal variable},
and if $Q_i = \exists$, we say that
$x_i$ is an \emph{existential variable}.

The truth value of a QBF $\cQ. \phi$ is 
defined recursively.
Let $\cQ = Q_1 x_1 \dots Q_n x_n$ and
$\cQ' = Q_2 x_2 \dots Q_n x_n$.
Then $\cQ. \phi$ is \emph{true} if
\begin{itemize}
  \item $\phi$ is a true Boolean expression, or
  \item $Q_1 = \exists$ and it holds that 
  $\cQ'. \phi\repl{x_1 = 0}$ \emph{or} $\cQ'. \phi\repl{x_1 = 1}$ is true, or
  \item $Q_1 = \forall$ and 
  $\cQ'. \phi\repl{x_1 = 0}$ \emph{and} $\cQ'. \phi\repl{x_1 = 1}$ are true.
\end{itemize}
Otherwise, $\cQ. \phi$ is \emph{false}. 
For a QBF $\Phi=\pref.\phi$ and a partial assignment $\tau : V' \rightarrow \{0,1\}$, where $V'\subseteq V(\Phi)$, we write $\Phi[\tau]$ for the QBF formula $\pref'.\phi[\tau]$, where $\pref'$ is obtained from $\pref$ after removing the variables in $V'$. We represent CNF formulas by sets and thus take us the liberty to write e.g.\ $\phi \in \Phi$ to denote a clause $\phi$ occurring in the QBF $\Phi$, i.e., we do not always bother to first specify the quantifier prefix and then extract the clause.

A QBF is in {\em block} form if it is of the form 
$Q_1 X_1 \dots Q_q X_q. \phi$ where $X_1,\ldots, X_q$ are pairwise
disjoint set of variables, $Q_i\in \{\exists,\forall\}$ for all $i\in
[q]$ and $Q_i\neq Q_{i+1}$ for all $i\in [q-1]$.
The parameter $q - 1$ is called the {\em quantifier alternation depth}.

\begin{definition}
For a class of formulas $\CCC$ we write $\QBF(\CCC)$ for the computational problem of deciding whether a QBF formula $\cQ . \phi$, where $\phi \in \CCC$, is true.
\end{definition}

The evaluation  of a QBF  $Q_1 x_1 \dots Q_n x_n . \phi$ can be seen
as a two player game (a.k.a. Hintikka game) played
between the {\em universal} and the {\em existential}
player. In the $i$-th step of the game the $Q_i$-player, i.e., the
existential player if $Q_i=\exists$ and the universal player
otherwise),
assigns a value to the
variable $x_i$. The existential player wins the game if $\phi$
evaluates
to true under the assignment constructed in the game and the universal player wins
if $\phi$ evaluates to false.

A \emph{strategy} for a player is a rooted binary tree \( T \) of
depth \( n+1 \) where nodes at level \( j \in \{0, \ldots, n-1\} \)
are labeled \( x_{j+1} \), leaves are labeled \( \top \) or \( \bot
\), and edges are labeled \( 0 \) or \( 1 \). If $T$ is an existential
strategy, nodes labeled by existential variables have one child, and
those by universal variables have two; the roles are reversed for
universal strategies. Each node \( t \) corresponds to the assignment
\( \tau_t^T \) given by the vertex and edge labels on the root-to-\( t
\) path, defining \( \phi[\tau_l^T] \) as the label of each leaf \( l
\). Let \( \ass(T) \) denote all root-to-leaf assignments of $T$. An
existential strategy is \emph{winning} if all leaves are \( \top \),
and a universal strategy is \emph{winning} if all leaves are \( \bot \).

    \begin{proposition}
    A QBF formula $\Phi$ is true ($\Phi \Leftrightarrow \top$) if and only if the existential player has a winning strategy, and $\Phi$ is false ($\Phi \Leftrightarrow \bot$) if and only if 
    the universal player has a winning strategy.     
    \end{proposition}

As an important intermediate problem in our algorithms and lower bound proofs,
we will use \QBF on disjunctive formulas.
Formally, for a class of formulas $\CCC$ and 
a positive integer $q \in \ZZ_+$, let 
$q$-\textsc{Disjunct}-$\CCC$ be the class of formulas
$\phi$ such that $\phi = \phi_1 \lor \ldots \lor \phi_q$
and let $q$-\textsc{Disjunct}-$\QBF(\CCC)$ be the class of QBF
formulas, where the matrix is in $q$-\textsc{Disjunct}-$\CCC$,
for some $\phi_1, \dots, \phi_q \in \CCC$.
We also use $q$-\textsc{Disjunct}-$\QBF(\CCC)$ to denote the \QBF{}
problem on $q$-\textsc{Disjunct}-$\QBF(\CCC)$ formulas.
For an instance $\Phi = \cQ. (\phi_1 \lor \ldots \lor \phi_q)$ 
of the latter problem,
we use shorthand $\phi \in \Phi$ for $\phi \in \{\phi_1,\dots,\phi_q\}$.
Additionally, if $\CCC \subseteq \CNF$,
for a disjunctive formula $\Phi = \bigvee_{i=1}^{q} \phi_i$ and a set of variables $B$, 
we write $\Phi - B$ for the formula $\bigvee_{i=1}^{q} \phi_i - B$.

    \subsection{Parameterized Complexity}
    
    We use standard notation and primarily follow~\cite{downey2013fundamentals}~and~\cite{fomin2019kernelization}.
    For a finite alphabet $\Sigma$ a 
    {\em parameterized problem} $L$ is a subset of $\Sigma^* \times \NN$.
    The problem $L$ is \emph{fixed-parameter tractable} (or, in FPT)
    if there is an algorithm deciding 
    whether an instance $(I, k) \in \Sigma^* \times \NN$ is in $L$
    in time $f(k) \cdot |I|^c$, where
    $f$ is some computable function and
    $c$ is a constant independent of $(I, k)$.
    An equivalent definition of FPT can be given by a \emph{kernelization (algorithm)} for $L$, i.e.,
    an algorithm that takes $(I, k) \in \Sigma^* \times \NN$ as input and
    in time polynomial in $|(I,k)|$, 
    outputs $(I',k') \in \Sigma^* \times \NN$ such that: $(I,k)
      \in L$ if and only if $(I',k')\in L$, and $|I'|,k' \leq h(k)$ for
      some computable function $h$. 
      We say that problem $L$ is in XP
      if there is an algorithm deciding 
      whether an instance $(I, k) \in \Sigma^* \times \NN$ is in $L$
      in time $|I|^{f(k)}$ for some computable function
      $f$.
    
    Let $L, L' \subseteq \Sigma^* \times \NN$ be two parameterized problems.
    A mapping $P : \Sigma^* \times \NN \to \Sigma^* \times \NN$ is a 
    \emph{fixed-parameter (FPT) reduction from $L$ to $L'$}
    if there exist computable functions $f,p : \NN \to \NN$
    and a constant $c$ such that the following conditions hold:
    \begin{itemize}
      \item $(I,k) \in L$ if and only if $P(I,k) = (I',k') \in L'$,
      \item $k' \leq p(k)$, and
      \item $P(I,k)$ can be computed in $f(k) \cdot |I|^c$ time.
    \end{itemize}
    
    Parameterized complexity also contains a complementary theory of hardness. Here, the hard classes $W[1]\subseteq W[2] \subseteq \ldots $, form a hierarchy of classes. We say that a problem is {\em $W[1]$-hard} if it admits an FPT-reduction from \textsc{Independent Set}
    (parameterized by the number of vertices in the independent set). 
    Similarly, we say that  a problem is {\em $W[2]$-hard} if it admits an FPT-reduction from \textsc{Hitting Set}
    parameterized by the number of vertices in the hitting set.
    Many problems in this paper turn out to be considerably harder than this and we say that a parameterized problem $X \subseteq \Sigma^* \times \mathbb{N}$ is {\em para-PSpace-hard} if there exists a constant $c \geq 0$ such that the slice $\{x \mid (x,k) \in X, k \leq c\}$ is PSpace-hard.

    For sharper lower bounds, stronger assumptions are sometimes necessary. For $d \geq 3$ let $c_d$ denote the infimum of all constants $c$ such that $d$-SAT is solvable in $2^{c n}$ time by a deterministic algorithm. The {\em exponential-time hypothesis} (ETH) then states that $c_3 > 0$, i.e., that 3-SAT is not solvable in subexponential time.
    The {\em strong exponential-time hypothesis} (SETH) additionally conjectures that the limit of the sequence $c_3, c_4, \ldots$ tends to $1$, which in particular is known to imply that the satisfiability problem for clauses of arbitrary length
    is not solvable in $2^{c n}$  time for {\em any} $c < 1$~\cite{impagliazzo2009}.

    We will also frequently make use of the \textsc{Hitting Set} problem,
    which is defined as follows.

    An instance $I=(U,\mathcal{F},k)$ of \textsc{Hitting Set} is a set
    (universe) $U$, a family $\mathcal{F}$ of subsets of $U$, and an
    integer $k$ and the task is to decide whether $I$ has a \emph{hitting
      set} of size at most $k$, i.e., a set $H$ of size at most $k$ such
    that $H\cap F\neq \emptyset$ for every $F \in \mathcal{F}$.
    \begin{proposition}[\cite{downey2013fundamentals}]
    \label[proposition]{pro:hit}
    \textsc{Hitting Set} can be solved in time $\bigoh(d^k|\mathcal{F}|d)$,
      where $d$ is the maximum size of any set in $\mathcal{F}$.
    \end{proposition}

\subsection{Backdoors}
We begin by defining the notion of a {\em strong backdoor}.

\begin{definition}
Let $\CCC$ be a class of QBF formulas.
    For a  \QBF formula $\Phi=\cQ. \phi$ we define a {\em (strong)
      backdoor} to $\CCC$ as a set $B \subseteq \vars(\phi)$ where $\Phi[\tau] \in \mathcal{C}$ for every $\tau \colon B \to \{0,1\}$.
\end{definition}

Moreover, a set $B \subseteq \vars(\phi)$ is a \emph{(variable) deletion backdoor} if $\phi -B \in \mathcal{C}$. Note that backdoors and deletion backdoors coincide for the base classes \TCNF and \HORN.

When working with backdoors to a class of formulas $\CCC$,
it is natural to assume that the class is 
\emph{closed under partial assignments}, i.e., for every formula
$\Phi \in \CCC$ and every partial assignment $\tau : U \to \{0,1\}$,
where $U \subseteq V(\Phi)$, the subformula $\Phi[\tau] \in \CCC$.
An even stronger property for the class is being
\emph{closed under taking subformulas}:
for every $\Phi \in \CCC$, any formula that can be
obtained from $\Phi$ by removing a subset of variables
and a subset of clauses also belongs to $\CCC$.
This property holds for all classes considered in the paper,
and we will assume henceforth for all classes that 
they are closed under taking subformulas without explicitly
stating this in every result.

\begin{example}
    We consider a simple example to exemplify the basic concepts. Consider the relation $R = \{(a,b,c,d) \in \{0,1\}^4 \mid a \geq 1, (a \xor b \xor c \xor d = 0)\}$ and the $\QBF$ instance over variables $\{x,y,z,u,v,w\}$ and \[\psi = \exists x \forall v, w, \exists y,z,u . R(x, y, z, u) \land (x \lor v \lor w) \land (x \lor v).\]
    The set $\{x,v\}$ is indeed a  backdoor to $\AFF$ since we for
    any $\tau \colon \{x,v\} \to \{0,1\}$ have 
    \begin{enumerate}
        \item 
    $R(x,y,z,u)[\tau]$ can
    be represented by $\{(y\oplus z\oplus u=1)\}$ if $\tau(x) =1$, and
    by a trivially false set of equation if $\tau(x) = 0$ (e.g., $\{(x=0),(x=1)\}$, 
    \item $(x \lor v \lor w)[\tau]$ by either  $\{(w=0)\}$ or $\{(w=1)\}$, and $(x \lor v)[\tau]$ by either a trivially true, or trivially false, system of equations. 
    \end{enumerate}
\end{example}

Our backdoor notion naturally yields QBF problems parameterized by
$|B|$, i.e., we are interested in
solving an instance as fast as possible provided we are given a
(deletion) backdoor, a problem usually known as {\em (deletion) backdoor
evaluation}. 
    On the other hand, the {\em detection} problems are
    much simpler since they do not differ from the
    corresponding detection problem for the Boolean satisfiability problem
    and the more general constraint satisfaction problem. Specifically,
    backdoor detection is FPT for $\TCNF$ and
    $\HORN$~\cite{DBLP:conf/birthday/GaspersS12}, and $\AFF$ backdoor
    detection for  \QBF with arbitrary atoms is FPT when additionally
    parameterized by the maximum arity of any involved
    relation~\cite[Theorem 4]{DBLP:conf/dagstuhl/GaspersOS17} and is
    $W[2]$-hard otherwise~\cite[Theorem 7]{DBLP:conf/dagstuhl/GaspersOS17}.
%    \todo{SO: say why it also works in the unmixed case}

    \subsection{Graphs and Iterated Logarithms}
    We use basic terminology for undirected and directed
    graphs~\cite{DiestelBook}. For the problem of finding a maximum matching in a graph we need the following lemma.
    
    \begin{proposition}\label[proposition]{pro:matching}
      Let $G$ be a bipartite graph with partition $\{A,B\}$. Then, a maximum matching in $G$ can be computed in time $\bigoh(|A|^{3}+|V(G)|+|E(G)|)$.
    \end{proposition}
    \begin{proof} 
      Let $G'$ be the bipartite graph obtained from $G$ as follows.
      For every vertex $a \in A$, let $N'(a)$ be an arbitrary
      subset of size at most $|A|$ of the neighbors $N_G(a)$ of $a$ in
      $G$. Then, $G'$ has vertices $A \cup \bigcup_{a \in A}N'(a)$
      and an edge between $a$ and $b$ if $b \in N'(a)$.
      Then, $|V(G')|\leq |A|^2+|A|$ and $|E(G)|\leq |A|^2$.
      Therefore,
      using e.g.~\cite{DBLP:journals/siamcomp/HopcroftK73} we can compute a maximum matching of $G'$ in time
      $\bigoh(\sqrt{|V(G')|}|E(G')|)=\bigoh(|A|^3)$ and since
      any maximum matching in $G'$ is also a maximum matching in $G$, this
      allows us to compute a maximum matching for $G$.
\end{proof}

    \emph{Tetration} $^na$, or the power tower of $a$s of order $n$, is (for natural numbers) defined as $a^{a^{a^{\dots^a}}}$, containing $n$ $a$s, or 
    \begin{equation*}
    	^nx = 
    	\begin{cases} 
    		% 1&\quad n=0\text{,} \\ 
    		1&\quad n<1\text{,} \\ 
    		% a^{(^{n-1}a)}&\quad n\geq1\text{.}
    		x^{(^{n-1}x)}&\quad n\geq1\text{.}
    	\end{cases}
    \end{equation*}
    
    The \emph{Iterated logarithm} (or \emph{Super-logarithm}, even if the definitions of these two normally differs slightly) of base $b$ $\log^*_bx$ is defined as 
    \begin{equation*}
    	\log^*_bx = 
    	\begin{cases} 
    		0&\quad x\leq1\text{,} \\ 
    		1+\log^*_b(\log_bx)&\quad x>1\text{.}
    	\end{cases}
    \end{equation*}
    Essentially, $\log^*_bx$ is the inverse of $^xb$ and $x=\log^*_b {^x b}$.
    
    As a alternative form of tetration we also define $\exp_x^{y}(z)$ as
    \begin{equation*}
    	\exp_x^{n}(y) = 
    	\begin{cases} 
    		% 1&\quad n=0\text{,} \\ 
    		y&\quad n<1\text{,} \\ 
    		% a^{(^{n-1}a)}&\quad n\geq1\text{.}
    		x^{\exp_x^{n-1}(y)}&\quad n\geq1\text{.}
    	\end{cases}
    \end{equation*}
    This is mainly a more pleasant expression to use when the top of a tower differs from from the tower itself. 
    Additionally $\exp_x^{n}(y) \geq ^{n+\log^*_xy}x$, were the difference comes only from a rounding error in $\log^*$\footnote{Optionally, more general definitions of $^nx$, $\log^*_bx$, and $\exp_x^{n}(y)$ over $\mathbb{R}$ can be used to allow $\exp_x^{n}(y) = ^{n+\log^*_xy}x$.}.

\section{From Backdoors to Disjunctive Formulas, and Back}
\label{sec:disjunct}

Recall that an instance of $k$-\textsc{Disjunct}-$\QBF(\CCC)$
is a quantified formula $\cQ. \phi_1 \lor \dots \lor \phi_k$,
where each disjunct $\cQ. \phi_i$ belongs to $\CCC$. We often represent the  disjuncts by the set $\{\phi_1, \ldots, \phi_k\}$ and if $\Phi$ is a $k$-\textsc{Disjunct}-$\QBF(\CCC)$ formula we write $\phi \in \Phi$ to extract a specific disjunct (i.e., we do not always bother to first specify the quantifier prefix).
We will connect the complexity of $\QBF(\CCC)$
parameterized by backdoor size to 
$k$-\textsc{Disjunct}-$\QBF(\CCC)$ parameterized by $k$.
We present two reductions in the forward direction:
one using generic quantifier elimination and
an alternative approach that works under a mild technical
assumption that unit clauses are allowed, 
does not change the variable set and produces equisatisfiable formulas.
The latter is utilized in some of our algorithms.
We also present a reduction in the opposite direction, i.e.,
producing a formula with a small backdoor to $\CCC$
from a $k$-\textsc{Disjunct}-$\QBF(\cC)$.

Let us start with the quantifier elimination approach.
To describe it,
let $\Phi = \cQ. \phi$, where $\cQ = Q_1 x_1 \dots Q_n x_n$.
For an index $i \in [n]$, let $\phi'_{>i}$ be the formula
obtained from $\phi$ by substituting every variable $x_j$ 
where $i+1 \leq j \leq n$ with $x'_j$.
We denote by $\cQ_{<i}$ the prefix $Q_1 x_1 \dots Q_{i-1} x_{i-1}$ and 
by $\cQ'_{>i}$ the prefix 
$Q_{i+1} x_{i+1} Q_{i+1} x'_{i+1} \dots Q_{n} x_{n} Q_{n} x'_{n}$. 
Observe that if $Q_i = \exists$, we have
\begin{equation}
  \label{eq:elim-exists}
  \Phi \Leftrightarrow \cQ_{<i} \cQ'_{>i} . \phi\repl{x_i = 0} \lor \phi'_{>i}\repl{x_i = 1}
\end{equation}
and if $Q_i = \forall$, then
\begin{equation}
  \label{eq:elim-forall}
  \Phi \Leftrightarrow \cQ_{<i} \cQ'_{>i} . \phi\repl{x_i = 0} \land \phi'_{>i}\repl{x_i = 1}.
\end{equation}

\begin{lemma}\label[lemma]{lem:from-backdoor-to-disjunct-qe}
  There is a reduction that takes
  an instance $\Phi$ of  \QBF
  with $n$ variables and 
  a backdoor of size $k$ to a class $\CCC$, 
  and in time $\bigoh(k' \cdot \sizeOf{\Phi})$
  computes an (satisfiability-)equivalent
  instance of 
  $k'$-\textsc{Disjunct}-$\QBF(\CCC)$
  with $k' = 2^{2^{k-1}}$ and 
  at most $2^k n$ variables.
\end{lemma}
\begin{proof}
  % \todo{check proof for runtime (G: done)}
  We apply quantifier elimination to the backdoor
  variables in the reverse order
  of their appearance in the prefix.
  One can think of the matrix of 
  the original  QBF formula
  as a $1$-\textsc{Disjunct}-$\CCC$.
  %\mateusz{We are using the following notation everywhere: $\Phi = \cQ. \bigvee_{i \in [q]}\phi_i$.}
  More generally,
  let $\Phi = \cQ. \phi$, $\phi =\bigvee_{i \in [q]}\phi_i$ be 
  $q$-\textsc{Disjunct}-$\CCC$
  %such that $\phi = \phi_1 \lor \dots \lor \phi_q$
  and let $B$ be a  backdoor to $\CCC$ 
  for all $\phi$ with $i \in [q]$.
  Let $x_i$ be the backdoor variable appearing the latest in $\cQ$.
  We show that using quantifier elimination,
  we can compute an equivalent
  $q'$-\textsc{Disjunct}-$\CCC$ with $q' \leq \max(2q,q^2)$, 
  at most twice as many variables as $\Phi$, and
  every disjunct having a backdoor $B \setminus \{x_i\}$
  to $\CCC$.
  We start with $q=1$ and backdoor of size $k$ to $\CCC$,
  so by applying these steps iteratively
  we obtain the lemma.

  Consider formulas
  $\phi\repl{x_i=0}$ and $\phi'_{>i}\repl{x_i=1}$.
  Note that $(\phi_1 \lor \ldots \lor \phi_q)\repl{x_i = 0} = 
  \phi_1\repl{x_i = 0} \lor \ldots \lor \phi_q\repl{x_i = 0}$ and similarly
  for $\phi'_{>i}\repl{x_i = 1}$.
  Thus, both formulas are in $q$-\textsc{Disjunct}-$\CCC$
  and $B \setminus \{x_i\}$ is a backdoor to $\CCC$ for both.
  In total, they have at most twice as many variables as $\phi$.
  If $Q_i  = \exists$, then by~\eqref{eq:elim-exists},
  we immediately obtain an equivalent $2q$-\textsc{Disjunct}-$\CCC$
  with the desired properties.
  If $Q_i = \forall$, then, following~\eqref{eq:elim-forall},
  $\Phi$ is equivalent to
  a formula with matrix $\phi\repl{x_i = 0} \land \phi'_{>i}\repl{x_i = 1}$,
  which can be converted into a $q^2$-\textsc{Disjunct}-$\CCC$
  by the standard rule, i.e., 
  $(\alpha_1 \lor \ldots \lor \alpha_q) \land (\beta_1 \lor \ldots
  \lor \beta_q)$ is equisatisfiable to $\bigvee_{j \in [q]} \bigvee_{j' \in [q]} (\alpha_j \land \beta_{j'})$.
\end{proof}

Now we present the alternative approach bypassing quantifier elimination.
For a conjunctive formula $\phi$
and a set of variables $X \subseteq V(\phi)$,
define the disjunctive formula 
\[ 
  \disj(\phi, X) = \bigvee_{\tau : X \to \{0,1\}} \phi[\tau] \land U_\tau,
\]
where $U_\tau$ is a conjunction of
unit clauses defined as follows:
if $\tau(x) = 1$, then $U_\tau$ contains the clause
$(x)$, otherwise it contains $(\lnot x)$.
Observe that $\tau$ is the only
satisfying assignment to $U_\tau$.

\lemequisat*
\begin{proof}
  Suppose $\alpha : V(\phi) \to \{0,1\}$
  satisfies $\phi$.
  Let $\tau$ be the restriction of $\alpha$ to $X$.
  By definition, $\alpha$ satisfies $U_{\tau}$.
  Moreover, it satisfies $\phi[\tau]$ because
  it is a subformula of $\phi$.

  On the other hand, suppose $\alpha' : V(\phi) \to \{0,1\}$
  satisfies $\disj(\phi,X)$.
  Then there exists $\tau' : X \to \{0,1\}$
  such that $\alpha'$ satisfies $\phi[\tau']$
  and $U_{\tau'}$.
  By the latter, $\alpha'$ agrees with $\tau'$ on $X$,
  so if $\tau'$ satisfies a clause in $\phi$,
  then so does $\alpha'$.
  Moreover, $\alpha'$ satisfies $\phi[\tau']$,
  i.e., all clauses in $\phi$ that are not
  satisfied by $\tau'$, hence $\alpha'$ satisfies $\phi$.
\end{proof}

We say that a class $\CCC$
\emph{allows adding unit clauses} if for every $\Phi$ in $\CCC$, 
the formula $\Phi \cup \{U_x\}$, 
where $U_x \in \{(x), (\lnot x)\}$ for some $x \in V(\Phi)$,
is also in $\CCC$.
Let $\CCC$ be a class of $\QBF$ formulas that allows adding unit clauses
and let $\Phi = \cQ. \phi$ be a QBF formula
with a backdoor $B \subseteq V(\phi)$ to $\CCC$.
Then, $\cQ. \disj(\Phi, B)$ is a
$2^{|B|}$-\textsc{Disjunct}\hy $\QBF(\CCC)$ formula.
Moreover, $\Phi$ and $\cQ. \disj(\Phi, B)$
are equisatisfiable by \Cref{lem:equisat}, 
which yields us the following lemma.

\lemfrombackdoortodisjunct*

Finally, we show a reduction in the opposite direction
that allows us to transfer lower bounds 
from the disjunctive problem to the backdoor problem.

\lemfromdisjuncttobackdoor*

\begin{proof}
  Let $\Phi = \cQ. \bigvee_{i \in [k]} \phi_i$ be a formula 
  in $k$-\textsc{Disjunct}-$\CCC$.
  We can assume without loss of generality that
  $k$ is a power of two: indeed,
  if it is not, we can add copies of $\phi_1$ as
  disjuncts until the total number of disjuncts
  becomes a power of two; clearly, the resulting
  formula is equivalent to the original.
  Assume $k = 2^q$ and
  introduce a set of variables $Z = \{z_1, \dots, z_q\}$.
  Let $\mathbf{t}_1, \dots, \mathbf{t}_k$ be
  an enumeration of the tuples in $\{0,1\}^q$,
  and define relations $T_i = \{0,1\}^q \setminus \{\mathbf{t}_i\}$.
  Now, let $\Phi' = \cQ \exists Z. \bigwedge_{i=1}^{k} \phi'_i$, 
  where $\phi'_i$ is obtained from $\phi_i$
  by replacing every atom $R(x_1,\dots,x_r)$ with
  $R'(x_1,\dots,x_r,z_1,\dots,z_q)$ where the relation $R'$ is such that 
  \[ R'(x_1,\dots,x_r,z_1,\dots,z_q) \equiv R(x_1,\dots,x_r) \lor T_i(z_1,\dots,z_q). \]
  Note that $Z$ is a backdoor to $\CCC$ for $\Phi'$.
  We claim that $\Phi \Leftrightarrow \Phi'$.
  Indeed, if the existential player has a winning strategy on $\Phi$,
  then at least one disjunct $\phi_i$ is satisfied in every branch.
  To extend the strategy to $\Phi'$, 
  set the variables in $Z$ to $\mathbf{t}_i$.
  In the opposite direction, 
  suppose the existential player 
  has a winning strategy on $\Phi'$,
  and let $\alpha$ be a partial assignment
  to the variables in $\cQ$
  in the winning strategy tree.
  Observe that $\bigwedge_{i=1}^{k} T_i(z_1,\dots,z_q)$
  is not satisfiable because $T_1 \cap \ldots \cap T_k = \emptyset$.
  Hence, for some $i \in [k]$, we have that
  $\phi'_i[\alpha] \not\equiv T_i(z_1,\dots,z_q)$.
  Since $\phi'_i \Leftrightarrow \phi_i \lor T_i(z_1,\dots,z_q)$,
  this implies that $\alpha$ satisfies $\phi_i$, 
  and we are done.
\end{proof}

In summary, 
\Cref{lem:from-backdoor-to-disjunct-qe},
\Cref{lem:from-backdoor-to-disjunct}~and~\Cref{lem:from-disjunct-to-backdoor}
imply that for a class of formulas $\CCC$,
 $\QBF$ parameterized by backdoor size to $\CCC$
is in FPT/W[1]-hard if and only 
$k'$-\textsc{Disjunct}-$\QBF(\CCC)$
parameterized by $k'$ is in FPT/W[1]-hard.
In the forthcoming sections our principal study object is therefore
$k'$-\textsc{Disjunct}-$\QBF(\CCC)$ and we do not 
always stress its connection to backdoors.

\section{Lower Bounds}
\label{sec:hardness}

Consider two parameters:
the size $k$ of a backdoor to
polynomial-time tractable classes of QBF formulas, namely
$\TCNF$, $\HORN$, and $\AFF$,
and the number $q$ of quantifier alternations.
We show that only having $k$ as the parameter
is not sufficient to obtain XP algorithms:
even for $k=2$ all these problems turn out to be PSpace-hard.
For Horn formulas, even parameterizing by $k+q$
does not allow to escape $W[1]$-hardness.
In the next section we show that the parameter $k+q$
can lead to FPT-algorithms for other base classes.

\subsection{The Squishing Lemma and Its Consequences}
\label{lem:squishing}

\sloppypar
We will show that evaluating 
QBF formulas with small backdoors to 
any class $\CCC \in \{\TCNF, \HORN, \AFF\}$ is para-PSpace-hard.
By Lemma~\ref{lem:from-disjunct-to-backdoor},
it suffices to show that $k$-\textsc{Disjunct}-$\QBF(\CCC)$
is PSpace-hard for a constant $k$.
As an intermediate step,
we will prove hardness for a simple class of formulas
$\Gamma_{\AFF} \subseteq \AFF$
which contains formulas with atoms of the form
$(x) = 0$, $(x) = 1$ and $(x \oplus y) = 0$.
To simplify the notation,
we will use the following shorthands
$(\neg x)$, $(x)$ and $(x = y)$ for the atoms above.
    The starting point of our hardness proof is a simple lemma.
    
    \begin{lemma}\label[lemma]{lem:dnf-to-disjunct}
      \textsc{Disjunct}-$\QBF(\Gamma_{\AFF})$ is PSpace-complete.
    \end{lemma}
    \begin{proof}
      Recall that $\QBF(\text{3-\CNF})$ is PSpace-complete.
      Let $\Phi = \cQ. (C_1 \land \ldots \land C_m)$ be a quantified
      $3$-CNF formula and define
      $\bar{\Phi} = \bar{\cQ}. (\bar{C}_1 \lor \ldots \lor \bar{C}_m$),
      where $\bar{\cQ}$ is the prefix obtained from $\cQ$
      by replacing every $\forall$ with $\exists$ and vice versa,
      while $\bar{C_i}$ is the negation of $C_i$, e.g.,
      if $C_i = (x \lor y \lor \neg z)$, then 
      $\bar{C}_i = (\neg x) \land (\neg y) \land (z)$.
      By de Morgan's law, $\Phi \Leftrightarrow \neg \bar{\Phi}$,
      and $\bar{\Phi}$ is a disjunction of $m$ formulas in $\Gamma_{\AFF}$.
    \end{proof}
    
    Observe that the \textsc{Disjunct}-$\QBF(\Gamma_{\AFF})$ formula obtained
    in the lemma above has many disjuncts, but each disjunct
    is of constant size.

The main technical part of the section is
``squishing'', where we reduce the number of 
disjuncts to a constant at the expense of 
increasing the size of the disjuncts.
Having access to equality is essential here.

    Before we present the squishing lemma,
    we need another technical step.
    
    \begin{lemma}\label[lemma]{lem:split}
      Let $\phi$ be a formula in $\Gamma_{\AFF}$ on variables $X$.
      For every positive integer $q \geq 2$,
      there exist $q$ satisfiable formulas
      $\phi_1, \dots, \phi_q$ in $\Gamma_{\AFF}$ on variables $X \cup Y$
      such that
      \begin{enumerate}[1.]
        \item for all $i \in [q]$, the size of $\phi_i$ is at most $\bigoh(|\phi|)$,
        \item for all $i \in [q]$, we have $\forall X \exists Y. \phi_i \Leftrightarrow \top$, and
        \item for every quantifier prefix $\cQ$ on variables $X$,
          we have $\cQ. \phi \Leftrightarrow \cQ \exists Y. \phi_1 \land \ldots \land \phi_q$.
      \end{enumerate}
    \end{lemma}
    \begin{proof}
      Let $\phi = C_1 \land \ldots \land C_m$.
      For every $i \in [m]$, choose 
      any literal $\ell_i$ in $C_i$ and create 
      fresh variables $y_{i,1}, \ldots, y_{i,q-1}$.
      Create formulas $\phi_1, \ldots, \phi_q$ as follows:
      for every $i \in [m]$, 
      \begin{itemize}
        \item add atom $C'_i$ to $\phi_1$, where $C'_i$ is obtained
          from $C_i$ by replacing literal $\ell_i$ with $y_{i,1}$,
        \item for every $2 \leq j \leq q-1$,
          add $(y_{i,j-1} = y_{i,j})$ to $\phi_j$, and
        \item add $(y_{i,q-1} = \ell_i)$ to $\phi_q$.
      \end{itemize}
      This completes the construction.
    
      The first statement in the lemma follows
      by construction.
      For the second statement, note that 
      every atom $C'_i$ of $\phi_1$ contains a unique variable $y_{i,1}$ from $Y$,
      and $C'_i$ is in $\Gamma_{\AFF}$,
      so for every assignment to $X$,
      the existential can choose the value of $y_{i,1}$ to satisfy $C'_i$.
      In case of $\phi_q$, let $\ell_i \in \{x, \neg x\}$.
      Then $\forall x \exists y_{i,q-1}. (y_{i,q-1} = \ell_i)$ is true because 
      any assignment for $x$ by the universal player
      can be matched on $y_{i,q-1}$ by the existential player. 
      Finally, for $\phi_j$ with $2 \leq j \leq q-1$,
      the formulas do not contain any variables from $X$ 
      and are clearly satisfiable.
      For the third statement, note that for each $i \in [m]$,
      the formula $\phi_i$ contains the atom $C'_i$
      where $\ell_i$ is replaced by $y_{i,1}$,
      and the atoms equivalent to the chain of equalities 
      $y_{i,1} = y_{i,2} = \dots = y_{i,q-1} = \ell_i$.
    \end{proof}

    We are now ready to prove the squishing lemma.

\begin{lemma}\label[lemma]{lem:squish}
  Let $k$ and $p$ be positive integers and $K = \binom{k}{2^p}$.
  There is a polynomial-time algorithm that takes
  an instance $\Phi$ of $K$-\textsc{Disjunct}-$\QBF(\Gamma_{\AFF})$ with $m$ atoms
  and produces an equivalent instance $\Phi'$ of
  $k$-\textsc{Disjunct}-$\QBF(\Gamma_{\AFF})$
  with $\bigoh(kp+m2^p)$ atoms.
\end{lemma}
    \begin{proof}
      Throughout this proof
      we will use the following indices:
      \[ i \in [K], \ j \in [2^p], \ \ell \in [p], \ s \in [k]. \]
      Let $\Phi = \cQ. \bigvee_{i} \phi_i$.
      Apply Lemma~\ref{lem:split} to each $\phi_i$
      with $q = 2^p$, obtaining formulas
      $\phi_{i,j}$ for all $j$.
      Note that we have split the disjuncts into $K \cdot 2^p$ ``pieces''.
      We will now ``distribute the pieces'' among $k$ newly defined formulas $\phi'_1,\dots,\phi'_k$.
      Arbitrarily enumerate the subsets of $\{1,\dots,k\}$ of size $2^p$
      as $S_1, \dots, S_K$.
      For every disjunct $\phi_i$, we now have a subset $S_i \subseteq [k]$ of size $2^p$,
      and we will add a ``piece'' of $\phi_i$ to $\phi'_s$ 
      if the element $s$ belongs to $S_i$.
      Note that there are $2^p$ ``pieces'' of $\phi_i$,
      and exactly $2^p$ elements in $S_i$.
      To formalize this, let the elements of $S_i$ be $s_{i,j}$,
      and for every $s \in [k]$, define 
      \[ 
        \phi'_s = \bigwedge_{i,j \text{ such that } s_{i,j} = s} \phi_{i,j}.
      \]
      Now, let $Y_{i,j}$ be the set of additional 
      variables in the formula $\phi_{i,j}$, and 
      $Y = \bigcup_{i,j} Y_{i,j}$.
      Introduce $kp + p$ additional variables 
      $Z = \{z_{1,1}, \dots, z_{k,p}\}$ and $W = \{w_{1}, \dots, w_{p}\}$.
      Finally, let 
      \[ 
        \Phi' = \cQ \exists Y \exists Z \forall W. 
        \bigvee_{s} \phi'_s \land \phi''_s.
      \]
      , where:
      \[  
        \phi''_s = \bigwedge_{\ell \in [p]} (z_{s,\ell} = w_\ell).
      \]
      Clearly, $\Phi'$ can be computed in polynomial time.
    
      Before we proceed with correctness, let us prove
      a useful claim.
    
      \begin{claim*}
        For every $T \subseteq [k]$, the formula 
        $(\exists Z \forall W. \bigvee_{t \in T} \phi''_t) \Leftrightarrow \top$
        if and only if $|T| \geq 2^p$.
      \end{claim*}
      \begin{claimproof}
        Let us view 
        $\mathbf{z}_t = (z_{t,1}, \dots, z_{t,p})$ for $t \in T$
        and $\mathbf{w} = (w_1, \dots, w_p)$
        as tuples of Boolean variables.
        The formula is then equivalent to 
        $\exists Z \forall W. \bigvee_{t \in T} \mathbf{z}_t = \mathbf{w}$.
        If $|T| \geq 2^p$, then the existential player
        can assign every Boolean tuple of size $p$ to some $\mathbf{z}_t$,
        leaving the universal player with no choice
        but to satisfy one of the disjuncts $\mathbf{z}_t = \mathbf{w}$.
        If $|T| < 2^p$, then for every assignment to $Z$,
        there exists a tuple of $p$ Boolean values 
        that is not assigned to any of $\mathbf{z}_t$,
        and the universal player wins by assigning
        that tuple to $\mathbf{w}$.
      \end{claimproof}
    
      To prove that $\Phi \Rightarrow \Phi'$,
      consider a winning existential strategy on $\Phi$
      and suppose $\phi_i$ is satisfied in a leaf.
      By Lemma~\ref{lem:split}, this strategy can be extended to the prefix
      $\cQ \exists Y$ so that $\phi_{i,j}$ is satisfied for all $j$.
      By definition, $\phi'_t$ for all $t \in S_i$ are satisfied,
      so the resulting formula contains 
      $\exists Z \forall W. \bigvee_{t \in S_i} \phi''_t$,
      and since $|S_i| = 2^p$, the claim implies that $\Phi' \Leftrightarrow \top$.
    
      For the opposite direction, i.e., $\Phi' \Rightarrow \Phi$,
      consider a winning existential strategy on $\Phi'$.
      Consider the formula obtained in a leaf node after playing $\cQ \exists Y$.
      Let $T$ be the set of indices $s$ such that
      $\phi'_s$ is satisfied.
      Then the formula is 
      $\exists Z \forall W. \bigvee_{t \in T} \phi''_t$.
      By the claim, $|T| \geq 2^p$,
      so $S_i \subseteq T$ for some $i$, hence 
      the existential player satisfies 
      $\phi_{i,j}$ for all $j$, and, hence, 
      $\phi_i$ by Lemma~\ref{lem:split}.
    \end{proof}

Note that by picking $k=4$ and $p=1$,
one can reduce any formula with $6$ disjuncts
to a formula with $4$ disjuncts.
If a formula has $5$ disjuncts, then
we can add a copy of any disjunct to obtain
an equivalent $6$-disjunct formula to squish again.
Observe further that the size of the formula only 
grows by a constant factor with every application
of the lemma, and each time the number of disjuncts
decreases.
Thus, by repeatedly applying squishing 
to $k$-\textsc{Disjunct}-$\Gamma_{\AFF}$,
    we obtain the following.

    \begin{corollary} \label{cor:bounded_arity_aff}
      \label{cor:squish}
      $4$-\textsc{Disjunct}-$\QBF(\Gamma_{\AFF})$ is PSpace-hard.
    \end{corollary}
In turn, the corollary implies lower bounds 
for $\TCNF$, $\HORN$, and $\AFF$.

\begin{theorem}
  \label{thm:PSpace-hard}
  If $\CCC$ is \TCNF, \HORN, or \AFF then
  $\QBF$ is PSpace-hard for backdoors to $\CCC$ of size $2$.
\end{theorem}
\begin{proof}
  The result for $\AFF$ follows directly since $\Gamma_{\AFF} \subseteq \AFF$.
  For \CNF classes, take a quantified formula
  with matrix in $4$-\textsc{Disjunct}-$\Gamma_{\AFF}$,
  and replace every atom of the form $(x = y)$
  with a conjunction of clauses
  $(x \lor \neg y) \land (\neg x \lor y)$,
  and unit atoms with equivalent unit clauses.
  Observe that the resulting
  formula is a disjunction of $4$ formulas in
  $\TCNF \cap \HORN$,
  and the result follows.
\end{proof}

In addition to these hardness proofs, squishing also 
gives us the following lower bounds
under the ETH.
We start by observing that squishing can be done quickly.
    \begin{lemma}\label[lemma]{lem:squick}
      There is an algorithm that takes
      an instance of
      $k$-\textsc{Disjunct}-$\QBF(\Gamma_{\AFF})$ and 
      in polynomial time returns 
      an equivalent instance
      with $4$ disjuncts
      and at most $2q$ additional quantifier blocks,
      where $q \leq \log^*_{1.889} k+1$.
    \end{lemma}
\begin{proof}
    We repeatedly use Lemma~\ref{lem:squish}, assuming iteration $i-1$ outputs an $k_i$-\textsc{Disjunct}-$\Gamma_{\AFF}$.
    Let $k_0=k$, and choose $p_i$ such that $k_i/3 \leq 2^{p_i} \leq 2k_i/3$.
    By e.g. Stirling's approximation, we have $({27/4})^x/x\leq\binom{3x}{x}\leq({27/4})^x$, giving us, in the worst case, \[{(\log_{1.889}k_{i-1})}\leq k_i \leq {(\log_{1.889}k_{i-1}+\log_{1.889} \log_{1.889}k_{i-1})}\] since $\sqrt[3]{27/4}=1.889$.
    In the edge case when $\binom{k_i}{2^{p_i}} > k_{i-1}$, we can simply pad
    the formula by adding sufficiently many copies of any of its disjuncts --
    this clearly yields an equivalent formula.
    As there for every $k_{i-1} > 4$ there exists a $k_i<k_{i-1}$ and $p_i$ such that $\binom{k_i}{2^{p_i}} \geq k_{i-1}$, 
    and by the previously mentioned worst-case relation between $k_{i-1}$ and $k_{i}$, we have $k_i\leq 4$ for all $i \geq \log^*_{1.889} k + 1$.
    Additionally, every iteration of Lemma~\ref{lem:squish} adds at most one $\exists$-block and one $\forall$-block.
\end{proof}

    Using Lemma~\ref{lem:squick} we can now take any 3-CNF-QBF formula, and produce an equisatisfiable $k$-\textsc{Disjunct}-QBF($\CCC$) formula for any $k\geq4$, but at a cost of increasing the number of quantifier alternations by at most two times $1+\log^*_{1.889} m$.
    This then gives us the following lemma:
    \begin{lemma}
      Let $\CCC \in \{\TCNF, \HORN, \AFF\}$.
      Assuming ETH, 
      there is a constant $c > 1$ such that no algorithm can solve
      \QBF on $k$-\textsc{Disjunct}-QBF($\CCC$) instances
      with $2q+1$ quantifiers in time 
      $c^{(^{(q-1+\log_{1.889}^*k)}\beta)}=c^{exp_\beta^{(q-1+\log_{1.889}^*k)}(1)}$
      for any $\beta < 1.889$.
    \end{lemma}
    \begin{proof}
      It is sufficient to show the result for $\CCC = \Gamma_{\AFF}$.
      Under the ETH, the sparsification lemma implies that
      there exists a constant $c > 1$ such that
      no algorithm can decide whether 
      a Boolean formula in 3-\CNF with $m$ clauses 
      is satisfiable in time $c^m$.
      By Lemma~\ref{lem:dnf-to-disjunct},
      this immediately implies that
      no algorithm can solve 
      \QBF on $m$-\textsc{Disjunct}-QBF($\Gamma_{\AFF}$) formulas
      in time $c^m$.
      Lemma~\ref{lem:squick} gives us that any such instance can be squished in $q \leq 1+\log^*_\alpha m$ steps to a $k$-\textsc{Disjunct}-QBF($\Gamma_{\AFF}$) instance, adding at most $q$ $\exists$- and $\forall$-blocks, with $\alpha=\sqrt[3]{27/4}=1.889$ and any $k\geq4$.
        Hence if $^{(q-1+\log_{\alpha}^*k)}\beta \leq (1-\epsilon)m$ for any $\epsilon>0$, i.e.
        for any $\beta<\alpha$, solving $\Phi'$ in 
        $c^{(^{(q-1+\log_{\alpha}^*k)}\beta)}=c^{exp_\beta^{(q-1+\log_{\alpha}^*k)}(1)}$ time implies solving $\Phi$ in $c^m$ time, contradicting the ETH.
    \end{proof}

\subsection{Lower bounds for Backdoor to Horn with Bounded Quantifier Depth}

We show that (even bounded arity) Horn behaves differently from $\TCNF$ and $\AFF$ and that it remains hard {\em even} if one bounds the number of quantifier alternations, or considers backdoors to unmixed Horn formulas. To prove this we provide a reduction from the 
\textsc{Multicolored Independent Set (MCIS)} problem, where
we are given a graph $G$ and an integer $k$; 
vertices $V(G) = V_1 \uplus \ldots \uplus V_k$ 
which are partitioned into $k$ sets that induce
cliques in $G$, and where the goal is to find
an independent set with exactly one
vertex in each $V_i$.

\begin{theorem} \label{thm:bounded_arity_horn}
  $\QBF(\CNF)$ parameterized by (1) backdoor size and quantifier-alternation depth to $3$-\HORN, or (2) backdoor size to unmixed \HORN, is W[1]-hard.
\end{theorem}
\begin{proof}
We prove the result by modifying a reduction from Eriksson et al.~\cite{DBLP:conf/ijcai/ErikssonLOOPR24}.
  Let $(G, k)$ with $V(G) = V_1 \uplus \ldots \uplus V_k$,
  be an instance of MCIS.
  We construct a formula $\Phi$ on variables
  $X = \{x_v : v \in V(G)\}$ and $Z = \{z_i : i \in [k]\}$
  as follows.
  For a vertex $v \in V_i$ in $G$, let $N_G(v)$ be the neighbors of $v$
  in $G$, and define a clause
  \begin{align*}
    C_v = \{x_v\} \cup
    \{\neg x_u : u \in N_G(v) \} \cup
    \{\neg z_i\}
  \end{align*}
  and define
  \[
    \Phi = \forall X \exists Z. \bigwedge_{v \in V} C_v \land (z_1 \lor \ldots \lor z_k).
  \]
  Observe that each clause $C_v$ is in \HORN, and if ignoring $Z$, even in unmixed \HORN. Moreover,
  $Z$ is thus a backdoor of size $k$.
  Using a standard reduction, one can transform
  $\forall X. \bigwedge_{v \in V(G)} C_v$ into an equivalent
  formula $\forall X \exists Y. \phi$, where $\phi$ is in $3$-\HORN. 
  Thus, to obtain the claim, we will show that $\Phi \Leftrightarrow \bot$
  if and only if $(G, k)$ contains an independent set of size $k$.
  
  Suppose the $\forall$-player has a winning strategy on $\Phi$, i.e.
  there exists an assignment $\alpha : X \to \{0,1\}$
  such that $\Phi[\alpha] \Leftrightarrow \bot$.
  Then $\Phi[\alpha] = \exists Z. (\neg z_1) \land\ldots\land (\neg z_k) \land (z_1 \lor \ldots \lor z_k)$,
  so for every $i \in [k]$, there exists a clause $C_v$ with 
  $v \in V_i$ such that all literals except $\neg z_i$ in $C_v$
  are assigned $0$ by $\alpha$, i.e.
  $\alpha(x_v) = 0$ and $\alpha(x_u) = 1$ for all $u \in N_G(v)$.
  For every $i \in [k]$, let $C_{v_i}$ be such a clause,
  and define a subset $S = \{v_1, \dots, v_k\}$ of vertices in $G$.
  Note that $\alpha(X_{v_1}) = \dots = \alpha(x_{v_k}) = 0$, hence
  if $v_i, v_j \in S$, then $v_i \notin N_G(v_j)$. 
  We conclude that $S$ is an independent set in $G$.

  Now assume $S = \{v_1,\dots,v_k\}$ is 
  an independent set in $G$ with $v_i \in V_i$.
  Define assignment $\alpha : X \to \{0,1\}$
  such that $\alpha(x_v) = 0$ if $v \in S$
  and $\alpha(x_v) = 1$ if $v \in V(G) \setminus S$.
  Note that in a clause $C_{v_i}$, the assignment
  $\alpha$ assigns $0$ to all literals in $C_{v_i}$
  except $\neg z_i$: indeed, $\alpha(x_{v_i}) = 0$ by definition,
  and $\alpha(x_u) = 1$ for all $u \in N_G(v_i)$
  because $S$ is an independent set.
  Hence, $\Phi[\alpha] \Leftrightarrow \bot$,
  and $\alpha$ is a winning strategy for the $\forall$-player.
\end{proof}

\section{Algorithms for Backdoors}
\label{sec:algorithms}

Here, we will show the algorithmic part of \Cref{thm:introtractcnf},
which is restated below.

\thmintrotractcnf*

We start by showing the algorithm for \TCNF in \Cref{ssec:m2cnf} and
then provide the algorithm for \AFF in \Cref{ssec:affine}.

\subsection{Algorithm for \TCNF}
\label{ssec:m2cnf}

Here, we will provide a formal proof of \Cref{the:tcnf-tract},
restated below, which
using \Cref{lem:from-backdoor-to-disjunct} implies that \QBF is fixed-parameter tractable parameterized by the
size of a minimum backdoor set to $\QBF(\TCNF)$ plus the number of
quantifier alternations, i.e., the algorithm for
\TCNF stated in \Cref{thm:introtractcnf}.

\thetcnftract*

We start by giving a detailed exposition of the notion of propagation.
Let $C_1 = (\ell \lor A)$ and $C_2 = (\neg \ell \lor B)$ be two
clauses in a \CNF formula, where $\ell$ is a literal. 
The {\em resolution} of $C_1$ and $C_2$ derives the {\em resolvent}
clause $C = (A \lor B)$.

We say that a $\TCNF$ formula $\phi$ is {\em propagated} if it 
satisfies the following three properties: 
\begin{itemize}
    \item[(1)] if $\bot \in \phi$ or $\emptyset \in \phi$, then $\phi = \{\bot\}$, 
    \item[(2)] it is closed under resolution, and 
    \item[(3)] no clause is a super-clause of another clause. 
    \end{itemize}
Note that (1) is implied by (2) and (3).

Let $\cnfprop{\phi}$ denote a \TCNF{} formula $\phi'$ such that 
$\phi'$ is propagated and $\phi \Leftrightarrow \phi'$.
\begin{example}\label{ex:cnFPTop}
  Consider the following formulas and their propagated forms:
  \begin{itemize}
  \item $\cnfprop{(\neg a \lor b) \land (\neg b \lor c)} 
    = (\neg a \lor b) \land (\neg b \lor c) \land (\neg a \lor c)$,
  \item $\cnfprop{(\neg a \lor b) \land (\neg b \lor c) \land c} 
    = (\neg a \lor b) \land c$,
  \item $\cnfprop{(\neg a \lor b) \land (\neg b \lor c) \land \neg c} 
    = \neg a \land \neg b \land \neg c$.
  \end{itemize}
\end{example}
It is well known that a \CNF{} formula $\phi$ is satisfiable if and only if
$\cnfprop{\phi}\neq \bot$.
While computing $\cnfprop{\phi}$ is not polynomial-time for general $\CNF$ formulas,  
it is polynomial-time for $\TCNF$, as stated below.

\begin{lemma}\label[lemma]{lem:prop_comp}
Let $\phi$ be a \TCNF formula. We can compute $\cnfprop{\phi}$ in 
time $\bigoh(\propT{|V(\phi)|})$, where $\omega$ is the matrix multiplication exponent ($2 \leq \omega \leq 2.38$~\cite{DBLP:conf/soda/WilliamsXXZ24}). Additionally, the size of 
$\cnfprop{\phi}$ is bounded by $\bigoh(|V(\phi)|^2)$.
\end{lemma}
\begin{proof}
  Note first that the fact that $|\cnfprop{\phi}|=\bigoh(|V(\phi)|^2)$
  holds, because every clause obtained via resolution from two clauses of
  size at most two has again size at most two and therefore
  $\cnfprop{\phi}$ is also a \TCNF{} formula.
  
  To propagate the formula, we will first compute all clauses that can
  be obtained from $\phi$ via resolution. Towards this aim, we will
  first build the implication graph $G_\phi$ of $\phi$ that has two
  vertices $x$ and $\lnot x$ for every variable $x \in V(\phi)$ and an
  arc from $l$ to $l'$ if and only if $\phi$ contains the clause
  $\lnot l \lor l'$. Note that this can be achieved in time
  $\bigoh(|V(\phi)|+|\phi|)=\bigoh(|V(\phi)|^2)$. We then compute the transitive closure
  $G_\Phi^T$ of
  $G_\phi$ in time 
  $\bigoh(|V(G_\phi)|^\omega)=\bigoh(|V(\phi)|^\omega)$
  using~\cite{DBLP:journals/ipl/Munro71}. Let $C$ be the set of clauses
  $x\lor y$ such that $(\lnot x,y) \in E(G_\Phi^T)$. Note that
  $C$ already contains all clauses of $\cnfprop{\phi}$ of size $2$, i.e.,
  $\SB c \in \cnfprop{\phi}\SM |c|= 2\SE \subseteq C$ and it
  remains only to compute the unit clauses of $\cnfprop{\phi}$ as well
  as to remove clauses from $C$ that contain such a unit clause.
  Note that the set of unit clauses of $\cnfprop{\phi}$, is given by
  $A=\SB N_{G_\phi^T}(\lnot l) \cup \{l\} \SM \{l\}\in \phi\SE$ and can
  be computed in time $\bigoh(|V(\phi)|^2)$. Finally, to ensure (3),
  let $C'$ be the set obtained from $C$ after removing all clauses  that contain a literal corresponding to a unit
  clause in $A$; note that $C'$ can be computed in time $\bigoh(|C|)$
  with the help of standard data structures. Then,
  $\cnfprop{\phi}=C'\cup A$ and the runtime to compute
  $\cnfprop{\phi}=C'\cup A$ is at most
  $\bigoh(|V(\phi)|^\omega)$. 
\end{proof}

We say that a $k$-\textsc{Disjunct-QBF}(\TCNF) $\Phi$ is {\em propagated} 
if every disjunct in $\Phi$ is propagated. Note that
$(\bigvee_{\phi \in \Phi}\phi)$ and $(\bigvee_{\phi \in
  \Phi}\cnfprop{\phi})$ are equisatisfiable and therefore also
$(\pref.\bigvee_{\phi \in \Phi}\phi)$ and $(\pref.\bigvee_{\phi \in
  \Phi}\cnfprop{\phi})$.
Moreover, for any partial assignment $\beta$
the formula $\phi[\beta]$ is either propagated or $\bot \in \phi[\beta]$. 
This directly yields the following lemma, which generalizes the
well-known Q-resolution to $k$-\textsc{Disjunct-QBF}(\TCNF) formulas,
and immediately allows us to remove the innermost existential
quantifier block.
Before we state the lemma, we briefly recall the notion of redundancy
that we first introduced in \Cref{ssec:unmixed-overview}. Let $\Phi$
be a $k$-\textsc{Disjunct}-$\QBF(\CNF)$ formula and let $\phi \in
\Phi$.
We say that $\phi$ is \emph{$\Phi$-redundant} if and only if $\Phi
\Leftrightarrow \Phi \setminus \{\phi\}$.
Similarly, a subset $\Phi' \subseteq \Phi$ is \emph{$\Phi$-redundant} if and only if $\Phi \Leftrightarrow \Phi \setminus \Phi'$.
Let $x \in V(\Phi)$.
We say that the variable $x$ is \emph{$\Phi$-redundant} if and only if 
$\Phi \Leftrightarrow \pref'.\bigvee_{\phi \in \Phi} \phi \setminus C(\phi, \{x\})$,
where $\pref'$ is the prefix of $\Phi$ after removing $x$.

\begin{lemma}\label[lemma]{lem:prop_2cnf}
    Let $\Phi$ be a propagated
    $k$-\textsc{Disjunct-QBF}(\TCNF),
    whose innermost quantifier block $X_q$ is existential.
    Then, $X_q$ is $\Phi$-redundant.
\end{lemma}
\begin{proof}
Let $\pref = Q_1 X_1 \dots \exists X_q$ be a prefix of $\Phi$.
For each $X \subseteq X_{q}$ ($\overline{X}=X_{q}\setminus X$), we define  
$$
\Phi^{X} = Q_1 X_1 \dots \exists X. 
\bigvee_{\phi \in \Phi} 
\left(\phi \setminus C(\phi, \overline{X})\right).
$$

Let $X \subset X_{q}$ and $x \in X_{q} \setminus X$.  
Now we show that $\Phi^{X \cup \{x\}}$ is equivalent to $\Phi^{X}$.

To show the forward direction, let $T$ be an existential winning 
strategy for $\Phi^{X \cup \{x\}}$. Let $l$ be any leaf node of $T$. 
Then, $\tau_l^T$ satisfies some disjunct of $\Phi^{X \cup \{x\}}$, 
and therefore also some disjunct of $\Phi^X$, as required.

For the reverse direction, let $T$ be an existential winning 
strategy for $\Phi^{X}$. Let $l$ be any leaf node of $T$. 
There exists $\phi \in \Phi^{X \cup \{x\}}$ such that
$\tau_l^T$ satisfies $(\phi \setminus C(\phi, \{x\}))$.
Let us assume that there exists $c_\top,c_\bot \in C(\phi, \{x\})$
such that
$c_\top$ and $c_\bot$ are not satisfiable by $\tau_l^T$, 
$x$ appears positively in $c_\top$ and $x$ appears negatively in $c_\top$.
Due to property (2) (closure under resolution), the resolvent $c_r$ of 
$c_\top$ and $c_\bot$ must be in $\phi$, and hence in 
$\phi \setminus C(\phi, \{x\})$. However, this implies 
$c_r[\tau_l^T] = \bot$, which is a contradiction. Therefore, such 
$c_\top$ and $c_\bot$ cannot exist simultaneously.
This means there is an assignment to $x$ that satisfies all clauses 
of $C(\phi, \{x\})$ not already satisfied by $\tau_l^T$. We modify $T$ 
by introducing a new level for $x$ after each leaf, assigning $x$ 
appropriately. This modified tree becomes an existential winning 
strategy for $\Phi^{X \cup \{x\}}$.

This shows that $\Phi^{X}$ and $\Phi^{X\cup \{x\}}$ are
equisatisfiable and by repeating this argument for all variables in
$X_{q}$, we obtain that $\Phi^{X_{q}} \Leftrightarrow \Phi^\emptyset$, 
which completes the proof.
\end{proof}
Let $\Phi = \pref.\bigvee_{i \in [k]} \phi_i$ be a 
$k$-\textsc{Disjunct-QBF}(\TCNF{}) with 
$\pref = Q_1 X_1 \dots Q_q X_q$. Because of \Cref{lem:prop_2cnf}, we will in the following assume that
the innermost quantifier block is universal. Recall from \Cref{ssec:2cnf} that
our aim is to construct an equivalent 
$2^{2k^2 + k}k$-\textsc{Disjunct-QBF}(\TCNF) formula  
$\pies(\Phi)$, where the
variables $X_{q-1}$ of the innermost
existential quantifier block have become redundant and their function
is replaced by a set $A_\Phi$ of at most $2k^2+k$ fresh existential
variables (inserted into the position of $X_{q-1}$ within the
quantifier prefix). For illustrative purposes, it
will be convenient to think of $\pies(\Phi)$ (and all intermediate
formulas leading to $\pies(\Phi)$) as the disjunction of at most $2^{2k^2 + k}$ many  
$k$-\textsc{Disjunct-QBF}(\TCNF) formulas $\Phi_L$, where $L$ is a
non-contradictory subset of literals over the set
$A_\Phi=\{a_1,\dotsc,a_{2k^2+k}\}$ of fresh existential variables, i.e.,
$\pies(\Phi) \Leftrightarrow \bigvee_L\Phi_L$, where we use
$\bigvee_L\Phi_L$ for the formula $\pref.\bigvee_L (\bigvee_{\phi \in
  \Phi_L}\phi)$ assuming that $\pref$ is the common prefix of all $\Phi_L$'s.  

Intuitively, the set $A_\Phi$ is a set of fresh existential variables
that will later replace all variables in $X_{q-1}$ and the existential player
uses the variables in $A_\Phi$ to
select which of the formulas $\Phi_L$ to play on by assigning values for the
variables in $A_\Phi$; importantly, every assignment of the variables
in $A_\Phi$ will falsify every $\Phi_L$ such that $L$ is falsified by
the assignment. Every formula $\Phi_L$ will have the following properties:
\begin{itemize}
\item[(A1)] $\Phi_L$ is propagated,
\item[(A2)] every two (possibly equal) disjuncts $\phi,\phi'\in \Phi_L$ 
  share the same set of clauses that do not contain variables from
  $X_q$, i.e., $\phi \setminus C(\phi, X_q) = \phi' \setminus
  C(\phi', X_q)$, and
\item[(A3)] for every variable $x \in X_{q-1}$, either the literal $x$
  or the literal $\neg x$
  does not occur in any clause (of any disjunct) of $\Phi_L$ with a variable in $X_q$.
\end{itemize}
Crucially (as we show in \Cref{lem:bqa-var-redundant}) it holds
that if each $\Phi_L$ satisfies (A1)--(A3), then $X_{q-1}$ is
$\pies(\Phi)$-redundant, which allows us to
safely ignore all variables in $X_{q-1}$, which have now been
replaced by the at most $2k^2+k$ variables in $A_\Phi$. Note that thanks to
\Cref{lem:tcnf_pies_redundant}, we can apply
\Cref{lem:bqa-var-redundant} already if (A3) only applies to all
disjuncts that do not contain $2k$ unit clauses with variables in
$X_{q}$, i.e., we can replace (A3) with the following weaker condition:
\begin{itemize}
\item[(A3')] The subset $\Phi'$ of $\Phi_L$ containing all disjuncts
  that contain fewer than $2k$ unit clauses with variables in $X_q$ satisfies:
  for every $x \in X_{q-1}$, either the literal $x$ or the literal $\neg x$
  does not occur in a clause with a variable in $X_q$.
\end{itemize}
Intuitively, this holds because all formulas that contain $2k$ unit
clauses with variables in $X_q$ can be safely ignored since they can
always be falsified by the universal player (see also
\Cref{lem:tcnf_pies_redundant}). Note that replacing (A3) with (A3')
is crucial to show that $\pies(\Phi)$ can be obtained from $\Phi$ in a bounded
number of steps.

Our goal is now to build such a disjunction  
of formulas $\Phi_L$ that satisfy (A1), (A2), and (A3').  
We begin by introducing the function $\del(\Phi)$,  
which can be seen as the initialization step of the algorithm.  
Intuitively, $\del(\Phi)$ produces an equivalent formula  
consisting of the disjunction of $2^k$  
distinct
$k$-\textsc{Disjunct-QBF}(\TCNF) formulas  
that satisfy properties (A1) and (A2).
We then create the formula $\pies(\Phi)$, which
satisfies property (A3'), as well as (A1), and (A2), iteratively using \Cref{eq:pies_update}  
starting from $\del(\Phi)$. 

We now recall the definition of $\del(\Phi)$.
Let $\Phi = \pref.\bigvee_{i \in [k]} \phi_i$ be a 
$k$-\textsc{Disjunct-QBF}(\TCNF{}) with 
$\pref = Q_1 X_1 \dots Q_q X_q$, and let 
$A_\Phi = \{a_i : i \in [\Asize]\}$ be a set of fresh (existential) variables.
For every non-empty subset $I \subseteq [k]$, let 
$\phi^{\del}_I$ denote the conjunction of all clauses from 
$\bigwedge_{i \in I} \phi_i$
that do not contain variables from 
$X_q$, i.e.,
$
\phi^{\del}_I = \bigwedge_{i \in I} 
\left( \phi_i \setminus C(\phi_i, X_q) \right)
$; see also \Cref{ex:del} for an example of the formula.
Then, the formula $\del(\Phi)$ is given as:
\begin{equation}%\label{en:del_def}
\del(\Phi) =  Q_1 X_1 \dots\exists (X_{q-1} \cup A_\Phi)\forall X_{q}.\bigvee_{\emptyset\neq I \subseteq [k]}
\Phi^\del_I,
% \bigvee_{i \in I} 
% \cnfpropname{}
% \biggl( 
% \phi_i \land \phi^{\del}_I
% \land (\bigwedge_{j \in I} a_j) \land (\bigwedge_{j \in [k] \setminus I} \neg a_j)\biggr)
\end{equation}
 where

$$
\Phi^\del_I =  \bigvee_{i \in I} 
\cnfpropname{}
\biggl( 
\phi_i \land \phi^{\del}_I
\land (\bigwedge_{j \in I} a_j) \land (\bigwedge_{j \in [k] \setminus I} \neg a_j)\biggr)
$$

\newcommand{\tcnfold}[1]{}
\tcnfold{
Let $\Phi = \pref.\bigvee_{i \in [k]} \phi_i$ be a
$k$-\textsc{Disjunct-QBF}(\TCNF)
with $\pref = Q_1 X_1 \dots Q_q X_q$.
We first define the function $\del(\Phi)$, which informally serves as
a kind of selection function.
If $Q_{q} = \exists$, then we let $\del(\Phi)$ be equal to
$\Phi$. Otherwise, let $\phi^{\del}_I = \bigwedge_{i \in I} (\phi_i
\setminus C(\phi_i, X_q))$ for every $I \subseteq [k]$ and 
let $\pref' = Q_1 X_1 \dots Q_{q-2} X_{q-2}\exists (A_\Phi \cup X_{q-1})\forall X_{q}$ be a
new quantifier prefix, where $A_\Phi = \{a_i : i \in [\Asize]\}$ is a
set of new variables; note that the variables in $A_\Phi$ will later
be used to replace all the variables in $X_{q-1}$. Then, we let $\del(\Phi)$ be equal to:
\[\pref'.\bigvee_{\emptyset\neq I \subseteq [k]}\bigvee_{i \in I} 
  \cnfprop{\phi_i \land \phi^{\del}_I
\land (\bigwedge_{j \in I} a_j) \land (\bigwedge_{j \in [k] \setminus I} \neg a_j)}.
\]
}
Note that $\del(\Phi)$ does not contain any clauses with variables
$a_i$ for $i > k$, which will later be used to obtain $\pies(\Phi)$.

The following two lemmas can now be used to show \Cref{cor:bqa-dog} and
establish the equivalence between $\del(\Phi)$ 
and $\Phi$, along with additional properties.

\begin{lemma}\label[lemma]{lem:bqa-dog-eq}
  Let $\Phi$ be a propagated $k$-\textsc{Disjunct-QBF}(\TCNF), whose innermost
  quantifier block is universal.
  Then, $\del(\Phi)$ is a $k2^k$-\textsc{Disjunct-QBF}(\TCNF) and $\Phi \Leftrightarrow \del(\Phi)$.
\end{lemma}
\begin{proof}
  Let $\Phi=\pref.\bigvee_{i \in [k]} \phi_i$  be a propagated
  $k$-\textsc{Disjunct-QBF}(\TCNF) with $\pref = Q_1 X_1 \dots Q_q
  X_q$, whose innermost quantifier block is universal. By construction
  $\del(\Phi)$ is a $k2^k$-\textsc{Disjunct-QBF}(\TCNF) and it
  remains to show that $\Phi \Leftrightarrow \del(\Phi)$.
  % Since $\Phi=\del(\Phi)$ if $Q_q=\exists$, we can assume that
  % $Q_q=\forall$.
      
  Towards showing the forward direction, let $T$ be an existential
  winning strategy for $\Phi$.
  Let $z$ be the outermost variable in $X_q$ with respect to $\pref$.
  Moreover, let $t \in V(T)$ be a node of $T$ labeled with the
  variable $z$ and let $I_t$ be the set of all indices $i \in [k]$ such that
  $\bot \notin \phi_i[\tau_t^T]$, i.e., the set of all indices that
  correspond to disjuncts of $\Phi$ that
  are not yet falsified. 
  For a subset $I \subseteq [k]$ with $I\neq \emptyset$, we denote by $\tau_I: A_\Phi
  \rightarrow \{0,1\}$ the assignment of the new variables in $A_\Phi$
  given by setting $\tau_I(a_i)=1$ if $i \in I$ and $\tau_I(a_i)=0$,
  otherwise. Note that $\tau_t^T$ satisfies $\phi_{I_t}^{\del}$
  (because $\phi_{I_t}^{\del}$ does not contain any clause containing
  a variable from $X_q$)  and
  $\tau_{I_t}$ satisfies $(\bigwedge_{j \in I_t} a_j) \land
  (\bigwedge_{j \in [k] \setminus I_t} \neg a_j)$. Therefore,
  any extension of $\tau_t^T$ (to the variables in 
  % \tdm{$X_{q-1}$ ?}
  $X_{q}$) satisfies
  $\phi_i$ if and only if the corresponding extension of $\tau_t^T
  \cup \tau_{I_t}$ satisfies
  $\Phi^\del_I$.
  This allows us to define an existential winning strategy $T'$ for
  $\del(\Phi)$ as follows. Informally, $T'$ is obtained from
  $T$ by adding the path $P_t$ corresponding to the assignment $\tau_{I_t}$
  before $t$ for every node $t$ of $T$ labeled by $z$. More formally, let $t\in V(T)$ be a node of $T$
  labeled by $z$, then $P_t$ is the path
  corresponding to the assignment $\tau_{I_t}$, i.e., $P=(v_1,\dotsc,
  v_{\Asize+1})$
  and for every $i \in [\Asize]$, $v_i$ is labeled
  $a_i$ and the edge $\{v_i,v_{i+1}\}$ is labeled
  $\tau_{I_t}[a_i]$. Now, we modify $T$ by adding the path $P_t$ 
  making $v_1$ the new child of the parent $p$ of $t$ and identifying
  $t$ with $v_{\Asize+1}$.
  
  Towards showing the reverse direction, let $T'$ be an existential
  winning strategy for $\del(\Phi)$ and let $T$ be the existential
  strategy for $\Phi$ obtained by restricting $T'$ to the variables in
  $V(\Phi)$; which is well-defined since all variables in $A_\Phi$ are
  existentially quantified. We claim that $T$ is an existential
  winning strategy for $\Phi$. To see this consider any leaf $l$ of
  $T'$ and let $\Phi^{\del}_I$ be some disjunct in
  $\del(\Phi)$ satisfied by $\tau_l^{T'}$. Then, the leaf
  corresponding to $l$ in $T$ satisfies $\phi_i$.  
\end{proof}
    
\begin{lemma}\label[lemma]{lem:bqa-dog-prop}
  Let $\Phi$ be a propagated $k$-\textsc{Disjunct-QBF}(\TCNF),
  whose innermost quantifier block $X_q$ is universal.
  Then, $\del(\Phi)$ can be computed in time $\bigoh(k2^k\propT{|V(\Phi)|})$.
  Moreover, $\del(\Phi)$ satisfies (A1) and (A2).
  % for every 
  % $\phi, \phi' \in \del(\Phi)$, if 
  % $C(\phi, A_\Phi) = C(\phi', A_\Phi)$, then 
  % $\phi \setminus C(\phi, X_q) = \phi' \setminus C(\phi', X_q)$.
\end{lemma}
\begin{proof}
  We begin by noting that the function $\del(\Phi)$ can be computed in 
  time $\bigoh(k2^k\propT{|V(\Phi)|})$. Specifically, each of the at most 
  $k2^k$ disjuncts of $\del(\Phi)$ can be computed via propagation in time 
  $\bigoh(\propT{|V(\Phi)|})$ 
  using \Cref{lem:prop_comp}. By construction, $\del(\Phi)$ is
  propagated and therefore satisfies (A1).
  
  Towards showing that $\del(\Phi)$ satisfies (A2),
  recall that every
  disjunct $\phi'$ of $\del(\Phi)$ has the form
  $\cnfprop{\phi_i \land \phi^{\del}_I \land (\bigwedge_{j \in I} a_j)
    \land (\bigwedge_{j \in [k] \setminus I} \neg a_j)}$ for some $I
  \subseteq [k]$ and $i \in I$, where $\phi^{\del}_I = \bigwedge_{i \in I} (\phi_i
  \setminus C(\phi_i, X_q))$.
  We claim that $\phi'\setminus
  C(\phi',X_q)=\cnfprop{\phi^{\del}_I}$, which
  concludes the proof that $\del(\Phi)$ satisfies (A2).
  Clearly, $\cnfprop{\phi^{\del}_I} \subseteq \phi'\setminus C(\phi',X_q)$, so suppose for a
  contradiction that there is a clause $c \in (\phi'\setminus C(\phi',X_q))\setminus
  \cnfprop{\phi^{\del}_I}$. Note that $c=a \lor b$ is in $\phi'$ if
  and only if the implication graph $G_{\phi'}$ of $\phi'$ contains a
  directed path $P$ from $\lnot a$ to $b$. Recall that the implication
  graph $G_\phi$ of a \TCNF formula $\phi$ has two
  vertices $x$ and $\lnot x$ for every variable $x \in V(\phi)$ and an
  arc from literal $l$ to literal $l'$ if and only if $\phi$ contains the clause
  $\lnot l \lor l'$. Note that if $P$
  does not contain any arc corresponding to a clause in
  $C(\phi_i,X_q)$, then $P$ is also in $G_{\cnfprop{\phi^{\del}_I}}$ and
  therefore $c \in \cnfprop{\phi^{\del}_I}$, a contradiction. Therefore,
  let $P'$
  be a maximal subpath of $P$ say from $l$ to $l'$ containing only arcs in
  $C(\phi_i,X_q)$. Because the endpoints of $P$ are in
  $V(\del(\Phi))\setminus X_q$, the same is true for the endpoints $l$
  and $l'$ of
  $P'$. But then because $\phi_i$ is propagated, $G_{\phi_i\setminus
    C(\phi_i,X_q)}$ (and therefore also $G_{\cnfprop{\phi^{\del}_I}}$) contains the arc $(l,l')$. Therefore,
  by replacing every such maximal subpath of $P$ in this manner
  we obtain a path $P''$ that is also in $G_{\cnfprop{\phi^{\del}_I}}$,
  which shows that $c \in \cnfprop{\phi^{\del}_I}$.
\end{proof}

We now recall the definition of $\pies(\Phi)$ (which will achieve (A3')) and state some
additional properties.
Let $\Phi = \pref.\bigvee_{i \in [k]} \phi_i$ be a
$k$-\textsc{Disjunct-QBF}(\TCNF) with $\pref = Q_1 X_1 \dots Q_q X_q$,
where $Q_q=\forall$. We start by setting $\Phi^0=\del(\Phi)$ and
iteratively obtain $\Phi^{i+1}$ from $\Phi^i$ as follows;
$\pies(\Phi)$ will then be equal to some $\Phi^j$, which satisfies (A3').
Note that
for every set $L$ of literals (that contains exactly one literal of
every variable in $\{a_1,\dotsc,a_k\}$), it holds that
$\Phi_L=\Phi^{\del}_{I(L)}$ and $\del(\Phi)=\pref.\bigvee_L\Phi_L$, where $I(L)$ is the subset
of $[k]$ containing $i \in [k]$ if and only if $a_i \in L$.

We say that a variable $x \in X_{q-1}$ is \emph{reducible} with
respect to $\Phi_L^i$ for some $i\geq 0$ and set $L$ of literals over $A_\Phi$, if there are two
disjuncts $\phi$ and $\phi'$ in $\Phi_L^i$ such that:
\begin{itemize}
\item[(P1)]\label{en:pies_2k} $\phi$ and $\phi'$ contain fewer than $2k$ unit clauses over
  variables in $X_q$,
\item[(P2)]\label{en:pies_x} $C(\phi,\{x\})\cap C(\phi,X_q)$ contains a clause where $x$
  occurs positively and $C(\phi',\{x\})\cap C(\phi',X_q)$ contains a
  clause where $x$ occurs negatively
\end{itemize}
If $x \in X_{q-1}$ is reducible with respect to $\Phi^i_L$, then
we denote by $\piesup(\Phi^i,L,x)$ the following formula (see also
\Cref{ex:piesup} for an example).
\begin{equation}\label{eq:pies_update}
\piesup(\Phi^i,L,x) = \pref.
\bigl(
\bigvee_{\phi \in (\Phi^i\setminus \Phi^i_L)} \phi 
\bigr)
\lor \Phi_{L\cup\{a_{|L|+1}\}}' \lor \Phi_{L\cup\{\neg a_{|L|+1}\}}'
\end{equation}
Here, $\Phi_{L\cup\{a_{|L|+1}\}}'=\bigvee_{\phi \in
  \Phi_L^i}(\cnfprop{\phi[x=1]}  \land x \land a_{|L|+1})$ and
$\Phi_{L\cup\{\neg a_{|L|+1}\}}'=\bigvee_{\phi \in
  \Phi_L^i}(\cnfprop{\phi[x=0]}  \land \neg x \land \neg a_{|L|+1})$.
Moreover, as long as $\Phi^i$ contains a
$k$-\textsc{Disjunct-QBF}(\TCNF) $\Phi_L^i$ 
that contains a variable $x \in X_{q-1}$ that is reducible, we let
$\Phi^{i+1}$ be equal to $\piesup(\Phi^i,L,x)$. Finally, $\pies(\Phi)$
is then equal to $\Phi^j$, where $j$ is the smallest integer such that
no variable $x \in X_{q-1}$ is reducible for any $\Phi^j_L$; note that
$\pies(\Phi)$ then satisfies (A3') by construction.
\tcnfold{
We now define the function $\pies(\Phi)$ that informally allows us to
assign (or reduce) variables in $X_{q-1}$ that occur both positively and
negatively in clauses with at least one variable in $X_q$.
Let $\Phi = \pref.\bigvee_{i \in [k]} \phi_i$ be a
$k$-\textsc{Disjunct-QBF}(\TCNF) with $\pref = Q_1 X_1 \dots Q_q X_q$.
If $Q_q=\exists$, then $\pies(\Phi)=\Phi$. Otherwise, let
$\Phi'=\del(\Phi)$. Let $\phi, \phi' \in \Phi'$ and $x \in
X_{q-1}$. We say that a variable $x \in X_{q-1}$ is \emph{reducible} with
respect to some set $L$ of literals over some variables in
$A_\Phi$ and formula $\Phi'$, if there are two
formulas $\phi$ and $\phi'$ in $\Phi'$ satisfying the following properties:
\begin{itemize}
\item[(P1)]\label{en:pies_PhiL} $C(\phi,A_\Phi)=C(\phi',A_\Phi)=L'$,
  where $L'$ is the set of unit clauses corresponding to the literals
  in $L$,
\item[(P2)]\label{en:pies_2k} $\phi$ and $\phi'$ contain fewer than $2k$ unit clauses over
  variables in $X_q$,
\item[(P3)]\label{en:pies_x} $C(\phi,\{x\})\cap C(\phi,X_q)$ contains a clause where $x$
  occurs positively and $C(\phi',\{x\})\cap C(\phi',X_q)$ contains a
  clause where $x$ occurs negatively
\end{itemize}
If $x \in X_{q-1}$ is reducible with respect to 
$L$ and $\Phi'$, then we denote by $\piesup(L,x,\Phi')$ the following
formula.
\begin{equation}\label{eq:pies_update}
\begin{split}
  \piesup(L,x,\Phi')  = & \bigvee_{\phi \in (\Phi'\setminus
                          \Phi'_L)}\phi \lor\\
  & \bigvee_{\phi \in \Phi_L'}\biggl((\cnfprop{\phi[x= 1]}  \land x \land
    a_{t}) \lor \\
  &  (\cnfprop{\phi[x= 0]}  \land \neg x \land \neg a_{t})\biggr)\\
\end{split}
\end{equation}
Here, $\Phi_L'=\SB \phi \in \Phi' \land C(\phi,A_\Phi)=L\SE$ and $t
\in [\Asize]$
is the smallest index such that $L$ does not contain a literal over
$a_t$; we will show in \Cref{lem:tcnf_pies_limit} that such an index
$t$ does always exist and therefore an exhaustive application of the
function $\piesup(-)$ is always possible.
Let $\Phi''$ be the formula obtained from $\Phi'$ after
exhaustively applying the function $\piesup(-)$, i.e., until
there is no longer a reducible variable. Then, $\pies(\Phi)$ is equal to $\Phi''$
after exhaustive propagation.
}

The following auxiliary lemma is a first step to bound the number of
disjuncts in $\pies(\Phi)$.

\begin{lemma}\label[lemma]{lem:pies_A_assigment}
  \tcnfold{
    Let $\Phi $ be a $k$-\textsc{Disjunct-QBF}(\TCNF), whose innermost
    quantifier block is universal, and let $\alpha$ be an
    assignment of $A_\Phi$, and let $\Phi'$ be obtained by an arbitrary number of
    applications of the function $\piesup(-)$ from $\del(\Phi)$. Then,
    $|\SB \phi \in \Phi'[\alpha]\SM \bot \notin \phi \SE|\leq k$.
  }
  Let $\Phi $ be a propagated $k$-\textsc{Disjunct-QBF}(\TCNF), whose
  innermost quantifier block is universal, let $\alpha$ be an
  assignment of $A_\Phi$, and let $i$ be an integer. Then,
  $|\SB \phi \in \Phi^i[\alpha]\SM \bot \notin \phi \SE|\leq
|\Phi_{L(\alpha)}^i|\leq k$, where $L(\alpha)$ is the set of literals
of $A_\Phi$ representing the assignment $\alpha$.
\end{lemma}
\begin{proof}
  Note that the statement of the lemma clearly holds
  initially, i.e., for $\Phi^0=\del(\Phi)$. Namely, $|\SB \phi \in
  \del(\Phi)[\alpha] \SM \bot \notin \phi\SE|\leq k$ in fact even
  $|\SB \phi\in
  \del(\Phi)[\alpha] \SM \bot \notin \phi\SE|\leq |I_\alpha|$, where $I_\alpha$ is the
  subset of $\Phi$ containing all $\phi_i$ with
  $\alpha(a_i)=1$. Moreover, the statement of the lemma is
  maintained after every application of the function 
  $\piesup(\Phi^i,L,x)$ since the two
  copies of every disjunct in $\Phi_L^i$ are added
  in conjunction with contradictory literals, i.e., $a_{|L|+1}$ and $\lnot a_{|L|+1}$, respectively.
\end{proof}

The following lemma now shows that $\piesup(-)$ only needs to be
applied a finite amount of times and therefore $\pies(\Phi)$ is well defined.
    
\begin{lemma}\label[lemma]{lem:tcnf_pies_limit}
  Let $\Phi$ be a propagated $k$-\textsc{Disjunct-QBF}(\TCNF), whose
  innermost quantifier block is universal.
  Then, the definition of $\pies(\Phi)$ is well-behaved, meaning the 
  variables in $A_\Phi$ are sufficient to allow for the exhaustive 
  application of the function $\piesup(-)$ to the formula
  $\Phi^0$. Moreover, $\pies(\Phi)=\Phi^i$ with $i \leq 2^{\Asize}$
  and $\pies(\Phi)$ has at most $k2^{\Asize}$
  disjuncts.
\end{lemma}
\begin{proof}
  By definition, every $\phi \in \Phi^0$ contains no clauses on any
  variable in $\SB a_i\SM k<i\leq \Asize\SE$. The main idea behind
  showing the first statement of the lemma, is to show that the number
  of unit clauses of variables from the last quantifier block
  increases with every application of $\piesup(-)$. In particular, we
  will show an increase of the following function $f$.    
  For a \CNF formula $\phi^*$, we let $f(\phi^*)$ be equal to the minimum of
  $2k$ and the number of unit clauses in $C(\phi^*,X_q)$.
  Moreover, we set $f(\Phi^*) = \sum_{\phi^* \in \Phi^*} f(\phi^*)$ for a set
  $\Phi^*$ of \CNF formulas.

  Let $x$ be a reducible variable for $\Phi^i_L$ that is witnessed
  by the disjuncts $\phi,\phi' \in \Phi_L^i$ satisfying (P1) and (P2)
  and assume that $\Phi^{i+1}=\piesup(\Phi^i,L,x)$.
  Recall that for every $\phi_L \in \Phi^i_L$, the formula
  $\piesup(\Phi^i,L,x)$ contains the two disjuncts
  $\phi_L^\top=(\cnfprop{\phi_L[x= 1]}  \land x \land
  a_{|L|+1})$ and $\phi_L^\bot=(\cnfprop{\phi_L[x= 0]}  \land
  \lnot x \land
  \lnot a_{|L|+1})$. Then,
  $f(\phi_L)\leq \min\{f(\phi_L^\top),f(\phi_L^\bot)\}$. Moreover,
  because of (P2) $C(\phi,A_\Phi)\cap C(\phi,X_q)$ contains a clause,
  say $l \lor x$,
  where $x$ occurs positively. Because $\phi$ is propagated, it does not contain the unit clause $l$; otherwise
  $l \lor x$ would no longer be contained in $\phi$. Therefore,
  $\cnfprop{\phi[x= 0]}  \land \lnot x \land
  \lnot a_{|L|+1}$ contains at least one more unit clause from $X_q$ than
  $\phi$, which because of (P2) shows that
  $f(\cnfprop{\phi[x= 0]}  \land \lnot x \land
  \lnot a_{|L|+1})>f(\phi)$. A similar argument shows that also
  $f(\cnfprop{\phi'[x= 1]}  \land x \land
  a_{|L|+1})>f(\phi')$. Therefore,
  $f(\Phi^{i+1}_{L\cup\{a_{|L|+1}\}})>f(\Phi^i_L)$ and $f(\Phi^{i+1}_{L\cup\{\lnot
    a_{|L|+1}\}})>f(\Phi^i_L)$. Since,
  $f(\phi^*)\leq 2k$ for any \CNF formula $\phi^*$ and because
  $|\Phi^\ell_{L^*}|\leq k$ for any $\ell$  and any set
  $L^*$ of literals using variables in $A_\Phi$ (because of \Cref{lem:pies_A_assigment}), we obtain that
  $f(\Phi^\ell_{L^*})\leq 2k^2$. Therefore, at most $2k^2$ literals from
  $A_\Phi$ can be added to any formula $\phi \in \Phi^i$, which
  concludes the proof of the first part of the lemma. Since
  $\Phi^x=\pies(\Phi)$ consists of at most $2^{\Asize}$ formulas
  $\Phi_L^x$ (one for every assignment of the variables in $A_\Phi$) and
  each $\Phi_L^x$ has at most $k$ disjuncts (using
  \Cref{lem:pies_A_assigment}), it follows that $\pies(\Phi)$ has at
  most $k2^{\Asize}$ disjuncts. Finally, because $\Phi^0$ consists of
  $2^k$ formulas $\Phi^0_L$ and every application of $\piesup(-)$ adds
  one such formula, it follows that at most $x=2^{\Asize}-2^k\leq
  2^{\Asize}$ applications of $\piesup(-)$ are necessary to obtain $\pies(\Phi)$.
  \tcnfold{Let $\Phi''$ be any formula obtained from $\Phi^0$ after some
  applications of the function $\piesup(-)$ and let $x$ be a
  reducible variable for some set $L$ of literals using
  variables in $A_\Phi$ and $\Phi''$ that is witnessed by the formulas
  $\phi,\phi' \in \Phi''$ satisfying (P1)--(P3).
  Recall that for every $\phi_L \in \Phi''_L$, the formula
  $\piesup(L,x,\Phi'')$ contains the two disjuncts
  $\phi_L^\top=(\cnfprop{\phi_L[x= 1]}  \land x \land
  a_{t})$ and $\phi_L^\bot=(\cnfprop{\phi_L[x= 0]}  \land
  \lnot x \land
  \lnot a_{t})$, where $t \in [\Asize]$ is the smallest index such
  that $L$ does not contain a literal over $a_t$. Then,
  $f(\phi_L)\leq \min\{f(\phi_L^\top),f(\phi_L^\bot)\}$. Moreover,
  because of (P3) $C(\phi,A_\Phi)\cap C(\phi,X_q)$ contains a clause,
  say $l \lor x$,
  where $x$ occurs positively. Because $\phi$ is propagated, it does not contain the unit clause $l$; otherwise
  $l \lor x$ would no longer be contained in $\phi$. Therefore,
  $\cnfprop{\phi[x= 0]}  \land \lnot x \land
  \lnot a_{t}$ contains at least one more unit clause from $X_q$ than
  $\phi$, which because of (P2) shows that
  $f(\cnfprop{\phi[x= 0]}  \land \lnot x \land
  \lnot a_{t})>f(\phi)$. A similar argument shows that also
  $f(\cnfprop{\phi'[x= 1]}  \land x \land
  a_{t})>f(\phi')$. Therefore,
  $f(\Phi'''_{L\cup\{a_t\}})>f(\Phi''_L)$ and $f(\Phi'''_{L\cup\{\lnot
    a_t\}})>f(\Phi''_L)$, where $\Phi'''=\piesup(L,x,\Phi'')$. Since,
  $f(\phi^*)\leq 2k$ for any \CNF formula $\phi^*$ and because
  of \Cref{lem:pies_A_assigment} also
  $|\Phi^*_{L^*}|\leq k$ for any $\Phi^*$ obtained from $\Phi'$ after
  any number of applications of the function $\piesup(-)$ and any set
  $L^*$ of literals using variables in $A_\Phi$, we obtain that
  $f(\Phi^*_{L^*})\leq 2k^2$. Therefore, at most $2k^2$ literals from
  $A_\Phi$ can be added to any formula $\phi \in \Phi'$, which
  concludes the proof of the lemma.}
\end{proof}

The next lemma now shows that $\pies(\Phi)$ and $\Phi$ are equivalent.
    
\begin{lemma}\label[lemma]{lem:phi_pies}
  Let $\Phi $ be a propagated $k$-\textsc{Disjunct-QBF}(\TCNF), whose
  innermost quantifier block is universal.
  Then, $\Phi \Leftrightarrow \pies(\Phi)$.
\end{lemma}
    
\begin{proof}
  The equivalence of $\Phi$ and $\del(\Phi)$ follows
  from \Cref{lem:bqa-dog-eq}. Therefore, it suffices to show that
  $\piesup(\Phi^i,L,x)$ is equivalent with $\Phi^i$ whenever $x$ is reducible with respect
  to $L$ and $\Phi^i$.
  
  Towards showing the forward direction of the equivalence,
  let $T$ be an existential winning strategy for
  $\piesup(\Phi^i,L,x)$ and let $l$ be any leaf of $T$. Then, $\tau_l^T$
  satisfies some disjunct of $\piesup(\Phi^i,L,x)$ and therefore also
  some disjunct of $\Phi^i$, as required.
  
  Towards showing the reverse direction of the equivalence,
  let $T$ be an existential winning strategy for $\Phi^i$. Informally, we
  obtain an existential winning strategy $T'$ for $\piesup(\Phi^i,L,x)$ by
  assigning $a_{|L|+1}$ the same value as $x$, i.e., whenever $T$ assigns
  $x$ to $0$ ($1$), then $T'$ also assigns $a_{|L|+1}$ to $0$ ($1$); note
  that this can be achieved since all variables in $A_\Phi\cup
  X_{q-1}$ are within the same quantifier block of
  $\piesup(\Phi^i,L,x)$. To see that $T'$ is a winning strategy for
  $\piesup(\Phi^i,L,x)$ consider any leaf $l$ of $T$ and let $\phi
  \in \Phi^i$ be the disjunct of $\Phi^i$ satisfied by $\tau_l^T$. If
  $\phi \in \piesup(\Phi^i,L,x)$, then $\phi$ is also satisfied at
  the leaf $l'$ corresponding to $l$ in $T'$. Otherwise, $\phi \in
  \Phi_L^i$, but then $l'$ either satisfies the disjunct $(\phi[x= 1]  \land x \land
  a_{t})$ (if $\tau_{l'}^{T'}(x)=\tau_{l'}^{T'}(a_{|L|+1})=1$) or the disjunct
  $(\phi[x= 0]  \land \lnot x \land
  \lnot a_{t})$ (if $\tau_{l'}^{T'}(x)=\tau_{l'}^{T'}(a_{|L|+1})=0$).
  
  From \Cref{lem:bqa-dog-prop},  $\del(\Phi)$ is in 
  propagated form. Furthermore, if $\Phi^i$ is propagated, then 
  $\piesup(\Phi^i,L, x)$ is also propagated, as the $\cnfprop{-}$ 
  function is applied whenever a new disjunct is added. Therefore, 
  $\pies(\Phi)$ is propagated.
\end{proof}
      
The next lemma demonstrates that 
$\pies(\Phi)$ can be efficiently computed.

\begin{lemma}\label[lemma]{lem:computecat}
  Let $\Phi$ be a propagated $k$-\textsc{Disjunct-QBF}(\TCNF), whose
  innermost quantifier block is universal. 
  Then, $\pies(\Phi)$ can be computed in time  
  $\bigoh(2^{2k^2 + k}k\propT{|V(\Phi)|})$.
\end{lemma}
\begin{proof}
  We begin by noting that the function $\del(\Phi)$ can be computed in 
  time $\bigoh(k2^k\propT{|V(\Phi)|})$ due to \Cref{lem:bqa-dog-prop}.
  To compute $\pies(\Phi)$ from $\del(\Phi)$, we apply the function 
  $\piesup(-)$ at most $2^{|A_\Phi|}=2^{\Asize}$ times,
  because of \Cref{lem:tcnf_pies_limit}.
  
  Once a 
  reducible variable $x$ is identified with respect to some set $L$ of 
  literals over $A_\Phi$ and the current formula $\Phi^c$, we compute the 
  difference between $\piesup(\Phi^c,L, x)$ and $\Phi^c$ in 
  $\bigoh(k\propT{|V(\Phi)|})$ time using \Cref{lem:prop_comp}. Thus, updating 
  $\Phi^c$ to $\piesup(\Phi^c,L, x)$ takes
  $\bigoh(k\propT{|V(\Phi)|})$ time.
  
  For a given set $L$ of literals over $A_\Phi$, we can check in 
  $\bigoh(\sizeOf{\Phi^c_L}) = \bigoh(k|V(\Phi)|^2)$ whether there exists 
  a reducible variable $x \in V(\Phi)$. If $x$ is reducible in 
  $\Phi_L^c$, it remains reducible for any $\Phi'$ such that 
  $\Phi_L^c \subseteq \Phi'$. Therefore, after calculating $\del(\Phi)$, 
  we can initialize a stack $S$. For each set of literals $L$ over 
  $A_\Phi$, we check whether there exists a reducible variable 
  $x \in V(\Phi)$ for $L$ and $\del(\Phi)$, adding the pair $(L, x)$ to 
  $S$ if so.
  
  If $S$ is not empty, we pop a pair $(L, x)$ from $S$ and apply 
  $\piesup(\Phi^c,L, x)$ to the current $\Phi^c$. Afterward, we check 
  whether the new disjuncts introduce additional reducible variables, 
  adding any such pairs to $S$. This approach requires 
  $\bigoh(2^{2k^2 + k}k|V(\Phi)|^2)$ time to identify reducible variables 
  throughout the computation.
  
  Combining these steps, the total runtime to compute $\pies(\Phi)$ is 
  $\bigoh(2^{2k^2 + k}k\propT{|V(\Phi)|})$.
\end{proof}

Our next aim is to show that we can assume that every disjunct of 
$\pies(\Phi)$ contains fewer than $2k$ unit clauses involving variables 
from the last (universal) quantifier block.
This will allow us, to use property (A3') instead of (A3), which was
crucial to bound the number of iterations used to obtain $\pies(\Phi)$.  
\begin{lemma}\label[lemma]{lem:tcnf_pies_redundant}
  Let $\Phi$ be a propagated $k$-\textsc{Disjunct-QBF}(\TCNF) 
  whose innermost quantifier block $X_q$ is universal and
  let $\Phi'$ be the subset of $\pies(\Phi)$ that contains
  all disjuncts with at least $2k$ unit clauses using variables in $X_q$.
  Then, $\Phi'$ is $\pies(\Phi)$-redundant.
\end{lemma}
    \begin{proof}
      It suffices to show that $(\lnot \pies(\Phi)) \Leftrightarrow (\lnot (\pies(\Phi)\setminus \Phi'))$. Since
      the forward direction is trivial, it only remains to show that
      if there is a universal winning strategy say $T$ for $\pies(\Phi)\setminus
      \Phi'$, then there is a universal winning strategy for $\pies(\Phi)$.
      Consider a play $\alpha \in \ass(T)$.
      Because of \Cref{lem:pies_A_assigment}, all but at most $k$
      disjuncts of $\pies(\Phi)$ are satisfied by the assignment $\alpha_{A_\Phi}$
      obtained by restricting $\alpha$ to the variables in
      $A_\Phi$. Let $\Phi''$ be the set of all those disjuncts that are
      not falsified by $\alpha_{A_\Phi}$.
      Because $T$ is a universal
      winning strategy for $\Phi''\setminus \Phi'$, there is a clause, say
      $c_\phi$ that is falsified by $\alpha$ for every $\phi \in
      \Phi''\setminus \Phi'$. Then, because $|c_\phi|\leq 2$ for every $\phi
      \in \Phi''\setminus \Phi'$, it holds that the number of variables
      in $X_q$ occurring in unit clauses of $\phi$ is at least $|\Phi'|$ for every
      $\phi \in\Phi'$. Therefore, we can choose one variable,
      say $v_\phi$, from $X_q$ for every $\phi \in \Phi'$ such that:
      $\phi$ contains a unit clause on $v_\phi$, $|\SB v_\phi \SM \phi \in
      \Phi'\SE|=|\Phi'|$, and $\SB v_\phi \SM \phi \in
      \Phi'\SE\cap \SB V(c_\phi)\SM \phi \in \Phi''\setminus \Phi'\SE=\emptyset$.
      Let $\alpha'$ be the assignment obtained from $\alpha$ by changing
      only the assignment of the variables in $\SB v_\phi\SM \phi \in
      \Phi'\SE$ such that the unit clause on $v_\phi$ is falsified for
      every $\phi \in \Phi'$. Then, $\alpha'$ is winning for the universal
      player on $\pies(\Phi)$. By doing the same modification for every play in
      $\ass(T)$, we obtain a winning strategy for the universal player on $\pies(\Phi)$.
    \end{proof}

The following is now the main technical lemma of this subsection,
showing that the variables of the last existential quantifier block of
the given formula $\Phi$ are redundant in $\pies(\Phi)$ and can
therefore be replaced by the boundedly many variables in $A_\Phi$.

\begin{lemma}\label[lemma]{lem:bqa-var-redundant}
  Let $\Phi$ be a propagated $k$-\textsc{Disjunct-QBF}(\TCNF)
  with prefix
  $\pref = Q_1 X_1 \dots Q_q X_q$ and $Q_q=\forall$.
  Let $\Phi'$ be the subset of $\pies(\Phi)$ containing all
  disjuncts with at least $2k$ unit clauses over variables in $X_q$.
  Then, the set $X_{q-1}$ is $\pies(\Phi)\setminus \Phi'$-redundant.
\end{lemma}
\begin{proof}
For each $X \subseteq X_{q-1}$ ($\overline{X}=X_{q-1}\setminus X$), we define  
$$
\Phi^{X} = Q_1 X_1 \dots \exists(A_\Phi \cup X)\forall X_q. 
\bigvee_{\phi \in \pies(\Phi) \setminus \Phi'} 
\left(\phi \setminus C(\phi, \overline{X})\right).
$$

Moreover, for a \CNF formula $\phi$ and a literal $l$, we denote by
$\CL(\phi,l)$ the set of all clauses in $\phi$ that contain literal $l$.

We start by showing the following.
Let $L$ be any set of literals from $A_\Phi$ with $\Phi_L\neq
\emptyset$ and let $X \subseteq X_{q-1}$. Then, for every $x \in X$ 
there is a literal $l_x \in \{x, \neg x\}$ such that:
\begin{itemize}
    \item[(A)] $V(\CL(\Phi^{X}_L, l_x)) \subseteq V(\Phi^X) 
      \setminus (X_q \cup A_\Phi)$,
    \item[(B)] For every $\phi, \phi' \in \Phi^X_L$,  
    $\phi \setminus C(\phi, X_q) = \phi' \setminus C(\phi',
    X_q)$
  \item[(C)] For each $\phi \in \Phi^X_L$, $\phi$ is
    closed under resolution with respect to $x$.
\end{itemize}

(A) follows because otherwise the variable $x$ would be reducible with respect
to $L$ and $\pies(\Phi)$, which is not possible since the function
$\piesup(-)$ has be exhaustively applied to $\pies(\Phi)$.

(B) follows immediately from \Cref{lem:bqa-dog-prop} together with the
observation that the function $\piesup(-)$ modifies all disjuncts
$\phi$ with $C(\phi,A_\Phi)=L$ in the same manner. 

Finally, (C) follows because every $\phi \in \Phi^X_L$ is
propagated, in particular, contains all resolutions obtained between
clauses of size at most $2$, together with the fact that $\phi$ is a
\TCNF.

Let $X \subset X_{q-1}$ and $x \in X_{q-1} \setminus X$.  
Now we show that $\Phi^{X \cup \{x\}}$ is equivalent to $\Phi^{X}$.

To show the forward direction, let $T$ be an existential winning 
strategy for $\Phi^{X \cup \{x\}}$. Let $l$ be any leaf node of $T$. 
Then, $\tau_l^T$ satisfies some disjunct of $\Phi^{X \cup \{x\}}$, 
and therefore also some disjunct of $\Phi^X$, as required.

For the reverse direction, let $T$ be an existential winning 
strategy for $\Phi^{X}$.  
We modify $T$ by introducing a new level for the variable $x$ between $A_\Phi \cup X$ 
and $X_q$ in the tree. 
Let $\alpha$ be an assignment to $V(\Phi^{X}) \setminus X_q$ 
corresponding to the path from the root to some node $n$ on the new level.  
Let $L$ be a set of literals in $A_\Phi$ such that 
$L[\alpha] = \top$ and $\Phi_L^{X} \neq \emptyset$. Note that due to
the construction of $\pies(\Phi)$, the set $L$ is unique.
Let $l_x$ be a literal of $x$ satisfying (A) with respect to $\Phi_L^{X \cup \{x\}}$.  

If there exists a clause $c \in \CL(\Phi_L^{X \cup \{x\}}, l_x)$ 
that is not satisfied by $\alpha$, we assign $l_x$ positively.  
Otherwise, we assign $l_x$ negatively.  
This assignment is then added to $T$ at node $n$, ensuring that the 
strategy now accounts for $x$. Let $T'$ be the existential strategy
for $\Phi^{X\cup \{x\}}$ obtained in this manner. It remains to show
that $T'$ is winning.

Let $l$ be any leaf of $T'$. Because $T'$ is an existential winning
strategy for $\Phi^X$, there is a $\phi \in \Phi_L^{X \cup \{x\}}$
such that $(\phi \setminus C(\phi, \{x\}))$ is satisfied by
$\tau_l^{T'}$ and it remains to show that $\tau_l^{T'}$ also satisfies
$C(\phi, \{x\})$. Let $l_x$ be the literal of $x$ chosen to determine
the assignment for $x$ in the ancestor of $l$ in $T'$. If $l_x$ is
assigned positively in $\tau_l^{T'}$, then already all clauses in
$\CL(\phi, l_x)$ are satisfied by $\tau_l^{T'}$. Moreover, 
there is a clause
$c \in \CL(\phi, l_x)$ such that $c\setminus \{l_x\}$ is falsified
by $\tau_l^{T'}$. It remains to show that also all clauses in 
$\CL(\phi, \lnot l_x)$ are satisfied by $\tau_l^{T'}$. Let $c'$ be
any such clause. Then, because of (C), we have that the resolvent
$c_r$ of $c$ and $c'$ is contained in $\phi$ and even in
$\phi\setminus C(\phi,\{x\})$. Therefore, $c_r$ is satisfied by
$\tau_l^{T'}$, which implies that $\tau_l^{T'}$ also satisfies $c'$;
this is because $c\setminus \{l_x\}$ is falsified by $\tau_l^{T'}$.

If, on the other hand, $l_x$ is assigned negatively by $\tau_l^{T'}$, then all clauses in 
$\CL(\phi, l_x)$ are already satisfied by $\tau_l^{T'}$, and all clauses in 
$\CL(\phi, \neg l_x)$ are satisfied by the assignment to $x$.
This shows that $\Phi^{X}$ and $\Phi^{X\cup \{x\}}$ are
equisatisfiable and by repeating this argument for all variables in
$X_{q-1}$, we obtain that $\Phi^{X_{q-1}} \Leftrightarrow \Phi^\emptyset$, 
which completes the proof.
\end{proof}

We are now ready to put everything together and proof \Cref{the:tcnf-tract}.
\begin{proof}[Proof of \Cref{the:tcnf-tract}]
  Using \Cref{lem:prop_comp}, we can assume that $\Phi$ is propagated. Moreover, if
  $Q_q=\exists$, we can use \Cref{lem:prop_2cnf} to remove the
  last existential quantifier block. Therefore, we can assume that
  $Q_q=\forall$. To decide the validity of $\Phi$, we start by 
  iteratively removing existential quantifier blocks starting from the right most
  existential quantifier block. In particular, we iteratively compute 
  $k_i$-\textsc{Disjunct-QBF}(\TCNF) formulas $\Phi_i$ for every $i
  \in [\ell]\cup \{0\}$ for some integer $k_i$, where $\ell$ is the number of
  existential quantifier blocks of $\Phi$. Importantly, $\Phi_i$ will
  only have $i$ quantifier blocks but will still be equivalent
  with $\Phi$. To compute the formulas $\Phi_i$, we start by setting
  $\Phi_\ell=\Phi$ and $k_\ell=k$. Moreover, for every $i\in [\ell]$, we
  obtain $\Phi_{i-1}$ and $k_{i-1}$ from
  $\Phi_{i}$ and $k_{i}$ as follows. We first use \Cref{lem:computecat} to
  compute $\pies(\Phi_{i})$ in time
  % $2^{\bigoh(k_{i}^2)}\propT{|V(\Phi_i)|}$ 
  $\bigoh(k_i2^{2k_i + k_i}\propT{|V(\Phi_i)|})$ 
  from
  $\Phi_{i}$. Note that $\pies(\Phi_{i})$ is equivalent with
  $\Phi_{i}$ and $|\pies(\Phi_{i})|\leq
  k_{i}2^{2k_{i}^2+k_{i}}$ because
  of \Cref{lem:phi_pies}. We then obtain the formula $\Phi_{i-1}'$ from
  $\pies(\Phi_{i})$ by removing all formulas $\phi \in
  \pies(\Phi_{i})$ that have at least $2k_i$ unit clauses over
  variables in the last (universal) quantifier block of $\pies(\Phi_{i})$. Note that this can be
  achieved in time 
  $\bigoh(\sizeOf{\pies(\Phi_{i})})= \bigoh(k_i2^{2k_i + k_i}|V(\Phi_i)|^2)$ 
  and because
  of \Cref{lem:tcnf_pies_redundant} the resulting formula $\Phi_{i-1}'$ is
  equivalent with $\Phi_{i}$.
    Using \Cref{lem:bqa-var-redundant}, we can remove all variables 
    in the innermost existential quantifier block of $\Phi_{i}$ 
    from $\Phi_{i-1}'$. This yields the equivalent formula 
    $\Phi_{i}''$, where the innermost existential quantifier block 
    now contains at most $2k_{i}^2 + k_{i}$ variables from 
    $A_{\Phi_{i}}$.
  We then
  obtain $\Phi_{i-1}$ by eliminating all variables in $A_{\Phi_{i}}$
  from $\Phi_{i-1}''$ using ordinary quantifier elimination.
  Note that $\Phi_{i-1}''$ still has property from \Cref{lem:pies_A_assigment}
  that for each assignment $\alpha$ of $A_{\Phi_{i}}$
  there exist at most $k_i$ disjuncts from $\Phi_{i-1}''$
  that are not unsatisfied by $\alpha$.
  This property give us that $|\Phi_{i-1}| \leq k_i2^{2k_i^2 + k_i} = k_{i-1}$,
  and we can  compute $\Phi_{i-1}$ in
  time $\bigoh(k_{i-1}\sizeOf{\Phi_{i}''}) = \bigoh(k_{i}|V(\Phi_i)|^2)$,
  because $\Phi_{i}$, $\Phi_{i-1}'$, $\Phi_{i-1}''$ and $\Phi_{i-1}$ are propagated. 
  Note furthermore,
  $|V(\Phi_{i-1})| \leq k_{i-1}|V(\Phi_{i})|$,
  $\Phi_{i-1}$ is equivalent with
  $\Phi_{i}$, and the number of existential quantifier blocks of
  $\Phi_{i-1}$ is one less than the number of existential quantifier
  blocks of $\Phi_{i}$.
  Putting everything together the time required for the computation of
  $\Phi_{i-1}$ from $\Phi_{i}$ is dominated by the time required to
  compute $\pies(\Phi_{i})$ and is therefore at most
  $\bigoh(k_i2^{2k_i^2 + k_i}\propT{|V(\Phi_i)|}) = \bigoh(k_{i-1}\propT{|V(\Phi_i)|}) 
  = \bigoh(k_{i-1} \propT{(\Pi_{i <j\leq \ell} k_i)}\propT{|V(\Phi)|})$.

  Next we analyze the time required to compute $\Phi_0$. Towards
  this aim, we start by obtaining an upper bound on $k_i$. Since
  $b^{(b^a)^2} = b^{(2b)^a}$,
  $k_\ell=k$ and 
  $k_{i-1}=k_{i}2^{2k_{i}^2+k_{i}}=2^{2k_{i}^2+k_{i}+\log_2(k_{i-1})}\leq
  2^{3k_{i}} = 8^{k_i^2}$,
  we obtain that 
  $k_i\leq 8^{\exp_{16}^{\ell-i}(k^2)}$, for $i \in [\ell] \cup \{0\}$.
  The total runtime to compute $\Phi_0$ from $\Phi_\ell$
  as
  $\bigoh(\sum_{0\leq i \leq \ell} k_{i-1} \propT{(\Pi_{i <j\leq \ell} k_j)}\propT{|V(\Phi)|})=
  \bigoh(k_{0} \propT{(\Pi_{0 <j\leq \ell} k_j)}\propT{|V(\Phi)|})
  =\bigoh(k_0^2\propT{|\Phi|}) = \bigoh(\exp_{16}^{\ell}(k^2)\propT{|\Phi|})$.
  
It remains to determine the validity of $\Phi_0$. For this, we 
use~\cite[Corollary 7]{DBLP:conf/ijcai/ErikssonLOOPR24}, which shows 
that any $k$-\textsc{Disjunct-QBF}(\TCNF) $\Phi'$ without existentially 
quantified variables can be solved in time 
$(2k)^{\bigoh(k)}|V(\Phi')|$. Applying this result, we conclude that 
$\Phi_0$ can be solved in time 
$
(2k_0)^{\bigoh(k_0)}|V(\Phi_0)| = 
(2k_0)^{\bigoh(k_0)}\left(\prod_{0 < j \leq \ell} k_j\right)|V(\Phi)| = 
2^{\bigoh(k_0 \log k_0)}|V(\Phi)|
$.
Both the construction of $\Phi_0$ and solving $\Phi_0$ admit an FPT-time parameterized by $k + \ell$. Thus, the overall runtime of the 
algorithm remains fixed-parameter tractable.
\end{proof}

\subsection{Algorithm for Affine}
\label{ssec:affine}

Here we show that deciding the validity of a
$k$-\textsc{Disjunct-QBF}(\AFF) formula is fixed-parameter tractable
parameterized by $k+q$, where $q$ is the number of quantifier
alternations. The overall strategy is to show that one can always remove
the innermost quantifier block without increasing the number of disjuncts
by too much. To achieve this, we first bring every disjunct of the
formula into reduced echelon form, a well-known normal form for affine
formulas obtained by Gaussian elimination. We then observe that if the
innermost quantifier block is existential, we can remove it together with
all equations that contain variables from said block since those
equations can always be simultaneously satisfied by the existential
player (here we use that the formula is in reduced echelon form and
therefore all equations are linearly independent). The main challenge,
however, is if the innermost quantifier block is universal. In this case,
we first observe that any disjunct that has at least $k$ (linearly
independent) equations, each using at least one variable from the last
universal quantifier block, can be safely removed
since it can always be falsified by the universal player. This leaves
us with at most $k$ disjuncts each containing at most $k-1$ equations
with variables from the innermost quantifier block. This now already
implies that the number of different variable types (given by the
subset of these $k(k-1)$ equations in which the variable occurs in) within the last
universal quantifier block is at most $2^{k(k-1)}$ and it is
straightforward to verify that keeping only one variable from each of
these types is sufficient to preserve
equivalence. Using this together with traditional quantifier
elimination already allows one to eliminate the innermost quantifier block,
while obtaining at most $2^{2^{k(k-1)}}$ disjuncts. We show, however,
and this is the main technical challenge and contribution of our
algorithm,
  which we give in~\Cref{lem:aff-red-forall}, that one can
reduce this to at most $2^{k^2}$ disjuncts by carefully analyzing
the formula resulting after employing traditional quantifier
elimination of all variables in the last universal quantifier block.

\iflong
We
start by defining the normal form for $k$-\textsc{Disjunct-QBF}(\AFF)
formulas that can be obtained by Gaussian elimination of every disjunct.

Let $\Phi$ be a $k$-\textsc{Disjunct-QBF}(\AFF) formula and $\phi \in
\Phi$. We define the \emph{clause-variable matrix of $\phi$}, denoted
by $M_\phi$, as follows.
$M_\phi$ has one column for
every variable in $V(\phi)$ and one additional column representing
the right-hand side of any equation; we assume that the columns are
ordered such that $i$-th column corresponds to the $i$-th variable
from the right in the quantifier prefix. Moreover, the column
representing the right-hand side of any equation is the last
column of $M_\phi$. Finally, $M_\phi$ has a row for every equation
$(A,b) \in \phi$ that is $1$ at every column corresponding to a
variable in $A$, $b$ in the last column of $M_\phi$, and $0$ at every other column.

Let $M$ be a $\{0,1\}^{m\times n}$ matrix. The \emph{leading
  coefficient} of a row $r$ of $M$, is the left-most non-zero entry of
$r$. We say that a matrix $M$ is in \emph{echelon
  form}, if $M$ has every all-zero row 
at the bottom and moreover the leading coefficient of any non-zero row $r$ of $M$
that is below another non-zero row $r'$ in $M$ is to the right of the leading
coefficient of $r'$. We say that $M$ is in \emph{reduced echelon form}
if it is in echelon form and additionally every column that contains the
leading coefficient of some row is zero at all other entries.
Additionally, we say that
$\phi \in \Phi$ is in \emph{reduced echelon form} if the matrix $M$ describing $\phi$ is in reduced echelon form.
We say that $\Phi$ is in \emph{reduced echelon form} if so is every $\phi
\in \Phi$.

\begin{lemma}\label[lemma]{lem:GE}
  Let $\phi$ be an affine formula, then there is an algorithm that in
  time $\bigoh(|\phi||V(\phi)|^2)$ decides whether the system of equations
  represented by $\phi$ is satisfiable. Moreover, if it is satisfiable
  the algorithm also returns the (unique) matrix in reduced echelon
  form obtained from $M_\phi$.
\end{lemma}
\begin{proof}
  This can be achieved via standard Gaussian elimination (see e.g.~\cite{farebrother2018linear}).
\end{proof}

The following lemma allows us to assume that we are given a
$k$-\textsc{Disjunct-QBF}(\AFF) formula in reduced echelon form.
\begin{lemma}\label[lemma]{lemma:triangular}
  Let $\Phi$ be a $k$-\textsc{Disjunct-QBF}(\AFF) formula. There is an
  algorithm that in time $\bigoh((\sum_{\phi \in \Phi}|\phi|)|V(\Phi)|^2)$ computes an equivalent
  $k$-\textsc{Disjunct-QBF}(\AFF) formula $\Phi'$ that has the same
  prefix as $\Phi$ and is in reduced echelon form.
\end{lemma}
\begin{proof}
  We obtain $\Phi'$ from $\Phi$ by doing the following for every $\phi
  \in \Phi$. First we apply \Cref{lem:GE} to either decide whether the system of equation
  corresponding to $\phi$ is satisfiable. If it is not, we simply
  remove $\phi$ from $\Phi$ and otherwise let $M_\phi'$ be the unique
  matrix in reduced echelon form obtained from $M_\phi$
  (using \Cref{lem:GE}). We now replace
  $\phi$ with the formula $\phi'$ that consists of the set of all equations defined by $M_\phi'$, i.e., every
  row $r$ of $M_\phi'$ provides the equation $(A,b)$, where $A$ contains
  all variables corresponding to non-zero columns in $r$ and $b$ is
  equal to the right-most column of $M_\phi'$. Because $M_\phi'$ is in
  reduced echelon form, it holds that 
  $\phi'$ is in reduced echelon form. Repeating this for
  every $\phi \in \Phi$, provides the required
  $k'$-\textsc{Disjunct-QBF}(\AFF) formula $\Phi'$ with $k'\leq k$
  in time $\bigoh((\sum_{\phi\in \Phi}|\phi|)|V(\Phi)|^2)$.
\end{proof}

We say that a set $S=\SB (A_i,b_i)\SM 1 \leq i \leq \ell\SE$ of equations is \emph{linearly
  independent} if no set $A_i$ can be expressed as a XOR of the sets in
$\SB A_j\SM 1 \leq j \leq \ell \text{ and } j\neq i\SE$, i.e., $A_i=\xor_{j
  \in J}A_j$ for some subset $J \subseteq [\ell]\setminus \{i\}$. We
obtain the following simple observation.
\begin{observation}\label[observation]{observation:affine}
  Let $\Phi$ be a $k$-\textsc{Disjunct-QBF}(\AFF) in reduced echelon form
  and $\phi \in \Phi$. Then, any subset $E$ of equations of $\phi$ is
  linearly independent and this still holds if all equations in $E$
  are projected to some suffix $S$ of the prefix of $\Phi$ that
  contains at least $1$ variable from every equation in $S$, 
  i.e.,
  the set $\SB (e\cap S,b) \SM (e,b) \in E\SE$ of equations is
  linearly independent if $S$ is a suffix of the prefix with $e\cap
  S\neq \emptyset$ for every $(e,b)\in E$.
\end{observation}

We also need the following simple technical lemma that allows us to
use the linear independence of equations to adapt winning strategies
for the universal player.
\begin{lemma}\label[lemma]{lem:affine-lind}
  Let $S$ be a system of equations. If all equations in $S$ are
  linearly independent, then $S$ has a solution. Moreover, if there is
  an assignment that falsifies all equations in $S$ and $C$ is an
  equation that is linearly independent from the equations in $S$,
  then there is an assignment that falsifies all equations in $S \cup \{C\}$.
\end{lemma}
\begin{proof}
  It is well-known that a system of linear equations has a solution if
  the vectors corresponding  the left-hand side are linearly
  independent. This already shows the first statement of the
  lemma. Towards showing the second statement of the lemma, first note
  that an assignment falsifies all equations in $S$ if and only if the
  same assignment satisfies all equations in $\overline{S}=\SB (A,1-b) \SM (A,b)
  \in S \SE$. Therefore, it suffices to show that if there is an
  assignment satisfying $\overline{S}$, then there is an assignment
  satisfying $\overline{S}\cup \{\overline{C}\}$, where
  $\overline{C}=(A,1-b)$ and $(A,b)=C$. It is
  well known that because $\overline{S}$ is satisfiable, it contains a subset $\overline{S}'$ of
  linearly independent equations that is equivalent with
  $\overline{S}$. Moreover, it follows from the first statement of the
  lemma that $\overline{S}'\cup \{\overline{C}\}$ has a solution, which
  implies that so does $\overline{S}\cup\{\overline{C}\}$.
\end{proof}

Using the properties of formulas in reduced echelon form, we can now show that we can simply remove the innermost quantifier
block (together with all equations containing variables from the last
quantifier block) if it is existential.
\begin{lemma}\label[lemma]{lem:aff-red-exists}
    Let $\Phi = \pref . \bigvee_{i \in [k]} \phi_i$ be a
    $k$-\textsc{Disjunct-QBF}(\AFF) in reduced echelon form, whose innermost
    quantifier block $X_q$ is existential.
    Then, $X_q$ is $\Phi$-redundant.
\end{lemma}
\begin{proof}
  Let $\Phi'=\pref'.\bigvee_{i \in [k]}\phi'_i$, where $\pref'$ is equal to $\pref$ minus the innermost quantifier block and
  for every $\phi_i \in \Phi$, $\phi'_i$ is equal to $\phi_i\setminus
  C(\phi_i,X_q)$. Clearly, $\Phi \models \Phi'$ and
  it therefore only remains to show that $\Phi' \Rightarrow \Phi$. Let $T$ be a winning strategy
  for the existential player on $\Phi'$ and let $\gamma \in
  \ass(T)$. Then, there is a $\phi' \in \Phi'$ that is satisfied by
  $\gamma$. Consider the set $S$ of equations in $\phi$ that contain at
  least one variable in $X_q$ and let $S'=\SB C[\gamma] \SM C \in
  S\SE$. Since $\Phi$ is in reduced echelon form all equations in $S'$ are
  linearly independent and it therefore follows from
  Lemma~\ref{lem:affine-lind} that there is an assignment $\gamma' :
  X_q \rightarrow \{0,1\}$ that satisfies all equations in
  $S'$. Therefore, $\gamma\cup \gamma'$ satisfies $\phi$ and shows
  that $T$ can be extended to a winning strategy for the
  existential player on $\Phi$.
\end{proof}

Our next step is to show that we can also deal with the case that the
innermost quantifier block is universal. Towards this aim,
we start by showing the following auxiliary lemma, which informally allows us to assume that every disjunct of
$k$-\textsc{Disjunct-QBF}(\AFF) in reduced echelon form has fewer than $k$
equations with variables from $X_q$; this is because any disjunct with
at least $k$ of such equations can always be falsified by the
universal player.
\begin{lemma}\label[lemma]{lem:affRedundantDisjunct}
  Let $\Phi$ be a $k$-\textsc{Disjunct-QBF}(\AFF) in reduced echelon form,
  whose right-most quantifier block $X_q$ is universal and let $\phi \in \Phi$ such
  that $\phi$ contains at least $k$ equations that contain a variable from
  $X_q$.
  Then, $\phi$ is $\Phi$-redundant.
\end{lemma}
\begin{proof}
  Clearly, $(\neg \Phi) \models (\neg (\Phi\setminus\{\phi\}))$, so it
  suffices to show that $(\neg (\Phi\setminus\{\phi\})) \models (\neg
  \Phi)$.
  Let $T$ be a winning strategy for the universal player on
  $\Phi\setminus\{\phi\}$ and consider any play $\gamma \in
  \ass(T)$. Then, for every $\phi' \in \Phi\setminus\{\phi\}$ there is
  (at least) one equation, say $C_{\phi'}$, that is not satisfied by
  $\gamma$. Because $\phi$ is in reduced echelon form, it follows
  that all equations in $\phi-(V(\phi)\setminus X_q)$ are linearly
  independent. Therefore, $\phi  - (V(\phi)\setminus X_q)$ contains at least $k$ linearly
  independent equations, which implies that there is an equation $C$ in $\phi-
  (V(\phi)\setminus X_q)$ that
  is linearly independent from the equations in $S=\SB C_{\phi'}[\gamma_{\overline{X_q}}]
  \SM \phi' \in
  \Phi\setminus \{\phi\}\SE$, where $\gamma_{\overline{X_q}}$ is the assignment
  $\gamma$ restricted to the variables in $V(\Phi)\setminus X_q$. Therefore, it follows from
  Lemma~\ref{lem:affine-lind} that there is an assignment
  $\gamma_{X_q}: X_q \rightarrow \{0,1\}$ that simultaneously
  falsifies all equations in $S \cup \{C\}$, which implies 
  that the play $\gamma_{\overline{X_q}}\cup \gamma_{X_q}$
  is winning for the universal player on $\Phi$. Since
  the play $\gamma_{\overline{X_q}}\cup \gamma_{X_q}$ only differs
  from the original play $\gamma$ on the variables in $X_q$, we obtain
  a winning strategy for the universal player on $\Phi$ by iterating
  the same argument for every play of $T$.
\end{proof}
The next lemma is the main technical contribution of our algorithm for
$k$-\textsc{Disjunct-QBF}(\AFF) and essentially allows us to eliminate
the last universal quantifier block of a formula.

\begin{lemma}\label[lemma]{lem:aff-red-forall}
  Let $\Phi  = \pref.\bigvee_{i \in [k]}\phi_i$ be a
  $k$-\textsc{Disjunct-QBF}(\AFF) in reduced echelon form, whose last
  quantifier block $X_q$ is universal, such that every
  $\phi \in \Phi$ contains at most $k-1$ clauses that contain a
  variable in $X_q$. Then, a $(2^{(k-1)}+1)^k$-\textsc{Disjunct-QBF}(\AFF)
  formula $\Phi'$ 
  in reduced echelon form 
  that is equivalent with $\Phi$ can be computed from
  $\Phi$ in time 
  $\bigoh(2^{2k^2}|X_q|^2+2^{k^2}|V(\Phi)|^3k)$.
  % $\bigoh(2^{k^2}1.5^k|X_q||V(\Phi)|^2)$.
  Moreover, $\sizeOf{\Phi'}\leq
  (2^{(k-1)}+1)^k\sizeOf{\Phi}$ and $\Phi'$ is obtained
  from $\Phi$ after quantifier eliminating all variables in $X_q$ and
  removing some (redundant) disjuncts.
\end{lemma}
\begin{proof}
  Let $\pref'$ be obtained from $\pref$ after removing the right-most
  quantifier block. After applying quantifier elimination of the variables in $X_q$ to
  $\Phi$, we obtain the formula:
  \[ \Phi'=\pref'.\bigwedge_{\tau : X_q \rightarrow \{0,1\}}(\bigvee_{\phi
      \in \Phi}\phi[\tau])
  \]

  Using distributivity
  to rewrite $\Phi'$ into a disjunction of \textup{CNF}s, we obtain that every
  disjunct of $\Phi'$ has the form:

  \[ \bigwedge_{\tau : X_q \rightarrow \{0,1\}}\phi_{i_\tau}[\tau]\]
  
  where for every $\tau : X_q \rightarrow \{0,1\}$, it holds that
  $i_\tau \in [k]$ and $\phi_{i_\tau} \in \Phi$.
  Observe that if such a disjunct
  contains $\phi[\tau_1]$ and $\phi[\tau_2]$ for some $\phi \in \Phi$
  and $\tau_1,\tau_2 : X_q \rightarrow \{0,1\}$ such that
  $\phi[\tau_1]\neq \phi[\tau_2]$, then the disjunct can never be
  satisfied. This is because every $\phi \in \Phi$ is in triangular
  form and it therefore holds that if $\phi[\tau_1]\neq \phi[\tau_2]$,
  then there are two equations $(A_1,b_1) \in \phi[\tau_1]$ and
  $(A_2,b_2) \in \phi[\tau_2]$ such that $A_1=A_2$ but $b_1\neq
  b_2$. Therefore, we can assume that without loss of generality (after
  removing redundant copies of $\phi[\tau]$) every disjunct of
  $\Phi'$ has
  the form:
  
  \[ \bigwedge_{i \in I}\phi_{i}[\tau_i]\]

  for some $I \subseteq [k]$ with $I\neq \emptyset$ and some $\tau_i : X_q \rightarrow \{0,1\}$. 

  We say that two assignments $\tau : X_q
  \rightarrow \{0,1\}$ and $\tau' : X_q \rightarrow \{0,1\}$ are
  \emph{equivalent with respect to $\phi \in \Phi$}, denoted by $\tau
  \equiv_\phi \tau'$, if $\phi[\tau]=\phi[\tau']$. We note that
  $\equiv_\phi$ is an equivalence relation on the set of all
  assignments $\tau : X_q \rightarrow \{0,1\}$ and we denote by
  $E_\phi$ the set of all equivalence classes of $\equiv_\phi$ and by
  $E_\phi(\tau)$ the equivalence class of $\tau$ in $E_\phi$. Then,
  $E_\phi(\tau)$ is equal to the set of all satisfying assignments of
  the set $S_\phi(\tau)=\SB (A\cap X_q,(A \cap X_q)[\tau]) \SM (A,b) \in \phi \text{ and }
  A\cap X_q\neq \emptyset \SE$ of equations. Note that
  $|S_\phi(\tau)|\leq k-1$ and therefore $|E_\phi|=|\SB \phi[\tau]\SM \tau : X_q \rightarrow \{0,1\} \SE|\leq 2^{k-1}$, which
  already shows that $|\Phi'|\leq (2^{(k-1)}+1)^k$ and therefore also
  $\sizeOf{\Phi'}\leq (2^{(k-1)}+1)^k\sizeOf{\Phi}$ since every
  disjunct of $\Phi'$ has size at most $\sizeOf{\Phi}$.

  It remains to show that $\Phi'$ can be
  computed in the stated time. To show this we first observe that
  a disjunct $\bigwedge_{i \in I}\phi_{i}[\tau_i]$ is in $\Phi'$
  if and only if $\bigcup_{i \in I}E_{\phi_i}(\tau_i)=2^{X_q}$ or in
  other words if and only if $\bigvee_{i \in I}S_{\phi_i}(\tau_i)$
  is a tautology. Now, $\bigvee_{i \in I}S_{\phi_i}(\tau_i)$ is a
  tautology if and only if $\lnot (\bigvee_{i \in
    I}S_{\phi_i}(\tau_i))$ is not
  satisfiable. Using De Morgan's Law, we obtain that
  $\lnot (\bigvee_{i \in I}S_{\phi_i}(\tau_i))=\bigwedge_{i \in
    I}\bigvee_{(A,b) \in S_{\phi_i}(\tau_i)}(A,1-b)$, which using
  distributivity is satisfied if and only if the set $\SB (A_1,1-b_1),
  \dotsc,(A_k,1-b_k)\SE$ of equations is satisfiable for some choice
  of $(A_i,b_i) \in S_{\phi_i}(\tau_i)$ for every $i \in [k]$. Since
  $|S_{\phi_i}(\tau_i)|\leq k-1$, there are at most $(k-1)^k$ choices for
  the set $\SB (A_1,1-b_1), \dotsc,(A_k,1-b_k)\SE$ of equations and
  using \Cref{lem:GE} we can decide whether one such
  set of equations is satisfiable in time
  $\bigoh(k|X_q|^2)$. Therefore, we can decide whether
  the disjunct $\bigwedge_{i \in I}\phi_{i}[\tau_i]$ is in
  $\Phi'$ in time $\bigoh((k-1)^kk|X_q|^2)$. Hence, assuming that we can enumerate
  all of the at most $(2^{(k-1)}+1)^k$ distinct disjuncts in $\Phi'$ in time
  $\bigoh((2^{(k-1)}+1)^k)$, we
  obtain that we can compute $\Phi'$ in time
  $\bigoh((2^{(k-1)}+1)^k(k-1)^kk|X_q|^2)$

  Next, we show
  that we can efficiently enumerate all possible and distinct
  disjuncts $\bigwedge_{i \in I}\phi_{i}[\tau_i]$ in
  $\Phi'$; in fact it is sufficient to enumerate all choices for
  $\bigvee_{i \in I}S_{\phi_i}(\tau_i)$ that correspond to a
  disjunct in $\Phi'$. Now, $S \in S_\phi$ if and only if $S$ contains one
  equation $(A\cap X_q,b')$ for some $b' \in \{0,1\}$ for every $(A,b) \in \phi$ with $A\cap
  X_q\neq \emptyset$ and $S$ is satisfiable. Therefore, there at most
  $2^{k-1}$ choices for $S$ and for each choice we can test whether
  $S$ is satisfiable in time $\bigoh(k|X_q|^2)$
  using \Cref{lem:GE}, which shows that we can compute
  the sets $S_\phi$ for every $\phi \in \Phi$ in time
  $\bigoh(2^{k-1}|X_q|^2k^2)$. Therefore, the total runtime to compute
  all disjuncts of $\Phi'$ is
  equal to
  $\bigoh(2^{k-1}|X_q|^2k^2+(2^{(k-1)}+1)^k(k-1)^kk|X_q|^2)\subseteq\bigoh(2^{2k^2}|X_q|^2)$. Finally,
  it remains to bring $\Phi'$ into reduced echelon form. That is for each
  of the at most $(2^{(k-1)}+1)^k$ disjuncts $\phi \in \Phi'$ we
  use \Cref{lemma:triangular} to get $\phi$ into reduced echelon form in time at most
  $\bigoh(|V(\Phi)|^3k)$; this is because $\phi$ has at most
  $k|V(\Phi)|$ clauses because $\Phi$ is triangulated and at most
  $|V(\Phi)|$ variables. Therefore, we obtain 
  $\bigoh(2^{2k^2}|X_q|^2+2^{(k-1)}+1)^k|V(\Phi)|^3k)=\bigoh(2^{2k^2}|X_q|^2+2^{k^2}|V(\Phi)|^3k)$
  as the total runtime of the algorithm. 
\end{proof}
\fi
The following theorem now puts everything together into an
FPT-algorithm for deciding the validity of
$k$-\textsc{Disjunct-QBF}(\AFF) formulas (with the parameter being $k$ and the number of quantifier alternations).
\begin{theorem}\label{the:tract-aff}
  Let $\Phi  = \pref.\bigvee_{i \in [k]}\phi_i$ be a
  $k$-\textsc{Disjunct-QBF}(\AFF) with $\ell$ universal quantifier
  blocks. Then, the validity of $\Phi$ can be
  decided in time $\bigoh(\exp_4^{\ell}(k^2) n^3+mn^2)$,
  where
  $m=\sum_{\phi \in \Phi}|\phi|$ and $n=|V(\Phi)|$.
\end{theorem}
\begin{proof}
We begin by applying \Cref{lemma:triangular} to transform
$\Phi$ into reduced echelon form within a runtime of 
$\bigoh(mn^2)$. Next, we iteratively remove
blocks of existential and universal quantifiers, starting
from the right-most block, using \Cref{lem:aff-red-exists} 
and \Cref{lem:aff-red-forall}. Specifically, at each step, 
we compute a $k_i$-\textsc{Disjunct-QBF}(\AFF) formula 
$\Phi_i$ in reduced echelon form for some integer $k_i$ and 
$i \in [q] \cup \{0\}$, where $q$ is the number of quantifier blocks
of $\Phi$.

We initialize the process by setting $\Phi_q = \Phi$ and $k_q=k$ and 
derive $\Phi_{i}$ from $\Phi_{i+1}$ as follows: If 
$Q_{i+1} = \exists$, we set 
$\Phi_i = \pref_i.\bigvee_{\phi \in \Phi_{i+1}} 
\phi\setminus C(\phi,X_{i+1})$, where $\pref_i$ is the 
prefix of $\Phi_{i+1}$ obtained after removing the 
right-most quantifier block. Notably, $\Phi_i$ can be 
computed from $\Phi_{i+1}$ in time 
$\bigoh(\sizeOf{\Phi_{i+1}}) \subseteq \bigoh(k_{i+1}n^2)$ because
$\Phi_{i+1}$ is in reduced echelon form. 
Furthermore, if $\Phi_{i+1}$ is in reduced echelon form, then 
$\Phi_i$ is also in reduced echelon form, and $k_i = k_{i+1}$. 
By \Cref{lem:aff-red-exists}, $\Phi_i$ is equivalent to 
$\Phi_{i+1}$.

If $Q_{i+1} = \forall$, let $\Phi_{i+1}'$ denote the 
formula obtained from $\Phi_{i+1}$ by removing all 
$\phi \in \Phi_{i+1}$ that contain at least $k$ equations 
involving a variable from $X_{i+1}$. This operation can 
be performed in time 
$\bigoh(\sizeOf{\Phi_{i+1}}) = \bigoh(k_{i+1}n^2)$ because
$\Phi_{i+1}$ is in reduced echelon form.
By \Cref{lem:affRedundantDisjunct}, $\Phi_{i+1}'$ is 
equivalent to $\Phi_{i+1}$ and retains its reduced echelon form. 
Next, we use \Cref{lem:aff-red-forall} to obtain $\Phi_i$ 
from $\Phi_{i+1}'$. Here, $\Phi_i$ is in reduced echelon form and equivalent to 
$\Phi_{i+1}'$, and  
$|\Phi_i| \leq (2^{(|\Phi_{i+1}'|-1)}+1)^{|\Phi_{i+1}'|}$. Thus, 
$k_i \leq 2^{k_{i+1}^2}$. 
The computation of $\Phi_i$ requires time 
$\bigoh(2^{2k_{i+1}^2}|X_{i+1}|^2+2^{k_{i+1}^2}|V(\Phi)|^3k_{i+1})=\bigoh(2^{2k_{i+1}^2}n^3)$. Since
we can assume that $k_i\geq k_{i+1}$ for every $i\in [q]\cup\{0\}$, we obtain
that the overall runtime of the algorithm is given by
$\bigoh(2^{2k_{1}^2}n^3q+mn^2)$, where the $mn^2$ term is for bringing
$\Phi$ into reduced echelon form. It therefore only remains to
provide an upper bound on $k_1$.
Note that $k_i=k_{i+1}$ if $Q_{i+1}=\exists$ and $k_i\leq
2^{k_{i+1}^2}$ otherwise. As
$2^{(2^{a^{2}})^2} = 2^{4^{a^2}}$, we obtain that $k_1 \leq
2^{\exp_4^{\ell -2}(k^2)}$, which gives us
$\bigoh(2^{2(2^{\exp_4^{\ell -2}(k^2)})^2}n^3)\in
\bigoh(\exp_4^{\ell}(k^2)n^3+mn^2)$
as the total runtime of the algorithm.
\end{proof}

\section{Elimination of Guarded Universal Sets and Enhanced Backdoors}
\label{sec:enhanced}

Here, we provide the details behind our new method for the elimination
of guarded universal sets and show its implications for enhanced
(deletion) backdoors. We start by proving the correctness of the
method in \Cref{ssec:melgu} and then provide its applications for
enhanced (deletion) backdoors in \Cref{ssec:mappenh}. Finally, we show
how to efficiently find enhanced (deletion) backdoors in
\Cref{ssec:mfindenh}.
In the following, we broaden the notion of $k$-\textsc{Disjunct}\hy \QBF to disjunctions of conjunctive formulas (which include affine formulas) as opposed to merely CNF formulas.

\subsection{Elimination of Guarded Universal Sets}\label{ssec:melgu}
We now provide the full construction and formal justification  
for the algorithm outlined in \Cref{ssec:unmixed-overview}.  
The main objective of this section is to prove  
\Cref{th:new_unmixed_qbf}, which forms the foundation  
of our elimination of guarded universal sets framework.  
Before presenting the proof, we first recall key definitions  
that will be used throughout. Let $\Phi$ be a $\QBF$ formula and let $Y \subseteq
V_\forall(\Phi)$ be a subset of universal variables.
The \emph{primal graph}
$\primG(\Phi)$ of $\Phi$ is the undirected graph with a vertex for every variable in $\Phi$ and 
an edge between two variables if they occur in a common clause. 
We denote by $\delta_{\Phi}(Y)$, or simply $\delta(Y)$ if $\Phi$ is
clear from the context, the \emph{boundary (or guard)} of $Y$ for $\Phi$, i.e.,
the set of all neighbors of $Y$ in the primal graph. 
We also denote by
$\SA(\Phi,Y)$ the maximum number of variables of $Y$ in any clause $c$
of $\Phi$, i.e., $\max_{c \in \Phi}|V(c)\cap Y|$, where $V(c)$ denotes
the set of variables in clause $c$. 

We begin by restating \Cref{th:new_unmixed_disj},  
which serves as a key technical component in the proof of \Cref{th:new_unmixed_qbf}.

\thnewunmixedqbfdisj*

The proof of \Cref{th:new_unmixed_disj} is deferred  
to the end of this section, as it relies on several  
key lemmas introduced later.  
However, we are now ready to restate the main result,  
\Cref{th:new_unmixed_qbf}, and provide its proof  
based on \Cref{th:new_unmixed_disj}.

\thnewunmixedqbf*
\begin{proof}
    Since $\Phi$ is a \textsc{QBF} formula, we can write it as $\Phi = \pref.\phi$,  
    where $\pref$ is the prefix of $\Phi$ and $\phi$ is a conjunction of constraints.  
    Let $\Phi' = \pref.\disj(\phi, \delta(Y))$ be the corresponding $2^{|\delta(Y)|}$-\textsc{Disjunct-QBF}.  
    By \Cref{lem:equisat}, $\disj(\phi, \delta(Y))$ and $\phi$ are equisatisfiable,  
    so $\Phi$ and $\Phi'$ are also equisatisfiable, as they share the same prefix.
    
    Now, observe that for every assignment $\tau$ over $\delta(Y)$ and every constraint $c \in C(\Phi, Y)$,  
    either $\tau$ satisfies $c$, or the residual constraint $c[\tau]$ is entirely over $Y$.  
    This implies that $Y$ is $\Phi'$-closed.
    We apply \Cref{th:new_unmixed_disj} to $\Phi'$ and $Y$ to obtain a partial assignment $\beta$.  
    Since $\Phi'$ and $\Phi$ are equisatisfiable (and therefore
    $\Phi'[\alpha]=\Phi[\alpha]$ for any partial assignment $\alpha$), and \Cref{th:new_unmixed_disj} guarantees that $\Phi' \Leftrightarrow \Phi'[\beta]$,  
    we conclude that 
    $\Phi \Leftrightarrow \Phi' \Leftrightarrow \Phi'[\beta]
    \Leftrightarrow \Phi[\beta]$. 
    Finally, since $|\Phi'| = 2^{|\delta(Y)|}$ and $\ariti(\Phi', Y) \leq \ariti(\Phi, Y)$,  
    all bounds on size and running time follow directly from \Cref{th:new_unmixed_disj}.
\end{proof}

The goal of the algorithm is to produce a subset $Y_{\mnu}\subseteq Y$
together with a subformula $\Phi_{\mnu} \subseteq \Phi$ and an
assignment $\beta$ of $Y_{\mnu}$ satisfying the conditions of \Cref{lem:beta_application} restated below.
The proof of the lemma was presented earlier in \Cref{ssec:unmixed-overview}.

\betaapplication*

We now show \Cref{lem:unit_clause},  
originally introduced in the overview.  
The lemma will be applied to  
$\Phi \setminus \Phi_{\mnu}$, where $\Phi$ is a  
$k$-\textsc{Disjunct-QBF} and  
$\Phi_{\mnu} = \{\phi \in \Phi \mid \ariti(\phi, Y) > 1\}$  
is the set of disjuncts containing constraints over $Y$ of arity greater than one.  
The goal is to extract a large subset  
$Y_{\mnu} \subseteq Y$ of $\Phi\setminus \Phi_{\mnu}$-redundant variables,  
which can then be used to define a suitable assignment $\beta$.

\lemunitclauses*
\begin{proof}
  Since $\ariti(\Phi,Y) \leq 1$, 
  all of the constraints from $C(\Phi,Y)$ are unit, i.e., of size $1$.
  Let $G_\Phi$ be the bipartite graph with partition $\{Y,\Phi\}$
  having an edge between $v \in Y$ and $\phi \in \Phi$ if $v \in
  V(\phi)$. Let $M$ be a maximum matching in $G_\Phi$ and denote by
  $Y_M$ and $\Phi_M$ the set of variables and formulas matched by $M$,
  respectively. We claim that
  setting $Y_{\mnu}=Y\setminus Y_M$ satisfies the statement of the
  lemma. First note that $Y_{\mnu}$ can be computed in time
  $\bigoh(k^3 + \sizeOf{\Phi})$
  by first building $G_\Phi$ in
  time $\bigoh(\sizeOf{\Phi})$ and then
  using \Cref{pro:matching} to compute $M$ in time
  $\bigoh(k^3+\sizeOf{\Phi})$, because $|\Phi|\leq k$ forms one part of
  the bipartition of $G_\Phi$.

  We now show that $Y_{\mnu}$ has the properties stated in the statement of
  the lemma.
  Because $|\Phi|\leq k$, it holds that $|Y\setminus Y_{\mnu}|=
  |Y_M|\leq |\Phi|\leq k$ and it only remains to show that $Y_{\mnu}$ is
  $\Phi$-redundant. Let $\Phi'=\pref'.\bigvee_{i \in
    [k]}\phi_i\setminus C(\phi,Y_{\mnu})$, where $\pref'$ is obtained from
  the quantifier prefix of $\Phi$ after removing all variables in
  $Y_{\mnu}$. It remains to show that if the universal player has a winning
  strategy, say $T$, on $\Phi$, then he also has a winning strategy on $\Phi'$. 
  Let $Y_M'$ be the set of all variables in $Y_M$ that are
  reachable from some variable in $Y_{\mnu}$ in $G_\Phi$ via an alternating path with
  respect to $M$.
  We claim that the strategy $T'$ obtained from $T$ after changing the
  assignment of all variables in $Y_M'$ to the assignment that
  falsifies
  the unit clause contained in their neighbor in $M$ is a
  winning strategy for the universal player on $\Phi'$; note
  that we can freely change the value of these variables within $T$
  because $Y \subseteq V_\forall(\Phi)$ and each node labeled with variables from $Y$ has only one child.
  Towards showing this we first show that $N_{G_\Phi}(Y_M'\cup Y_{\mnu})\subseteq
  \Phi_M$. Suppose for a contradiction that this is not the case and
  let $v$ be a vertex in $Y_M'\cup Y_{\mnu}$ that has a neighbor $\phi_v
  \in \Phi\setminus \Phi_M$. If $v \in Y_{\mnu}$, then $M\cup \{v,\phi_v\}$
  is a matching that is larger than $M$, which contradicts the
  maximality of $M$. Otherwise, $v \in Y_M'$
  and therefore there is an alternating path $P$ from some $u \in Y_{\mnu}$ to
  $v$ in $G_\Phi$ with respect to the matching $M$. 
  Note that without loss of
  generality we can assume that $V(P)\cap
  (\Phi\setminus\Phi_M)=\emptyset$ since otherwise there is a subpath $P'$
  of $P$ that ends in some vertex in $Y_M'$ and satisfies $V(P')\cap
  (\Phi\setminus\Phi_M)=\emptyset$, which we can use instead of $P$. But then $P$ together with the
  edge $\{v,\phi_v\}$ is an augmenting path for $M$ in $G_\Phi$, which
  again contradicts the maximality of $M$.

  Recall that for a leaf $l$ of $T$ ($T'$), we denote by $\tau_l^T$
  ($\tau_l^{T'}$) the assignment of all variables in $V(\Phi)$ given
  by the root to $l$ path in $T$ ($T'$). Moreover, for
  an assignment $\tau : V(\Phi) \rightarrow \{0,1\}$ and a subset $Y'
  \subseteq Y$, we denote by $\Phi_\tau(Y')$ the set of all $\phi \in \Phi$ that contain a
  constraint $c \in C(\phi,Y')$ that is falsified by $\tau$. We claim
  that $\Phi_{\tau_l^T}(Y)\subseteq \Phi_{\tau_l^{T'}}(Y_M)$ for every leaf $l$
  of $T$ (and $T'$), which shows that $T'$ is a winning strategy for
  the universal player on $\Phi'$. Suppose for a contradiction that
  this is not the case and let $\phi \in \Phi_{\tau_l^T}(Y)\setminus
  \Phi_{\tau_l^{T'}}(Y_M)$. Because $\phi \in \Phi_{\tau_l^T}(Y)\setminus
  \Phi_{\tau_l^{T'}}(Y_M)$, it
  holds that $\tau_l^T$ falsifies a unit clause of $\phi$ containing
  some variable $v \in Y_{\mnu} \cup Y_M'$. Because $N_{G_\phi}(Y_M'\cup
  Y_{\mnu})\subseteq \Phi_M$ and $\{v,\phi\} \in E(G_\Phi)$ also $\phi \in \Phi_M$.
  Therefore, the neighbor $u$
  of $\phi$ in $M$ is in $Y_M'$, which shows that $\tau_l^{T'}$
  falsifies a unit clause of $\phi$ containing $u$, contradicting our
  assumption that $\phi \notin \Phi_{\tau_l^{T'}}(Y_M)$. 
\end{proof}

We now show \Cref{lem:greedy_long}, which  
presents an algorithm for computing an assignment $\beta$  
over $Y_{\mnu}$ (which is defined using the previous lemma).  
While the algorithm does not always succeed in producing $\beta$  
(even when such an assignment exists),  
it guarantees the identification of a small subset of variables  
that can be used to further modify the formula $\Phi$.

We briefly recall the notion of a hitting set for conjunctive formulas,  
which was introduced earlier but is restated here for clarity.
Let $\phi$ be a conjunctive formula and let $Y_\phi \subseteq Y$ be a set of variables.  
We say that $Y_\phi$ is a \emph{hitting set} of $\phi$ with respect to $Y$  
if every constraint in $\phi$ that contains a variable from $Y$  
also contains at least one variable from $Y_\phi$,  
i.e., $C(\phi, Y_\phi) = C(\phi, Y)$.
Later, we will use the following property of a hitting set:  
for every assignment $\tau: Y \rightarrow \{0,1\}$,  
if $\ariti(\phi, Y)>1$, then
$\ariti(\phi, Y) = \ariti(\phi, Y_\phi) > \ariti(\phi[\tau], Y_\phi) = \ariti(\phi[\tau], Y)$.  
Although it is possible to compute a minimal hitting set,  
doing so would not improve the asymptotic time complexity  
of our overall algorithm.

\lemgreedylong*

\begin{proof}
Consider now the following greedy procedure on $\Phi$, whose aim is to find an assignment $\beta$
over $Y_{\mnu}$ that falsifies as many of the disjuncts in $\Phi_{\mnu}$ as
possible. 
The procedure starts with the empty assignment $\beta'=\emptyset$ and as long as there is a disjunct $\phi \in
\Phi_{\mnu}$ that is not yet falsified by $\beta'$ and a 
constraint $c$ in $C(\phi[\beta'],Y_{\mnu})$
such that $c$ does not contain a variable from $(Y\setminus Y_\mnu)\cup V(\beta')$, it extends $\beta'$ with the unique assignment that falsifies
$c$. 
Then, either the greedy procedure succeeds by producing a
partial assignment $\beta'$ of $Y_{\mnu}$ that falsifies all disjuncts in $\Phi_{\mnu}$ or there is a 
disjunct $\phi \in \Phi_{\mnu}$ that is not yet falsified by $\beta'$,
such that $V(\beta')$ is a hitting set of
$\phi \setminus C(\phi, Y \setminus Y_\mnu)$ with respect to $Y_\mnu$.
In the first case, we can just obtain an assignment $\beta$ over $Y_\mnu$ by extending $\beta'$.
In the other case, we define $Y_\phi$ as $V(\beta') \cup (Y \setminus Y_{\mnu})$, since $Y\setminus Y_\mnu$ is a hitting set of $\phi$ with respect to $Y\setminus Y_\mnu$.
Moreover, the greedy algorithm on each step assigns at most $\ariti(\Phi,Y)$ new variables, therefore $|V(\beta)| \leq |\Phi_\mnu| \cdot \ariti(\Phi,Y)$ and $|Y_\phi| \leq |\Phi_\mnu| \cdot \ariti(\Phi,Y) + |Y \setminus Y_{\mnu}| \leq |\Phi_\mnu| \cdot \ariti(\Phi,Y) + |\Phi \setminus \Phi_{\mnu}| \leq k \cdot \ariti(\Phi,Y)$.

Clearly, such a greedy algorithm checks each constraint in $\Phi$ at most once.  
Therefore, its time complexity is $\bigoh(\sizeOf{\Phi})$.
\end{proof}

If, for a given $k$-\textsc{Disjunct-QBF} $\Phi$,  
the previous lemma produces an assignment $\beta$,  
then we are done by applying \Cref{lem:beta_application}.  
Otherwise, we proceed by modifying $\Phi$  
through replacing a formula $\phi$  
with  $\phi_1 \lor \dots \lor \phi_t$,  
such that $\phi$ and $\phi_1 \lor \dots \lor \phi_t$  
are equisatisfiable, and for every $i \in [t]$,  
it holds that $\ariti(\phi, Y) > \ariti(\phi_i, Y)$.  

We repeat this transformation  
until \Cref{lem:greedy_long} yields $\beta$.  
The precise sequence of equivalent
modifications  
will be defined later (see \Cref{alg:beta}).  
Before that, we analyze two distinct  
modification strategies.  
The following lemma presents the first simpler modification,  
which relies on \Cref{lem:equisat}  
and was already introduced in the overview.

\begin{lemma}\label{lem:equisat_application}
    Let $\Phi$ be a $k$-\textsc{Disjunct-QBF} and let $\phi \in \Phi$.
    Let $Y$ and $Y_\phi$ be subsets of $V(\Phi)$ such that
    $\ariti(\phi,Y)>1$ and $Y_\phi$ is a hitting set of $\phi$ with respect to $Y$.
    Then, 
    $\Phi$ and $\pref.(\bigvee_{\phi' \in \Phi \setminus \{\phi\}}\phi')\lor  \disj(\phi, Y_\phi) $
    are equisatisfiable, where $\pref$ is a prefix of $\Phi$.
    Moreover, $\ariti(\phi,Y) > \ariti(\disj(\phi, Y_\phi) ,Y)$.
\end{lemma}
\begin{proof}
    Let $\Phi' = \pref.(\bigvee_{\phi' \in \Phi \setminus \{\phi\}}\phi')\lor  \disj(\phi, Y_\phi)$.
    From \Cref{lem:equisat}, we know that $\phi$ and $\disj(\phi, Y_\phi)$ are equisatisfiable. 
    Moreover, $\Phi \setminus \{\phi\}$ and $\Phi' \setminus \disj(\phi, Y_\phi)$
    are equisatisfiable, since both contain the same set of disjuncts and have the same prefix.
    Therefore, $\Phi$ and $\Phi'$ are equisatisfiable. 
    
    Let $\tau$ be an assignment over $Y_\phi$.
    Since $Y_\phi$ is a hitting set of $\phi$ with respect to $Y$,
    every constraint $c \in C(\phi, Y)$ contains a variable from $Y_\phi$.
    Thus, $\tau$ either satisfies $c$, or reduces its length,
    i.e., $|c[\tau]| < |c|$.
    From the fact that $\ariti(U_\tau ,Y)=1$ and $\ariti(\phi,Y)>1$,
    we obtain $\ariti(\phi,Y) > \ariti(\disj(\phi, Y_\phi) ,Y)$.
\end{proof}

While the modification introduced in \Cref{lem:equisat_application}  
is sufficient to obtain
an FPT-algorithm,  
it is neither time efficient --- since, it introduces  
$2^{|Y_\phi|}$ new disjuncts --- nor space efficient,  
as $\sizeOf{\disj(\phi,Y_\phi)} \leq 2^{|Y_\phi|} \cdot \sizeOf{\phi}$.  
Therefore, we introduce the operation $\magic(-)$,  
which provides linear size while preserving  
the key properties of $\disj(-)$
at the cost of introducing new variables.

Let $\phi$ be a conjunctive formula  
and let $X = \{x_1, \dotsc, x_t\} \subseteq V(\phi)$  
be a set of variables.  
Define a set of fresh variables
$A_X^\phi = \{a_1, \dotsc, a_t\}$ such that  
$A_X^\phi \cap V(\phi) = \emptyset$  
(and later we assume $A_X^\phi \cap V(\Phi) = \emptyset$).  
We define the disjunctive formula:

$$
  \magic(\phi, X) = (\phi^0\land\bigwedge_{i \in t} \neg a_i) \lor
    (\bigvee_{i \in [t]} (a_i \land \phi^i[x_i = 1] \land x_i )) \lor
    (\bigvee_{i \in [t]} (a_i \land \phi^i[x_i = 0] \land \neg x_i))    
$$
where $\{\phi^0, \dotsc, \phi^t\}$ is a partition of $\phi$,  
possibly with empty parts,  
such that $\phi^0 = \phi \setminus C(\phi, X)$,  
and for each $i \in [t]$, we have $\phi^i \subseteq C(\phi, x_i)$.

\begin{lemma}\label{lem:equisat_magic}
  Let $\phi$ be a conjunctive formula.
  Then, for every $X \subseteq V(\phi)$,
  the formulas $\phi$ and $\bigwedge_{\tau:A_X^\phi\rightarrow\{0,1\}} \magic(\phi, X)[\tau]$
  are equisatisfiable.

  Moreover, we can compute $\magic(\phi, X)$ in $\bigoh(|X|\cdot \sizeOf{\phi})$,
  $|\magic(\phi, X)| = 2|X| + 1$
  and if $\ariti(\phi,X) > 1$, then $\ariti(\phi,X) >
  \ariti(\magic(\phi, X),X)$.
  Finally, if all constraints in $C(\phi,X)$ are clauses, then $\sizeOf{\magic(\phi, X)} \leq
  \sizeOf{\phi} + 5|X| + 2$.
\end{lemma}
\begin{proof}
    Let $\tau_0$ be the assignment over $A_X^\phi$  
    that sets all variables to $0$.  
    For each $i \in [t]$, let $\tau_i$ be the assignment  
    that sets $a_i = 1$ and all other variables in $A_X^\phi$ to $0$.  
    We now show that for every $0 \leq i \leq t$,  
    the formulas $\magic(\phi, X)[\tau_i]$ and $\phi^i$ are equisatisfiable.  
    The case $i = 0$ is immediate,  
    since $\magic(\phi, X)[\tau_0] = \phi^0$.  
    For $i \in [t]$, observe that $\magic(\phi, X)[\tau_i]$ simplifies to  
    $(\phi^i[x_i = 1] \land x_i) \lor (\phi^i[x_i = 0] \land \neg x_i)$,
    which is equisatisfiable with $\phi^i$ by \Cref{lem:equisat}.  
    Hence, $\magic(\phi, X)[\tau_i] \equiv \phi^i$ for all $i \in [t]$.

    Let $\alpha$ be an assignment of $V(\phi)$ that falsifies $\phi$.  
    Then, there exists a constraint $c \in \phi$ and an index $0 \leq i \leq t$  
    such that $c$ is falsified by $\alpha$ and $c \in \phi^i$.  
    Since $\magic(\phi, X)[\tau_i]$ and $\phi^i$ are equisatisfiable,  
    we conclude that $\magic(\phi, X)[\tau_i]$ is falsified by $\alpha$,  
    and therefore so is $\bigwedge_{\tau : A_X^\phi \to \{0,1\}}  
    \magic(\phi, X)[\tau]$.
    
    Conversely, let $\alpha$ be an assignment of $V(\phi)$ that falsifies  
    $\bigwedge_{\tau : A_X^\phi \to \{0,1\}} \magic(\phi, X)[\tau]$.  
    Then, there exists an index $0 \leq i \leq t$ such that  
    $\magic(\phi, X)[\tau_i]$ is falsified by $\alpha$.
    Since, $\magic(\phi, X)[\tau_i]$ and $\phi^i$ are equisatisfiable,
    $\phi^i$ is falsified by $\alpha$ and so is $\phi$.

    Clearly, $|\magic(\phi, X)| = 2|X| + 1$,  
    and the partition $\phi^0, \dotsc, \phi^t$  
    can be computed in time $\bigoh(|X| \cdot \sizeOf{\phi})$  
    using a simple greedy algorithm.

    Note that for every $i \in [t]$,  
    if $\ariti(\phi^i, X) > 1$, then  
    $\ariti(\phi^i, X) > \ariti(\phi^i[x_i = 1] \land x_i, X)$  
    and $\ariti(\phi^i, X) > \ariti(\phi^i[x_i = 0]\land x_i, X)$,  
    since every constraint $c \in \phi^i$ contains $x_i$,  
    and assigning $x_i$ reduces the size of $c$,  
    i.e., $|c| > |c[x_i = 0]|$ and $|c| > |c[x_i = 1]|$.
    Moreover, if $c$ is a clause,
    then 
    $c$ is satisfied by setting $x_i$ to $0$ or $1$,  
    which implies $c[x_i = 0] = \top$ or $c[x_i = 1] = \top$.
    Therefore, if all constraints in $C(\phi, X)$  
    are clauses,  
    then $\sizeOf{\phi^i} \geq \sizeOf{\phi^i[x_i = 0]} + \sizeOf{\phi^i[x_i = 1]}$,  
    which implies that  
    $\sizeOf{\magic(\phi, X)} \leq \sizeOf{\phi} + 2|X| + 3|X|$,  
    where $2|X|$ counts the new unit clauses over $X$,  
    and $3|X|$ those over $A_X^\phi$.
\end{proof}

Note that $\phi$ and $\magic(\phi, X)$ are not equisatisfiable.  
However, $\phi$ and $\bigwedge_{\tau: A_X^\phi \rightarrow \{0,1\}}  
\magic(\phi, X)[\tau]$ are equisatisfiable.  
While the introduction of new variables and the construction  
of a seemingly more complex formula may appear unintuitive,  
the rationale behind defining $\magic$ lies in its structure:  
the formula $\bigwedge_{\tau: A_X^\phi \rightarrow \{0,1\}}  
\magic(\phi, X)[\tau]$ arises naturally from  
$\magic(\phi, X)$ via universal quantifier elimination  
over the variables $A_X^\phi$.
That is why we can present the second modification in the following lemma.

\begin{lemma}\label{lem:equisat_magic_application}
    Let $\Phi$ be a $k$-\textsc{Disjunct-QBF} with prefix $\pref$ and let $\phi \in \Phi$.
    Let $Y$ and $Y_\phi$ be subsets of $V(\Phi)$ such that
    $\ariti(\phi,Y)>1$ and $Y_\phi$ is a hitting set of $\phi$ with respect to $Y$.
    Then,
    $\Phi$ and $\Phi'$
    are equisatisfiable,
    where $\Phi'$ is obtained from $\pref\forall A_{Y_\phi}^\phi.(\bigvee_{\phi' \in \Phi \setminus \{\phi\}}\phi')\lor  \magic(\phi, Y_\phi)$ after quantifier eliminating $A_{Y_\phi}^\phi$.
    Moreover,  $\ariti(\phi,Y) > \ariti(\magic(\phi, Y_\phi) ,Y)$.
\end{lemma}
\begin{proof}
  \sloppypar
     Assume that $\Phi'$ is in \textsc{Disjunct-QBF} form and let  
    $\Phi'' = \pref.(\bigvee_{\phi' \in \Phi \setminus \{\phi\}} \phi')  
    \lor \bigwedge_{\tau : A_{Y_\phi}^\phi \to \{0,1\}}  
    \magic(\phi, Y_\phi)[\tau]$.  
    Note that for every $\phi' \in \Phi \setminus \{\phi\}$,  
    $\Phi'$ contains the formulas  
    $\bigwedge_{\tau : A_{Y_\phi}^\phi \to \{0,1\}} \phi'[\tau]$.  
    Since $V(\phi') \cap A_{Y_\phi}^\phi = \emptyset$,  
    we have that $\phi'$ and $\phi'[\tau]$ are equisatisfiable.  
    Therefore, adding all such conjunctive formulas from  
    $\Phi \setminus \{\phi\}$ to $\Phi'$,  
    and removing redundant ones that are proper subsets of others,  
    yields a \textsc{Disjunct-QBF} representation of $\Phi''$.  
    It follows that $\Phi'$ and $\Phi''$ are equisatisfiable.

    From \cref{lem:equisat_magic}, we know that $\phi$ and $\bigwedge_{\tau:A_{Y_\phi}^\phi\rightarrow\{0,1\}} \magic(\phi, Y_\phi)[\tau]$ are equisatisfiable. 
    Moreover, $\Phi \setminus \{\phi\}$ and $\Phi' \setminus \{\bigwedge_{\tau:A_{Y_\phi}^\phi\rightarrow\{0,1\}} \magic(\phi, Y_\phi)[\tau]\}$
    are equisatisfiable, since both contain the same set of disjunct and have the same prefix.
    Hence, $\Phi$ and $\Phi''$ are equisatisfiable,  
    and thus so are $\Phi$ and $\Phi'$.    

    Since $Y_\phi$ is a hitting set of $\phi$ with respect to $Y$,  
    it follows that $Y_\phi$ is also a hitting set of  
    $\magic(\phi, Y_\phi)$ with respect to $Y$.  
    Therefore, we have  
    $\ariti(\magic(\phi, Y_\phi), Y_\phi)  
    = \ariti(\magic(\phi, Y_\phi), Y)$.  
    From \Cref{lem:equisat_magic}, it then follows that  
    if $\ariti(\phi, Y) > 1$,  
    we have  
    $\ariti(\phi, Y) > \ariti(\magic(\phi, Y_\phi), Y)$.
\end{proof}

We are now ready to define the algorithm for computing $\beta$. It
proceeds by iteratively modifying the input \textsc{Disjunct-QBF} by
replacing high arity disjuncts $\phi^i$ with a disjunction of conjunctive formulas of strictly smaller arity.
\begin{algorithm}
\caption{Algorithm for computing the partial assignment $\beta$. \\
\textbf{Input:} A $k$\hy\textsc{Disjunct-QBF} formula $\Phi$ and a
$\Phi$-closed set $Y$ of universal variables. \\
\textbf{Output:} A partial assignment $\beta$ over $Y$ such that $\Phi
\Leftrightarrow \Phi[\beta]$ and $|Y\setminus V(\beta)|$ is bounded by
a function of $k+\ariti(\Phi,Y)$.
}\label{alg:beta}
\begin{algorithmic}[1]
\State $\Phi^0 \gets \Phi$
\State $i \gets 0$
\While{true}
    \State $\Phi^i_\mnu \gets \{\phi \in \Phi^i \mid \ariti(\phi, Y) > 1\}$ %\Comment{disjuncts with $Y$-arity $>1$}
    \State $Y^i_\mnu \gets$ output of \Cref{lem:unit_clause} on $\Phi^i \setminus \Phi^i_\mnu$ and $Y$
    \State Apply \Cref{lem:greedy_long} to $\Phi^i$, $\Phi^i_\mnu$, $Y$, $Y^i_\mnu$
    \If{\Cref{lem:greedy_long} yields $\beta$}
        \State $i_{\text{last}} \gets i$
        \State \Return $\beta$
    \Else
        \State \Cref{lem:greedy_long} yields  $\phi^i$ and $Y_{\phi^i}$
        \State $\pref^i \gets$ prefix of $\Phi^i$
        \State $A^{i+1} \gets  A_{Y_{\phi^i}}^{\phi^i}$
        \State $\Phi^{i+1} \gets \pref^i \forall A^{i+1}. \left( \bigvee_{\phi' \in \Phi^i \setminus \{\phi^i\}} \phi' \lor \magic(\phi^i, Y_{\phi^i}) \right)$ 
        \Comment{see \Cref{lem:equisat_magic_application}}
        \State $i \gets i + 1$
    \EndIf
\EndWhile
\end{algorithmic}
\end{algorithm}

Note that none of the algorithms or lemmas involved are deterministic;
different implementations may yield different intermediate values. Nevertheless, the number of iterations performed by the algorithm is strictly bounded in the worst-case scenario.
The following lemma, a detailed restatement of \Cref{lem:new_limit},  
shows that \Cref{alg:beta} terminates after a number of steps that is
bounded by a function of the parameters.
%which is the main challenge in this subsection.

\begin{lemma}\label{lem:new_limit_long}
  \Cref{alg:beta} terminates with $i_{\text{last}}\leq
  \exp_{4d+2}^{d-2}(k)$, where $d=\SA(\Phi,Y)$.
\end{lemma}
\begin{proof}
    We start by defining our progress measure $f(\Phi^i)$ that is 
    the $(d+1)$-dimensional
    vector, whose every entry is equal to the number of disjuncts of the given arity,
    where $d =\ariti(\Phi^0,Y)$,
    i.e.,
    for every $0 \leq a \leq d$, $f(\Phi^i)[a] = |\SB \phi \SM \phi \in \Phi^i \land \ariti(\phi,Y)=a \SE|$.
    Clearly, the sum of all values in $f(\Phi^i)$ is equal to $|\Phi^i|$.
    Moreover,
    $f(\Phi^{i+1})$ is lexicographically smaller than $f(\Phi^{i})$,
    i.e.,
    there exists $d' \in [d]$ such that $f(\Phi^{i+1})[d'] < f(\Phi^{i})[d']$ and for every $d' < a  \leq d$, it holds that $f(\Phi^i)[a] = f(\Phi^{i+1})[a]$.

    Let $i_{\max}$ be a worst case scenario index, such that
    for every possible run of \Cref{alg:beta}, the obtained $i_{\text{last}}$ is smaller than 
    $i_{\max}$.
    To calculate $i_{\max}$, we assume 
    that, 
    $f(\Phi^0)[d] = k$ (since this yields the maximum number of iterations) and
    on each step $i \in [i_{\max}]$, the number $f(\phi^{i-1})$ corresponds 
    to the largest index with a nonzero value in $f(\Phi^i)$,
    and $f(\Phi^i)$ 
    is increased only at position $\ariti(\phi^i,Y) - 1$.
    Hence, at most two positions in $f(\Phi^i)$ are nonzero, and these are 
    consecutive.
    Moreover, $\Phi^{i_{\max}}_\mnu = \emptyset$, i.e., $f(\Phi^{i_{\max}})$ has zero value on positions larger than one, 
    since \Cref{lem:greedy_long} would always generate $\beta$.
    In this scenario, 
    there exist indices 
    $0 = i_d \leq i_{d-1} \leq \dots \leq i_1 = i_{\max}$ such that, 
    for each 
    $j \in [d]$, $f(\Phi^{i_{j}})$ contains a nonzero value only at position $j$, 
    i.e., $f(\Phi^{i_{j}})[j] = |\Phi_{i_j}|$, 
    and $i_{j-1} = i_{j} + |\Phi^{i_{j}}|$. 

    From  \Cref{lem:greedy_long}, we obtain that
    $|Y_{\phi^i}| \leq \ariti(\Phi^i,Y) | \Phi^i|$.
    Moreover, from \Cref{lem:equisat_magic_application},
    we obtain that
    $|\Phi^{i+1}| = |\Phi^i|  - 1 + 2|Y_{\phi^i}| + 1 \leq |\Phi^i| + 2\ariti(\Phi^i,Y) |\Phi^i| = (2\ariti(\Phi^i,Y)+1)|\Phi^i|$.
    Therefore, 
    for every $j \in [d]$
    it holds that
    $|\Phi^{i_{j-1}}| \leq (2j+1)^{i_{j} - i_{j+1}}|\Phi^{j_{i}}| \leq
    (\Pi_{j\leq l <d} (2l + 1)^{i_{l} - i_{l+1}})|\Phi^{i_d}|
    % (2j+1)^{i_{j} - i_{j+1}} \cdot (2j + 3)^{i_{j+1} - i_{j+2}} \cdot \dotsb \cdot (2d+1)^{i_d} |\Phi^{i_d}|
    \leq (2d+1)^{i_j} k$.
    Since $i_{j-1} = i_{j} + |\Phi^{i_{j}}|$,
    we obtain that 
    $i_{j-1} \leq i_j + (2d+1)^{i_j}k$.
    Since $i_{d-1} = k$ and $k\geq 2$,
    it holds that
    for every $j \in [d-1]$,
    $i_j + k\leq 2^{i_j}$,
    Therefore,
    $i_{j-1} \leq i_j  + (2d+1)^{i_j} k \leq 
    (2d+1)^{i_j} (i_j + k) \leq
    (2d+1)^{i_j} 2^{i_{j}} \leq 
    (4d+2)^{i_j}$
    and 
    $i_{\text{last}} \leq i_{\max} = i_1 \leq \exp_{4d+2}^{d-2}(k)$.
\end{proof}

Now we are ready to provide the proof of \Cref{th:new_unmixed_disj}.

\begin{proof}[Proof of \Cref{th:new_unmixed_disj}]
For simplicity, let $d = \ariti(\Phi,Y)$.
Let $\Phi'$ be a $k$-\textsc{Disjunct-QBF} obtained from $\Phi$ by
replacing every constraint $c$ containing only variables from $Y$ with
an equisatisfiable CNF formula.
Clearly, $\Phi$ and $\Phi'$ are equisatifable and $\sizeOf{\Phi'}\leq 2^{d}\sizeOf{\Phi}$.
We apply  \Cref{alg:beta} to $\Phi'$ and $Y$ and let $\beta$ be the partial assignment obtained from the algorithm.
From \Cref{lem:unit_clause}, we know that $Y_{\mnu}^{i_{\text{last}}}$  is $\Phi^{i_{\text{last}}}\setminus\Phi^{i_{\text{last}}}_{\mnu}$-redundant
and from \Cref{lem:greedy_long}, $\beta$ falsifies every conjunctive formula from $\Phi^{i_{\text{last}}}_{\mnu}$.
By applying
$\Phi^{i_{\text{last}}}$, $\Phi^{i_{\text{last}}}_{\mnu}$, $Y_{\mnu}^{i_{\text{last}}}$ and $\beta$ to \Cref{lem:beta_application}, we obtain that $\Phi^{i_{\text{last}}} \Leftrightarrow \Phi^{i_{\text{last}}}[\beta]$.

From \Cref{lem:equisat_magic_application}, we obtain that
$\Phi^{i}$ and the formula obtained from $\Phi^{i+1}$ after quantifier elimination of $A^{i+1}$ are equisatisable.
Therefore, $\Phi'$ and $\Phi^\star$ are equisatisable
and so are $\Phi$ and $\Phi^\star$,
where $\Phi^\star$ is a
formula obtained from $\Phi^{i_{\text{last}}}$ after quantifier elimination of $A=\bigcup_{i\in [i_{\text{last}}]}A^{i}$.
Since, $V(\beta) \cap A = \emptyset$ and $\Phi^{i_{\text{last}}} \Leftrightarrow \Phi^{i_{\text{last}}}[\beta]$, it holds that $\Phi^\star\Leftrightarrow \Phi^\star[\beta]$

Note that at each step we have the bound $|\Phi^i| \leq (2d+1)^i \cdot k$.  
Using \Cref{lem:new_limit_long}, we obtain the worst-case bound
$|\Phi^{i_{\text{last}}}| \leq (2d+1)^{\exp_{4d+2}^{d-2}(k)}k \leq \exp_{4d+2}^{d-1}(k)$.
Furthermore, since $|Y \setminus Y_{\mnu}^{i_{\text{last}}}| \leq |\Phi^{i_{\text{last}}}|$,
we conclude that $|Y\setminus V(\beta)| \leq \exp_{4d+2}^{d-1}(k)$.

Note that during the $i$-th iteration of the algorithm,  
there is no need to retain the previous formulas $\Phi^j$ for $j < i$,  
and thus the space used in the $i$-th iteration is $\bigoh(\sizeOf{\Phi^{i+1}})$.  
Consequently, the total space required by the algorithm is bounded by $\bigoh(\sizeOf{\Phi^{i_{\text{last}}}})$.
Therefore the space complexity of algorithm is $\bigoh(2^d\sizeOf{\Phi} + d(\exp_{4d+2}^{d-1}(k))^2)$.
Since every constraint in $C(\Phi^i, Y)$ is a disjunction of literals,  
by \Cref{lem:equisat_magic}, we obtain
$
\sizeOf{\Phi^{i+1}} \leq \sizeOf{\Phi^i} - \sizeOf{\phi^i} + \sizeOf{\phi^i} + 5|Y_{\phi^i}| + 2.
$
From \Cref{lem:greedy_long}, this simplifies to
$
\sizeOf{\Phi^{i+1}} \leq \sizeOf{\Phi^i} + 5 \cdot \ariti(\Phi^i, Y) \cdot |\Phi^i| + 2.
$.
Thanks to \Cref{lem:new_limit_long}, 
we have that for each $i \in [i_{\text{last}}]$,
it holds:
\begin{align*}
\sizeOf{\Phi^{i}} \leq \\
\sizeOf{\Phi'} + \sum_{j\in[i]}(5\ariti(\Phi^j,Y) |\Phi^j| + 2)  \leq \\
2^{d} \sizeOf{\Phi} + 2i +5\sum_{j\in[i]}d|\Phi^{j}| \leq\\ 
2^{d} \sizeOf{\Phi} + 2i +5\sum_{j\in[i]}d(2d+1)^j \leq\\ 
2^{d} \sizeOf{\Phi} + 2i +10d(2d+1)^i \leq\\ 
2^{d} \sizeOf{\Phi} +5(2d+1)^{i+1}
\end{align*}
Therefore the overall space complexity is $\bigoh(2^{d} \sizeOf{\Phi} + d\exp_{4d+2}^{d-1}(k)))$.

During $i$-th iteration we will spend time:
\begin{itemize}
    \item $\bigoh(|\Phi^i|^3 + \sizeOf{\Phi^i})$ in \Cref{lem:unit_clause}
    \item $\bigoh(\sizeOf{\Phi^i})$ in \Cref{lem:greedy_long}
    \item $\bigoh(|Y^i_{\phi^i}| \cdot \sizeOf{\phi^i}) = \bigoh(d |\Phi^i| \cdot \sizeOf{\Phi^0}) = \bigoh(d 2^d |\Phi^i|\cdot\sizeOf{\Phi})$ in \Cref{lem:greedy_long}, since $\sizeOf{\phi^i}\leq \sizeOf{\Phi^0}$
\end{itemize}
Moreover, we spend $\bigoh(2^d \sizeOf{\Phi})$ time for constructing $\Phi'$ ($\Phi^0$).
Since $|\Phi^i|\leq |\Phi^{i+1}|$, we obtain the following time complexity:
\begin{align*}
\bigoh(2^d \sizeOf{\Phi} + \sum_{i \in [i_{\text{last}}]}|\Phi^i|^3 + \sizeOf{\Phi^i} + \sizeOf{\Phi^i} + d 2^d |\Phi^i|\cdot\sizeOf{\Phi}) = \\
\bigoh(\sum_{i \in [i_{\text{last}}]}|\Phi^i|^3 + d 2^d |\Phi^i|\cdot\sizeOf{\Phi}) = \\
\bigoh(\sum_{i \in [i_{\text{last}}]}(2d+1)^{3i}k + d 2^d (2d+1)^i\sizeOf{\Phi}) = \\
\bigoh((2d+1)^{3i_{\text{last}}}k + d 2^d (2d+1)^i_{\text{last}}\sizeOf{\Phi}) = \\
\bigoh((\exp_{4d+2}^{d-1}(k))^3 + d 2^d \exp_{4d+2}^{d-1}(k)\sizeOf{\Phi})
\end{align*}
\end{proof}
 
\subsection{Applications for Enhanced (Deletion) Backdoors}\label{ssec:mappenh}

Here, we show the algorithmic consequences of the elimination of
 guarded universal sets method provided in
\Cref{th:new_unmixed_qbf} as overviewed in \Cref{ssec:enh}. We start
by providing the proofs for \Cref{thm:backdoorunmixed},
\Cref{the:extendableclasses}, and \Cref{the:extdel} from \Cref{ssec:enh}.
Recall that a class $\CCC$ of
$\QBF$ formulas is \emph{(deletion) extendable} if deciding the validity of
any $\QBF$ formula is
fixed-parameter tractable parameterized by the size of a given (deletion)
backdoor $B$ to $\CCC$.

\theextendableclasses*

\begin{proof}
  \ref{DCo} follows directly from \Cref{thm:introtractcnf}.

  Towards showing~\ref{DCt}, let $\Phi$ be a $\QBF$ formula with a
  backdoor $B$ to the class $\DDD_{\exists}(\HORN)$ of existential $\QBF(\HORN)$
  formulas. Then, using \Cref{lem:from-backdoor-to-disjunct-qe}, we can transform $\Phi$ into
  a $k$\hy\textsc{Disjunct}\hy$\QBF(\DDD_{\exists}(\HORN))$ (where $k$ is bounded in $|B|$) formula $\Phi'$ in FPT-time parameterized
  by $|B|$ such that every disjunct of $\Phi'$ is an existential
  $\QBF(\HORN)$ formula. We can then solve this in polynomial-time by
  solving every disjunct of $\Phi'$ using the algorithm for
  propositional \HORN{} formulas.
  
  Towards showing~\ref{DCh}, let $\Phi$ be a $\QBF$ formula with a
  deletion backdoor $B$ to the class of $\QBF$ formulas with at most
  $q$ quantifier alternations and incidence treewidth at most
  $\omega$. Then, the number of quantifier alternations of $\Phi$ is bounded by (a function of) $|B|+q$, and the incidence treewidth is at most $|B|+\omega$. It is therefore
  fixed-parameter tractable parameterized by $|B|$ because of the well
  known result that $\QBF$ is fixed-parameter tractable parameterized
  by the number of quantifier alternations plus the incidence
  treewidth~\cite{DBLP:conf/stacs/CapelliM19}.

  Towards showing~\ref{DCf}, let $\Phi$ be a $\QBF$ formula with a
  deletion backdoor $B$ to the class of $\QBF$ formulas with incidence
  vertex cover number at most $s$. Then, $\Phi$ has incidence vertex cover
  number at most $|B|+s$ and fixed-parameter tractability by $|B|$
  follows from a known result showing fixed-parameter tractability of
  $\QBF$ parameterized by incidence vertex cover number~\cite{DBLP:conf/lics/FichteGHSO23}.
\end{proof}

We now provide a proof for \Cref{the:extdel}, which uses our
elimination of guarded universal sets method to show that
enhanced (deletion) backdoor evaluation to any (deletion) extendable class
is fixed-parameter tractable
parameterized by the enhanced (deletion) backdoor size plus the
maximum number of variables from the guarded universal set appearing
in any constraint. Recall that $B \subseteq V(\Phi)$ is an \emph{enhanced
  (deletion) backdoor} to $\CCC$ for $\Phi$ if $B\cup
\RU(\Phi,B)$ is a (deletion) backdoor to $\CCC$ for $\Phi$, for a class $\CCC$ of $\QBF$
formulas. Here $\RU(\Phi,B)$ denotes the union of all 
components in $\primG(\Phi)-B$ that contain only universal variables.

\theextdel*

\begin{proof}
Let $Y=\RU(\Phi,B)$.
  Note that $\delta(Y) \subseteq B$ and therefore we can employ
  \Cref{th:new_unmixed_qbf} to compute a partial assignment $\beta$ of
  $Y$ such that $\Phi\Leftrightarrow\Phi[\beta]$ and
  $|Y\setminus V(\beta)|$ is bounded by a function, say $g$, of
  $|\delta(Y)|+\SA(\Phi,Y)$ in FPT-time parameterized by
  $|\delta(Y)|+\SA(\Phi,Y)$. But then, $B\cup (Y\setminus V(\beta))$
  is a (deletion) backdoor to $\CCC$ for $\Phi[\beta]$ of size at most
  $|B|+g(|\delta(Y)|+\SA(\Phi,Y))$, which because $\CCC$ is (deletion)
  extendable
  implies that deciding the validity of $\Phi[\beta]$ is
  fixed-parameter tractable parameterized by
  $|B|+|\delta(Y)|+\SA(\Phi,Y)$. Therefore, altogether, deciding the
  validity of $\Phi$ is fixed-parameter tractable parameterized
  by $|B|+|\delta(Y)|+\SA(\Phi,Y)$, which because $|\delta(Y)|\leq |B|$
  completes the proof of the theorem.
\end{proof}
In order to use enhanced (deletion) backdoors (for the classes
provided in \Cref{the:extendableclasses}), we also need to be able to
find these backdoors efficiently. Using standard branching in
combination with important separators~\cite[Section 8]{DBLP:books/sp/CyganFKLMPPS15} allows us to obtain the
following FPT-algorithm for detecting enhanced backdoors into
syntactically defined classes (see Section~\ref{ssec:mfindenh} for proof).

\begin{restatable}{theorem}{theenhdetcon}\label{the:enhdet-con}
  Let $\CCC\in \{\DDD_q(\TCNF),\DDD_q(d\hy\AFF),\DDD_{\exists}(\HORN)\}$.
  Let $\Phi$ be a $\QBF$ and let $k$ be an integer. There is an
  FPT-algorithm parameterized by $k$ that computes an enhanced
  backdoor $B$ to $\CCC$ for $\Phi$ of size at most $k$ if such an enhanced backdoor exists and
  otherwise returns no.
\end{restatable}

\Cref{the:enhdet-con} together with
\Cref{the:extendableclasses,the:extdel} now allows us to obtain the following fixed-parameter tractability
results for \QBF{}.
\begin{corollary}\label{cor:enhdetcont}
  Let $\CCC$ be any class of $\QBF$ formulas in
  $\{\DDD_q(\TCNF),\DDD_q(d\hy\AFF)\}$.
  Then, $\QBF$ is fixed-parameter tractable parameterized by the
  minimum size of an enhanced backdoor to $\CCC$.
\end{corollary}
\begin{proof}
  Let $\Phi$ be a \QBF{}
  formula. Using \Cref{the:enhdet-con},
  we can compute an enhanced backdoor $B$ of $\Phi$ to $\CCC$ of minimum size in FPT-time
  parameterized by the size of $B$. Moreover, because of \Cref{the:extendableclasses}, it follows that $\CCC$
  is extendable and therefore using \Cref{the:extdel}, we can evaluate
  $\Phi$ in FPT-time with respect to $|B|+\SA(\Phi,Y)$, where
  $Y=\RU(\Phi,B)$. Because $\SA(\Phi,Y)$ is bounded by the maximum
  arity of any constraint in $\Phi$ and the maximum arity of any constraint in
  $\Phi$ is at most $|B|+2$ for the case that $\CCC=\DDD_q(\TCNF)$ and at most
  $|B|+d$ for the case that $\CCC=\DDD_q(d-\AFF)$, the result follows.
\end{proof}

\begin{corollary}\label{cor:enhdetconh}
  $\QBF$ is fixed-parameter tractable parameterized by the
  minimum size of an enhanced backdoor to
  $\DDD_{\exists}(\HORN)$ plus the maximum number of universal
  variables in any constraint.
\end{corollary}
\begin{proof}
  The proof is analogous to the proof of \Cref{cor:enhdetcont} after
  observing that $\SA(\Phi,Y)$ is at most equal to the maximum number
  of universal variables in any constraint.
\end{proof}

We are now able to show the algorithmic results of \Cref{thm:backdoorunmixed} from
\Cref{ssec:enh}, which are an almost immediate consequence of~\Cref{cor:enhdetcont,cor:enhdetconh}.

\thmbackdoorunmixed*
\begin{proof}
  Let $\Phi$ be a \QBF{} formula and $B$ be a backdoor to unmixed
  $\CCC$. We first show that $\Phi-B$ is an unmixed formula. Suppose
  not, and let $c$ be a constraint (or clause) of $\Phi$ such that
  $V(c)-B$ is not unmixed. Let $\alpha$ be any assignment of the
  variables in $B$ such that $c[\alpha]$ is not trivially satisfied
  (which exists because otherwise $c$ would have been trivial and
  could therefore have been removed from $\Phi$ in the
  beginning). Then, $\Phi[\alpha]$ is not unmixed, because it contains
  $c[\alpha]$. It now follows that $\RU(\Phi,B)$ contains all
  universal variables of $\Phi$ and therefore $B$ is an enhanced
  backdoor for $\Phi$ to existential $\CCC$. In other words, we have
  just shown that $B$ is an enhanced backdoor to \QBF{} formulas in $\CCC$ with $0$
  quantifier alternations. Therefore, the result of the
  theorem for the cases $\CCC \in \{\TCNF,d\hy\AFF\}$ follows from
  \Cref{cor:enhdetcont} and for the case that $\CCC=d\hy\HORN$, the
  result follows from \Cref{cor:enhdetconh} together with the
  observation that any constraint in $\Phi$ has size at most $|B|+d$,
  because $B$ is a backdoor to $d\hy\HORN$.
\end{proof}

Compared to the detection of enhanced backdoors to classes defined via
constraint languages, the detection of enhanced deletion backdoors to
the two structural base classes $\TTT_q^\omega$ and $\VC_s$ turns out
to be quite a bit more involved. In particular, in addition to
standard branching and important separators (both are mainly used to
reduce the number of quantifier alternations for the base class
$\TTT_q^\omega$), we now also need to employ the framework of
unbreakability~\cite{LokshtanovR0Z18}. Using this we obtain the following
results for detecting enhanced deletion backdoors
into these classes (see Section~\ref{ssec:mfindenh} for proofs).
\begin{restatable}{theorem}{theenhdetstrutw}\label{the:enhdet-stru-tw}
  Let $\Phi$ be a $\QBF$ formula and let $k$ be an integer. There is an
  FPT-algorithm parameterized by $k$ that either computes an enhanced
  deletion backdoor $B$ to $\TTT_q^\omega$ for $\Phi$ of size at most $2k$ if such
  an enhanced backdoor of size at most $k$ exists and
  otherwise returns no.
\end{restatable}
\begin{restatable}{theorem}{theenhdetstruvc}\label{the:enhdet-stru-vc}
  Let $\Phi$ be a $\QBF$ formula and let $k$ be an integer. There is an
  FPT-algorithm parameterized by $k$ that either computes an enhanced
  deletion backdoor $B$ to $\VC_s$ for $\Phi$ of size at most $k$ if such
  an enhanced backdoor of size at most $k$ exists and
  otherwise returns no.
\end{restatable}

As a corollary of \Cref{the:enhdet-stru-tw,the:enhdet-stru-vc} together with
\Cref{the:extendableclasses,the:extdel}, we obtain the following
tractability results for \QBF.
\begin{corollary}\label{cor:enhdetstruc}
  Let $\CCC$ be any class of $\QBF$ formulas in $\{\TTT_q^\omega,\VC_s\}$.
  Then, $\QBF$ is fixed-parameter tractable parameterized by the
  minimum size of an enhanced deletion backdoor to $\CCC$ plus the
  maximum number of universal variables in any constraint.
\end{corollary}
\begin{proof}
  Let $\Phi$ be a \QBF{} formula. Using \Cref{the:enhdet-stru-tw} (or \Cref{the:enhdet-stru-vc}),
  we can compute an enhanced deletion backdoor $B$ of $\Phi$ to
  $\CCC$, whose size is at most two times the size of a minimum one, in FPT-time
  with respect to $|B|$. Moreover, because of \Cref{the:extendableclasses}, it follows that $\CCC$
  is deletion extendable and therefore using \Cref{the:extdel}, we can evaluate
  $\Phi$ in FPT-time with respect to $|B|+\SA(\Phi,Y)$, where
  $Y=\RU(\Phi,B)$. Because $\SA(\Phi,Y)$ is bounded by the maximum
  number of universal variables occurring in any constraint of $\Phi$, the result follows.
\end{proof}
Note that \Cref{the:enhresover} restated below
is now an immediate consequence of \Cref{cor:enhdetcont,cor:enhdetconh,cor:enhdetstruc}.

\theenhresover*

It remains to show how to efficiently find enhanced (deletion)
backdoors, i.e., \Cref{the:enhdet-con,the:enhdet-stru-tw,the:enhdet-stru-vc}, which we do in
the next subsection.

\subsection{Finding Enhanced (Deletion) Backdoors}\label{ssec:mfindenh}

\newcommand{\nois}[1]{4^{#1}}

This subsection is devoted to a proof of
\Cref{the:enhdet-con,the:enhdet-stru-tw,the:enhdet-stru-vc}. We start
by showing \Cref{the:enhdet-con}, which is restated below.

\theenhdetcon*

We start with introducing some notation and known results about
important separators. Let $G$ be an undirected graph and let $X$ and
$Y$ be two subsets of $V(G)$. An \emph{$X$-$Y$-separator} is
a set $S \subseteq V(G)$ such that the graph $G-S$ has no path connecting
a vertex in $X\setminus S$ and a vertex in $Y\setminus S$. 
Let $R_{G'}(Z)$ denote the set of
vertices reachable from a set $Z$ of vertices in a graph $G'$.
An
\emph{important $X$-$Y$-separator} is
an inclusion-wise minimal $X$-$Y$-separator $S$ such that there is no
$X$-$Y$-separator $S'$ with $|S'|\leq |S|$ and
$R_{G-S}(X\setminus S)\subset R_{G-S'}(X\setminus S')$.
\begin{theorem}[{\cite[Proposition 5.20 and 5.21]{DBLP:books/sp/CyganFKLMPPS15}}]\label{the:impsep}
  Let $G$ be a graph and let $X$ and $Y$ be two subsets of $V(G)$ such
  that there is an inclusion-wise minimal $X$-$Y$-separator $S$ in $G$. Then, the set
  $\mathcal{S}$ of all important $X$-$Y$-separators of
  size at most $|S|$ is non-empty, contains at most $\nois{|S|}$ elements, and can be
  computed in time $\bigoh(|\mathcal{S}||S|^2(|V(G)|+|E(G)|))$.
\end{theorem}

We are now ready to show the following auxiliary lemma, which informally allows
us to assume that any enhanced (deletion) backdoor $B$ for a \QBF
formula $\Phi$ contains an important $y$-$V_{\exists}(\Phi)$-separator
for every $y \in \RU(\Phi,Y)$, which is important to reduce the number
of branches that have to be considered. Recall that we always assume
that every class of \QBF{} formulas is closed under taking
subformulas.
\begin{lemma}\label{lem:enhdetimsep}
  Let $\CCC$ be a class of $\QBF$ formulas, let
  $\Phi$ be a $\QBF$ formula with an enhanced (deletion) backdoor $B$ to $\CCC$ for $\Phi$
  with $Y=\RU(\Phi,B)$, let $B_0 \subseteq B$, $y \in Y$, and let
  $S\subseteq B\setminus B_0$ be an
  inclusion-wise minimal $y$-$V_{\exists}(\Phi)$-separator  in $\primG(\Phi-B_0)$. Then, there is an important
  $y$-$V_{\exists}(\Phi)$ separator $S'$ in $\primG(\Phi-B_0)$ of size at most $|S|$
  such that
  $B'=(B\setminus S)\cup S'$ is also an enhanced (deletion) backdoor to $\CCC$ for $\Phi$.
\end{lemma}
\begin{proof}
  Let $S'$ be any important
  $y$-$V_{\exists}(\Phi)$ separator in $\primG(\Phi-B_0)$ of size at most $|S|$ such that $R_{\primG(\Phi-(B_0\cup S'))}(y)\subseteq
  R_{\primG(\Phi-(B_0\cup S))}(y)$. Note that $S'$ exists
  because of \Cref{the:impsep} together with the fact that $S$ is an inclusion-wise minimal $y$-$V_{\exists}(\Phi)$ separator in
  $\primG(\Phi-B_0)$ and $|S| \geq |S'|$. Then, $B\cup Y\subseteq B'\cup Y'$, where
  $Y'=\RU(\Phi,B')$. Therefore, it holds that
  $B'\cup Y'$ is also a (deletion) backdoor to $\CCC$ for $\Phi$
  (because $\CCC$ is closed under taking subformulas),
  which implies that $B'$ is also an enhanced (deletion) backdoor to
  $\CCC$ for $\Phi$.
\end{proof}

\newcommand{\QQQ}{\mathcal{Q}}
We are now ready to provide an algorithm to find enhanced (deletion)
backdoors to the class $\QQQ_q$ of all \QBF{} formulas with at most $q$
quantifier alternations.
\begin{lemma}\label{lem:EB-QA-B}
  Let $\Phi$ be a $\QBF$ formula and let $k$ and $q$ be integers. There is an
  FPT-algorithm parameterized by $k+q$ that outputs an enhanced (deletion)
  backdoor $B$ to $\QQQ_q$ for $\Phi$ of size at most $k$ if it exists
  and otherwise outputs no.
\end{lemma}
\begin{proof}
  We employ a branching algorithm that builds a depth-bounded search
  tree, where every inner node is labeled with a potential enhanced (deletion)
  backdoor $B'$ of size at most $k$ and every leaf is additionally
  labeled yes if $B'$ is an enhanced (deletion) backdoor to $\QQQ_q$
  for $\Phi$ and otherwise with no. We build
  the tree iteratively starting at the root, which we label with
  $\emptyset$. We then expand the tree at every inner node, which is
  not yet labeled yes or no as follows.

  Namely, consider any inner node $n$ with label $B'$ (that is not yet
  labeled yes or no) and let
  $\Phi'=\Phi-(B'\cup \RU(\Phi,B'))$. If $\Phi'$ has at most $q$
  quantifier alternations, we update the label of $n$ to
  yes. Moreover, if $|B'|=k$, we update the label of $n$ to no and stop.
  Otherwise,
  let $X_1,\dotsc, X_{q+2}$ be the sets of variables within the outermost $q+2$ quantifier blocks
  of $\Phi'$, which exists because $\Phi'$ has more than $q$
  quantifier alternations. Then, for every enhanced backdoor $B$ with
  $Y=\RU(\Phi',B)$ of
  $\Phi'$ to $\QQQ_q$, it holds that there is an $X_i$ such that
  $X_i \subseteq B\cup Y$. Moreover, if $X_i \subseteq B\cup Y$, then
  for every variable $x \in X_i$ either:
  \begin{itemize}
  \item $x$ is existential and $x \in B$,
  \item $x$ is universal and $x \in B$, or
  \item $x$ is universal and $x \in Y$.
  \end{itemize}
  These observations allow us to exhaustively branch as follows.
  Namely, for every $i$ with $1 \leq i \leq q+2$, we do the
  following. Let $x$ be an arbitrary variable in $X_i$. We add a
  new child to $n$ with label $B'\cup \{x\}$. Moreover, if $x$ is
  universal, then for every $k'\leq k-|B'|$ and every important
  $x$-$V_{\exists}(\Phi)$-separator $S$ in $\primG(\Phi-B')$ of size at most $k'$, we add
  to $n$ a child with label $B'\cup S$.
  This completes the description of the algorithm.

  Towards showing the correctness of the algorithm, let $B$ be an
  enhanced (deletion) backdoor to $\QQQ_q$ for
  $\Phi$ of size at most $k$. We show by induction on the depth of the
  search tree that for every level $i$, there is a node $n$ of the search
  tree at level $i$ whose label $B'$ is a subset of some enhanced
  (deletion) backdoor to $\QQQ_q$ for $\Phi$ of size at most $k$, which
  completes the correctness proof. For the
  induction start, this clearly holds for the root of the search tree
  since its label $B'=\emptyset\subseteq B$. Now let $n$ be an inner
  node of the search tree with label $B'$ such that $B'\subseteq
  B$, which exists due to the induction hypothesis. Moreover, let
  $\Phi'=\Phi-(B'\cup\RU(\Phi,B'))$ and let $X_1,\dotsc,X_{q+2}$ be
  the sets of variables within the outermost $q+2$ quantifier blocks
  of $\Phi'$. Because $B$ is an enhanced (deletion)
  backdoor, there is an $i$ with $1\leq i \leq q+2$
  such that $X_i\subseteq (B\cup\RU(\Phi,B))$. Let $x \in X_i$ be the
  variable chosen by the algorithm to create the children of $n$ for
  $X_i$. Since $n$ has a child with label $B'\cup \{x\}$, we have shown
  the induction step for the case that $x \in B$. So suppose that $x
  \notin B$, then $x \in \RU(\Phi,B)$.
  Let $S$ be an inclusion-wise minimal subset of $B\setminus B'$
  separating $x$ from $V_{\exists}(\Phi)$ in
  $\primG(\Phi-B')$. Because of \Cref{lem:enhdetimsep}, there is an
  important $x$-$V_{\exists}(\Phi)$-separator $S'$ in $\primG(\Phi-B')$
  such that $B''=(B\setminus S)\cup S'$ is also an
  enhanced (deletion) backdoor to $\QQQ_q$ for $\Phi$. Moreover, by
  construction of the search tree, $n$ has a child $n'$ with label $B'\cup S' \subseteq B''$,
  as required; note that $B$ is then replaced by $B''$ for the later steps
  of the induction.

  Towards analyzing the runtime of the algorithm first note that its
  runtime is at most equal to the number of nodes of the
  search tree times the maximum time spend on any of its nodes.
  Since, the search tree has height at most $k$ and every of its nodes
  has at most $(q+2)(2+\nois{k})$ many children (because of \Cref{the:impsep}), we obtain that the
  search tree has at most $((q+2)(2+\nois{k}))^k$ many
  nodes. Moreover, since the time spend at any node is at most
  $\bigoh(\nois{k}k^2\|\Phi\|)$ (because of \Cref{the:impsep}), we obtain that the algorithm is fixed-parameter tractable
  parameterized by $k+q$.
\end{proof}

We are now ready to show \Cref{the:enhdet-con}.

\newcommand{\BBB}{\mathcal{B}}
\begin{proof}[Proof of \Cref{the:enhdet-con}]
  Let $\ell=q$ if $\CCC \in \{\DDD_q(\TCNF),\DDD_q(d\hy \AFF)\}$ and
  $\ell=0$ otherwise. We start by building the search tree $T$ used in the
  proof of \Cref{lem:EB-QA-B} to find an enhanced backdoor to $\QQQ_\ell$ for
  $\Phi$ of size at most $k$ in FPT-time with respect to $k+\ell$. Then, for
  every leaf $l$ of $T$ labeled yes, we change its second label to
  unknown and iteratively expand $T$ at every inner node, which is
  not yet labeled yes or no as follows.

  Namely, consider any inner node $n$ with label $B'$ (that is not yet
  labeled yes or no) and let
  $\Phi'=\Phi-(B'\cup \RU(\Phi,B'))$. If $\Phi' \in \CCC$, we update the label of $n$ to
  yes. Moreover, if $|B'|=k$, we update the label of $n$ to no and stop.
  Otherwise, we distinguish the following cases.

  If $\CCC=\DDD_q(\TCNF)$, consider any constraint $c \in \Phi'$ with
  $|c|>2$ and let $c'$ be any subset of $c$ with $|c'|=3$. Then, for every enhanced backdoor $B$ with
  $Y=\RU(\Phi',B)$ of
  $\Phi'$ to $\CCC$, it holds that there is a variable $x \in V(c')$
  such that $x \in B\cup Y$. Moreover, if $x \in B\cup Y$, then either:
  \begin{itemize}
  \item $x$ is existential and $x \in B$,
  \item $x$ is universal and $x \in B$, or
  \item $x$ is universal and $x \in Y$.
  \end{itemize}
  These observations allow us to exhaustively branch as follows.
  Namely, for every $x \in V(c')$, we add a
  new child to $n$ with label $B'\cup \{x\}$. Moreover, if $x$ is
  universal, then for every $k'\leq k-|B'|$ and every important
  $x$-$V_{\exists}(\Phi)$-separator $S$ in $\primG(\Phi-B')$ of size
  at most $k'$, we add
  to $n$ a child with label $B'\cup S$.

  If $\CCC=\DDD_{\exists}(\HORN)$, we first check whether $\Phi'$ is
  existential and if not, we simply label $n$ with no. Otherwise,
  the case for $\CCC=\DDD_q(\HORN)$ is analogously to the case of
  $\CCC=\DDD_q(\TCNF)$ only that now we consider any subset $c'$ of a
  constraint $c \notin \HORN$ that contains exactly two positive
  literals to branch on.

  If $\CCC=\DDD_q(d\hy\AFF)$, we first observe that $\Phi'$ cannot
  contain a constraint of arity more than $k-|B'|+d$, because
  otherwise we can correctly label $n$ with no. So assume that this is
  the case, then the case is analogous to the case for
  $\CCC=\DDD_q(\TCNF)$ only that now we branch on all variables of a
  constraint $c \notin \CCC$.

  This completes the description of the algorithm. The correctness of
  the algorithm is shown in the same manner as the correctness for the
  algorithm in \Cref{lem:EB-QA-B}.

  Towards analyzing the runtime of the algorithm first note that
  the search $T$ from \Cref{lem:EB-QA-B} can be build in FPT-time with respect to the
  parameter $k+\ell$ and that $T$ has at most $((q+2)(2+\nois{k}))^k$
  leaves. Moreover, for every leaf of $T$ we build (add) a search tree
  with at most:
  \begin{itemize}
  \item $(3\cdot \nois{k})^k$ nodes (in the case that $\CCC=\DDD_q(\TCNF)$),
  \item $((k+d)\nois{k})^k$ nodes (in the case that $\CCC=\DDD_q(d\hy\AFF)$),    
  \item $(2\cdot \nois{k})^k$ nodes (in the case that $\CCC=\DDD_{\exists}(\HORN)$).    
  \end{itemize}
  Since the time spent for each such new node is at most
  $\bigoh(k^2\nois{k}\|\Phi\|)$ (because of \Cref{the:impsep}), we obtain that the algorithm is
  fixed-parameter tractable parameterized by $k+\ell+d$, which completes
  the proof of the theorem.
\end{proof}

Next, we want to compute enhanced backdoors when $\CCC \in
\{\TTT_q^\omega,\VC_s\}$, i.e., we want to show \Cref{the:enhdet-stru-tw,the:enhdet-stru-vc} restated
below.
\theenhdetstrutw*

\theenhdetstruvc*

Towards that we observe that the problem of finding enhanced backdoors can be encoded as a Monadic Second Order (MSO) sentence and use a result on model checking~\cite{LokshtanovR0Z18}. To state this result, we first define MSO of graphs and then {\em unbreakable} graphs. Briefly, the syntax and semantics of MSO of graphs includes the following. 

\begin{itemize}
    \item Variables for vertices, edges, sets of vertices and sets of
      edges. 
    \item Universal and existential quantifiers that can be applied on the variables. 
    \item Logical connectives $\neg$, $\vee$, and $\wedge$. 
    \item $v \in U$ for a vertex variable $v$ and a vertex set
      variable $U$,
    \item $e \in F$ for an edge variable $e$ and an edge set variable $F$,
    \item The following binary relations.
    \begin{itemize}
        \item ${\sf inc}(e,u)$ where $e$ is an edge variable and $u$
          is a vertex variable. The semantics is that ${\sf inc}(e,u)$
          is true if and only if $u$ is an endpoint of $e$. 
        \item ${\sf adj}(u,v)$ where both $u$ and $v$ are vertex variables and the semantics is that ${\sf adj}(u,v)$ if and only if $u$ and $v$ are adjacent. 
         \item For any two variables $x$ and $y$ of same type, $x=y$ interprets that they are equal.    
    \end{itemize}
\end{itemize}
If an MSO sentence $\psi$ is true on a graph $G$, then we write $G\models \psi$. For an MSO sentence $\psi$, {\sc MSO[$\psi$]} is the algorithmic problem where the input is a graph $G$ and the output is yes if and only if $G\models \psi$. Additionally, if there is a parameter $k$ associated with $\psi$ and {\sc MSO[$\psi$]} is solvable in time $f(k)n^{O(1)}$, then we say that {\sc MSO[$\psi$]} parameterized by $k$ is non-uniform FPT (because the parameter $k$ is part of the problem and we can have different algorithms for different values of $k$),  where $f$ is a function and $n$ is the input length.

\begin{definition}
 For a graph $G$, a pair $(X,Y)$ with $X\cup Y=V(G)$ is called a {\em separation} if there is no edge with one endpoint in $X\setminus Y$ and the other in $Y\setminus X$. The order of $(X,Y)$ is $|X\cap Y|$.
\end{definition}

\begin{definition}
Let $G$ be a graph, and $s$ and $c$ be two non-negative integers. If there is a separation $(X,Y)$ of order at most $c$ with $|X\setminus Y|,|Y\setminus X|>s$, then $G$ is {\em $(s,c)$-breakable}, and such a separation is called an {\em $(s,c)$-witnessing separation}. If no such separation exists in $G$, then $G$ is {\em $(s,c)$-unbreakable}.     
\end{definition}

\begin{theorem}[\cite{LokshtanovR0Z18}]
\label{thm:MSOmodelchecking}
Let $\psi$ be an MSO sentence and  $c\colon {\mathbb{N}} \rightarrow {\mathbb{N}}$ be a function. There exists $s\colon {\mathbb{N}} \rightarrow {\mathbb{N}}$ such that if
{\sc MSO[$\psi$]} parameterized by $k$ is FPT on $(s(k), c(k))$-unbreakable graphs, then it is FPT on general graphs. 
\end{theorem}

Thus, to check the existence of an enhanced deletion backdoor to
$\CCC{}$, where $\CCC\in \{\TTT^w,\VC_s\}$, it is sufficient to
design an MSO sentence with parameter $k$, a function $c$, and an FPT
algorithm on $(s(k), c(k))$-unbreakable graphs. Next we explain this
task, which is similar to the $t$-{\sc Pendant} problem considered in \cite{LokshtanovR0Z18} as an application of Theorem~\ref{thm:MSOmodelchecking}.

Now, we explain how to encode our problem as an MSO sentence. Recall that for a QBF formula $\Phi$, 
$\primG(\Phi)$ is the primal graph of $\Phi$. We define a graph
$\primincG(\Phi)$, which is a super graph of $\primG(\Phi)$, as
follows: $\primincG(\Phi)$ has a vertex for every variable (for simplicity, sometimes called just ``variable'' rather than a vertex) and a
vertex for every constraint in $\Phi$. There is an edge between two
variables if they occur together in a constraint. Also, there is an edge
between a variable and a constraint if the variable occurs in the
constraint. Let $E_1$ be the set of edges between the variables and let
$E_2$ be the set of edges between variables and constraints. Notice that
the subgraph induced on $E_1$ is the primal graph $\primG(\Phi)$ and
the subgraph induced on $E_2$ is the incidence graph of $\Phi$. For
the ease of presentation, let us denote the graph $\primincG(\Phi)$ by
$G$ and the subgraphs of $G$ induced on $E_1$ and $E_2$ by $G_1$ and
$G_2$, respectively. Let $V_{\forall}$ and $V_{\exists}$ denote the
sets of vertices of $G$ corresponding to the universal and existential
variables of the \QBF formula, respectively. 
Define an MSO formula $\psi$  as follows: 

\begin{eqnarray*}
\psi=\exists B \exists Y && [(|B|\leq k) \wedge (Y\subseteq V_{\forall}\setminus B)\wedge (\forall y \in Y \forall u\in V(G_1)\setminus (Y\cup B)\ \neg {\sf adj}(y,u))\\
&&\wedge (G_2-(Y\cup B) \models {\sf tw}_w)].
\end{eqnarray*}

Here, $B$ corresponds to the enhanced deletion backdoor and $Y$ corresponds to the set of vertices in $U(\Phi,B)$. Also, ${\sf tw}_w$ is a standard MSO sentence that tests whether the treewidth of a graph is at most $w$ \cite{FominLMS12} and this follows from the fact that the set of graphs with treewidth at most $w$ is minor closed. That is, there is a finite set ${\cal F}$ of graphs such that a graph $H$ has treewidth at most $w$ if and only if all the graphs in ${\cal F}$ are not minors of $H$. Also, $|B|\leq k$ can be encoded using the following MSO sentence: 
$\exists v_1,\ldots,v_k [\forall u\in V(G)  \
u=v_1\vee u=v_2 \vee \ldots \vee v_k=u\vee \neg (u\in B)]$. So, we have the following observation. 

\begin{observation}
$G\models \psi$ if and only if $\Phi$ has an enhanced deletion backdoor of size $k$ to $\TTT^w$.
% \todo{G: Could we just say "if and only if $\Phi$ has an enhanced backdoor to ... of size $k$? i.e. not mention $G$ after iff?}
\end{observation}

Similarly, we can define an MSO formula $\psi'$ that encodes enhanced
deletion backdoors to $\VC_s$ for $\Phi$. To obtain $\psi'$, replace
${\sf tw}_s$ with an MSO sentence ${\sf vc}_s$ that tests whether the vertex cover number of a graph is at most $k$.

\begin{observation}
$G\models \psi'$ if and only if $\Phi$ has an enhanced deletion backdoor of size $k$ to $\VC_s$ .     
% \todo{G: Same comment as above.}
\end{observation}

Now, we prove the following two simple lemmas. Recall the definitions of $G$, $G_1$, and $G_2$.  

\begin{lemma}
\label{lem:sep_in_primal}
Let $(X,Y)$ be a separation of $G_1$. Then, there exist $X'\supseteq X$ and $Y'\supseteq Y$ such that $(X',Y')$ is a separation of $G$ of order $|X\cap Y|$.    
\end{lemma}

\begin{proof}
For any constraint $c'$, there do not exist two variables $x,y\in c'$ such that $x\in X\setminus Y$ and $y\in Y\setminus X$. If there are two such variables, then there will be an edge between $x$ and $y$ in $G_1$, contradicting the fact that $(X,Y)$ is a separation in $G_1$. Let $X'\supseteq X$ be the set such that we add all the vertices corresponding to the constraints that contain a variable from $X\setminus Y$, to $X'$. We add the remaining vertices corresponding to the constraints to $Y$ and obtain $Y'$. Then, 
$(X',Y')$ is a separation of $G$ of order $|X\cap Y|$.      
\end{proof}

\begin{lemma}
\label{lem:primal_unbreakable}
If $G$ is $(s(k),k)$-unbreakable, then $G_1$ is $(s(k),k)$-unbreakable.      
\end{lemma}
\begin{proof}
Suppose $G_1$ is $(s(k),k)$-breakable and let $(X,Y)$ be an $(s(k),k)$-witnessing separation. Then, by Lemma~\ref{lem:sep_in_primal}, there exist $X'\supseteq X$ and $Y'\supseteq Y$ such that $(X',Y')$ is a separation of $G$ of order $|X\cap Y|$. This implies that $G$ is $(s(k),k)$-breakable.
\end{proof}

Next, we prove that {\sc MSO[$\psi$]} and  {\sc MSO[$\psi'$]} parameterized by $k$ are FPT on unbreakable graphs. Let $c$ be the function defined as follows. For all $k\in \mathbb{N}_0$, $c(k)=k$. Let $s$ be the function in Theorem~\ref{thm:MSOmodelchecking}.

\begin{lemma}
\label{lem:unbreakable_size_tw}
Suppose $G\models \psi$ and $G$ 
is $(s(k),k)$-unbreakable. Then, there is an enhanced deletion backdoor $B$ to $\TTT^w$ for $\Phi$ of size at most $k$ such that either 
$|U(\Phi,B)|\leq s(k)$ or $|V(G_1)\setminus (U(\Phi,B) \cup B)| \leq s(k)$. 
\end{lemma}

\begin{proof}
Suppose for any enhanced deletion backdoor $B$ to $\TTT^w$ for $\Phi$ of size at most $k$, if both  $|U(\Phi,B)|$ and  $|V(G_1)\setminus (U(\Phi,B) \cup B)|$ are greater than $s(k)$, then $G_1$ is $(s(k),k)$-breakable. Then, by Lemma~\ref{lem:primal_unbreakable}, $G$ is $(s(k),k)$-breakable which is a contradiction.     
\end{proof}

Similarly, one can prove the following lemma. 

\begin{lemma}
\label{lem:unbreakable_size_vs}
Suppose $G\models \psi'$ and $G$ 
is $(s(k),k)$-unbreakable. Then, there is an enhanced deletion backdoor $B$ to $\VC_s$ for $\Phi$ of size at most $k$ such that either 
$|U(\Phi,B)|\leq s(k)$ or $|V(G_1)\setminus (U(\Phi,B) \cup B)| \leq s(k)$. 
\end{lemma}

\begin{lemma}
\label{lem:FPT-UNB}
{\sc MSO[$\psi$]} and {\sc MSO[$\psi'$]} are  FPT on $(s(k),k)$-unbreakable graphs.         
\end{lemma}

\begin{proof}
We give the proof for {\sc MSO[$\psi$]} and the proof for {\sc MSO[$\psi'$]} has similar arguments. We use the color coding technique to design the required FPT-algorithm. For each vertex in $G_1$, color it with red or blue uniformly at random. Let $R$ be the set of vertices colored red and $W$ be the vertices colored blue. 
Let $X$ be the set of vertices in the union of all components in $G_1-R$ that contain at least one existential variable.
Let $B_1=\delta_{\Phi} (X)$ 
and $B_2=V_{\exists}(\Phi)\setminus (B_1\cup X)$. If $B_1\cup B_2$ is an enhanced deletion backdoor to $\TTT^w$ for $\Phi$ of size at most $k$, then output yes. Otherwise, the algorithm tries to compute an enhanced deletion backdoor as follows. Recall that $U(\Phi,R)$ denotes the union of all components in $G_1-R$ that only contain universal variables.
Let ${\mathcal K}$ be the set of components with the following property. Any component $C$ in $U(\Phi,R)$ with incidence treewidth strictly more than $w$ belongs to ${\mathcal K}$. 
Let $Y$ be the set of vertices in the union of all components in ${\mathcal K}$. 
Let $P_1=\delta_{\Phi}(Y)$. If 
$G_2-(Y\cup P_1)$ has a treewidth $w$-modulator $P_2\subseteq V(G_1)$ (i.e., treewidth of $G_2-(Y\cup P_1\cup P_2)$ is at most $w$) of size $k-|P_1|$, then $P_1\cup P_2$ is an enhanced deletion backdoor of size $k$  (we will prove this statement later in the proof of the lemma) and the algorithm outputs yes. 
Otherwise, the algorithm outputs no.

Next, we prove that if {\sc MSO[$\psi$]} is a yes-instance, then the above algorithm outputs yes 
with probability at least $2^{-(s(k)+k)}$. Suppose there is an enhanced deletion backdoor $B$ of size at most $k$. Then by Lemmas~\ref{lem:unbreakable_size_tw} and \ref{lem:unbreakable_size_vs} either $U(\Phi,B)\leq s(k)$ or $|V(G_1)\setminus (U(\Phi,B)\cup B)|\leq s(k)$.

\medskip
\noindent
\textbf{Case 1: $|V(G_1)\setminus (U(\Phi,B)\cup B)|\leq s(k)$. }
In this case with probability at least $2^{-(s(k)+k)}$, all the vertices in $Z=V(G_1)\setminus (U(\Phi,B)\cup B)$ are colored blue and all vertices in $B$ are  colored red. This implies that all the connected components of $G_1[Z]$ are connected components in $G_1-R$. Moreover, by definition of $U(\Phi,B)$, we know that each connected component of $G_1-(B\cup U(\Phi,B))=G_1[Z]$ contains a variable from $V_{\exists}(\Phi)$. This implies that the set $X$ constructed in the algorithm is equal to $Z$. Therefore, $B'=B_1\cup B_2 \subseteq B$ and $G_2-(U(\Phi,B)\cup B)=G_2-(U(\Phi,B')\cup B')$. So, $B'$ is an enhanced deletion backdoor to $\TTT^w$ for $\Phi$ of size at most $|B|$.

\medskip
\noindent
\textbf{Case 2: $U(\Phi,B)\leq s(k)$. }
In this case, with probability at least $2^{-(s(k)+k)}$, all the vertices in $U(\Phi,B)$ are colored blue and all vertices in $B$ are colored with red. 
Then, any component in $U(\Phi,B)$ is also a component in $U(\Phi,R)$.  
That is, $U(\Phi,B)\subseteq U(\Phi,R)$. Since any component in ${\mathcal K} \subseteq U(\Phi,R)$ has treewidth strictly more than $w$ and $B\subseteq R$, we have that ${\mathcal K} \subseteq U(\Phi,B)$. Notice that any component in $G_1-B$
that has incidence treewidth strictly more than $w$, belongs to $U(\Phi,B)$, and hence it also belongs to $U(\Phi,R)$. This implies that ${\mathcal K}$ contains all the components of $G_1-B$ that have incidence treewidth strictly more than $w$.   
Since, 
${\mathcal K} \subseteq U(\Phi,B)$, 
the set $P_1=\delta_{\Phi}(Y)$ constructed in the algorithm is a subset of $B$.
Now consider the set $B\setminus P_1$. Since  
${\mathcal K}$ contains all the components of $G_1-B$ that have incidence treewidth strictly more than $w$, 
$B\setminus P_1$ is a treewidth $w$-modulator of $G_2-(Y\cup P_1)$ and $|B\setminus P_1|\leq k-|P_1|$. Now, consider a  treewidth $w$-modulator $P_2$ of $G_2-(Y\cup P_1)$ of size $k-|P_1|$. We claim that $P_1\cup P_2$ is an enhanced deletion backdoor to $\TTT^w$ for $\Phi$. 
Consider a connected component $C$ in $G_1-(P_1\cup P_2)$ that contains at least one variable from $V_{\exists}(\Phi)$. Clearly, $V(C)$ is disjoint from $Y$.  
Then, $V(C)$ is a subset of $V(G_2)\setminus (Y \cup P)$. This implies that $V(G_1)\setminus (U(\Phi,P_1\cup P_2)\cup P_1\cup P_2)$ is a subset of $G_2-(Y\cup P_1)$. This implies that 
$G_2-(U(\Phi,P_1\cup P_2)\cup P_1\cup P_2)$
is a subgraph of $G_2-(Y\cup P_1 \cup P_2)$. 
Since, 
$G_2-(Y\cup P_1 \cup P_2)$ has treewidth at most $w$, the subgraph  $G_2-(U(\Phi,P_1\cup P_2)\cup P_1\cup P_2)$ has treewidth at most $w$.

Clearly, when the algorithm outputs yes, {\sc MSO[$\psi$]} is a yes-instance. That is, if {\sc MSO[$\psi$]} is a no-instance, then the algorithm outputs no with probability one. So if we repeat the above algorithm $2^{s(k)+k} \log n$ times and output yes if the algorithm outputs yes at least once, then the success probability will be at least $1-\frac{1}{n}$. Moreover the running time will be upper bounded by $f(k)n^{O(1)}$ for some function $f$. 
This algorithm can be derandomized using standard techniques using perfect hash families. 
\end{proof}

Now, we prove Theorems~\ref{the:enhdet-stru-tw} and \ref{the:enhdet-stru-vc}. 

\begin{proof}[Proof of Theorem~\ref{the:enhdet-stru-tw}]
Let $\Phi$ be the given \QBF formula. Suppose 
there exists an enhanced deletion backdoor to $\TTT^w_q$ for $\Phi$ of size at most $k$.   
We apply Lemma~\ref{lem:EB-QA-B} and obtain 
an enhanced deletion backdoor $B_1$ to $\QQQ_q$ for $\Phi$ of size at most $k$. By Theorem~\ref{thm:MSOmodelchecking} and Lemma~\ref{lem:FPT-UNB}, 
we know that the existence of an enhanced deletion backdoor to $\CCC$ for $\Phi$ of size at most can be tested in FPT-time. One can use self-reducibility to find an   enhanced deletion backdoor $B_2$ to $\TTT^w$ of size at most $k$. Then, $B_1\cup B_2$ is the required enhanced deletion backdoor to $\TTT_q^w$ for $\Phi$ of size at most $2k$.       
\end{proof}

\begin{proof}[Proof of Theorem~\ref{the:enhdet-stru-vc}]
By Theorem~\ref{thm:MSOmodelchecking} and Lemma~\ref{lem:FPT-UNB}, 
we know that existence of an enhanced deletion backdoor to $\VC_s$ for $\Phi$ of size at most can be tested in FPT-time. One can use self-reducibility to find the required enhanced deletion backdoor of size at most $k$.      
\end{proof}

\section{Comparison to Dependency Backdoors}
\label{sec:comp-dep-back}

Here, we will compare our notion of backdoors to the backdoors
introduced by Samer and Szeider~\cite{SamerSzeider07a} that take into
account the dependencies among variables in the quantifier prefix of a \QBF{}
formula. To the best of our knowledge, there have not been any other successful attempts in the literature. To distinguish these backdoors from our backdoors, we will
refer to them as dependency backdoors. Informally, a set of variables $B$
is a dependency backdoor to a class $\CCC$ of \QBF{} formulas if $B$
is a backdoor to $\CCC$ and additionally $B$
contains every variable that is before a variable $b \in B$ in the
quatifier prefix and that $b$ depends on; in other words $B$ is closed
under adding all variables that any variable in $B$ depends on. Note
that since every dependency backdoor is also a backdoor, it follows
that dependency backdoors are always at least as large as backdoors.
It therefore only remains to be asked whether the difference in their
sizes can be arbitrary large, which we answer affirmatively in this section.
In order to define the notion of dependence between variables, we first need to
introduce the notion of dependency schemes (introduced by Samer and
Szeider~\cite{SamerSzeider07a}).

For a binary relation $\RR$ over some set $V$ we write
$\overline{\RR}$ to denote its inverse, i.e., $\overline{\RR}=\{(y,x):
(x,y)\in \RR\}$, and we write $\RR^{*}$ to denote the reflexive and
transitive closure of $\RR$ i.e., the smallest set $\RR^{*}$
such that $\RR^{*} = \RR \cup \{(x, x) : x \in V \}\cup\{(x, y) :
\exists z$ such that $(x, z) \in \RR^{*}$ and $(z, y) \in
\RR\}$.
Moreover, we let $\RR(x) = \{y : (x, y) \in \RR\}$ for $x \in V$ and
$\RR(X) = \cup_{x\in X} \RR(x)$ for $X \subseteq V$.
Given a \QBF{} formula $\Phi$ we will also need the
following binary relation over $\var(\Phi)$:
$\RR_{\Phi}=\SB (x,y) \SM x,y\in \var(\Phi),$ $x$ is before (to the
left of) $y$ in the quantifier prefix of $\Phi\SE$.

We also need the notion of
\emph{shifting}, which takes
some subset of the variables of $\Phi$ and puts them
together with their quantifiers, in the same relative order, to the
end (\emph{down-shifting}) or to the beginning (\emph{up-shifting}) of
the prefix.
\begin{definition}[Shifting,~\cite{SamerSzeider07a}] Let $\Phi$ be a
  \QBF{} formula
  and $A\subseteq \var(\Phi)$. We say that 
  $\Phi'$ is obtained from $\Phi$ by \emph{down-shifting}
  (\emph{up-shifting}) $A$, in symbols $\Phi' = S^{\downarrow}(\Phi,A)$ ($\Phi'
  = S^{\uparrow}(\Phi,A)$), if $\Phi'$ is obtained from $\Phi$ by reordering
  quantifiers together with their variables in the prefix such that the following holds:
  \begin{enumerate}
  \item $A=\RR_{\Phi'}(x)$ ($A=\overline{\RR}_{\Phi'}(x)$) for some $x\in \var(\Phi)$ and
  \item $(x,y)\in \RR_{\Phi'}$ iff $(x,y)\in \RR_{\Phi}$ for all $x,y\in A$ and
  \item $(x,y)\in \RR_{\Phi'}$ iff $(x,y)\in \RR_{\Phi}$ for all $x,y\in \var(\Phi)\setminus A$.
  \end{enumerate}
\end{definition}

We are now ready go define dependency schemes, in particular,
cumulative and tractable dependency schemes, which are required for
the application of dependency backdoors.
\begin{definition}[(Cumulative) Dependency Scheme,~\cite{SamerSzeider07a}]
  A  \emph{dependency scheme} $D$ assigns to each \QBF{} formula $\Phi$ a
  binary relation $D_{\Phi}\subseteq \RR_{\Phi}$ such that $\Phi$ and
  $S^{\downarrow}(\Phi,D_{\Phi}^{*}(x))$ are satisfiability-equivalent
  for every variable $x$ of $\Phi$. Moreover, $D$ is
  \emph{cumulative}, if $\Phi$ and
  $S^{\downarrow}(\Phi,D_{\Phi}^{*}(X))$ are satisfiability-equivalent
  for every subset $X$ of variables $\Phi$. Finally, $D$ is
  \emph{tractable} if it can be computed in time polynomial in the
  size of $\Phi$.
\end{definition}
In the following, we assume that we are given a cumulative and
tractable dependency scheme $D$ (see~\cite{SamerSzeider07a} for
details on suitable dependency schemes).

Let $\CCC$ be a tractable class of \QBF{} formulas and let $\Phi$ be a
\QBF{} formula. A set $B$ of variables of $\Phi$ is a \emph{dependency
  backdoor} of $\Phi$ to $\CCC$ if $B$ is a backdoor of $\Phi$ to
$\CCC$ and $\overline{D_\Phi}^*(B)=B$. Observe that the exact definition of
dependency backdoor therefore depends on the employed dependency
scheme (which is required to be cumulative and tractable in order to
ensure that the backdoor approach is amenable). Note also that since
dependency backdoors are backdoors, the former are always at least as
large as the latter. The following theorem now shows that dependency
backdoors can be arbitrarily larger, even on \QBF{} formulas with only one
quantifier alternation that have a backdoor of size $1$.
We show the theorem for all of the dependency schemes
introduced in~\cite{SamerSzeider07a}; we remark, however, that
the reduction works for all known dependency schemes.

\begin{theorem}
  Let $\CCC \in \{\TCNF,\HORN\}$ and let $n$ be an integer. There is a
  \QBF{} formula $\Phi_n$ that has a backdoor into $\CCC$ of size $1$,
  but no dependency backdoor of size less than $n$ under any
  dependency scheme introduced in~\cite{SamerSzeider07a}, i.e., the
  trivial, standard and triangle dependency scheme. Moreover,
  $\Phi_n$ has only one quantifier alternation.
\end{theorem}
\begin{proof}
  Let $\Phi_n$ be the \QBF{} formula with quantifier prefix $\forall
  y_1,\dotsc,y_{n+2}\exists x$, whose matrix has all possible clauses
  $c$ with three literals such that:
  \begin{itemize}
  \item $x \in \var(c)$,
  \item $|\var(c)\cap Y|=2$ with $Y=\{y_1,\dotsc,y_{n+2}\}$, and
  \item $c$ contains at most one positive literal on the variables in
    $Y$.
  \end{itemize}
  Then, $\{x\}$ is a backdoor for $\Phi_n$ to $\CCC$ for every $\CCC \in
  \{\TCNF,\HORN\}$. Moreover, any backdoor for $\Phi_n$ to $\CCC$ that does not
  contain $x$ has to contain at least $n$ variables from $Y$. Finally,
  because every variable $y \in Y$ is contained in a clause containing
  $x$ and a clause containing $\lnot x$, we obtain that $Y \subseteq \overline{D_\Phi}^*(x)$ under
  any dependency scheme $D$ introduced
  in~\cite{SamerSzeider07a}, namely, the
  trivial, standard and triangle dependency scheme
  (see~\cite{SamerSzeider07a} for more details on these dependency schemes).
  Therefore, any dependency backdoor of
  $\Phi$ to $\CCC$ has size at least $n$.
\end{proof}

\section{Concluding Remarks}
\label{sec:discussion}

We have proven that the classical notion of a backdoor, perhaps contrary to belief, {\em is} applicable to $\QBF$, and discovered the first non-trivial backdoor evaluation algorithms. We complemented this by novel hardness results with the squishing technique. To obtain even more general results and to be able to combine the backdoor approach with existing FPT-classes that fall outside the classical tractable $\QBF$ fragments (e.g., bounded treewidth) we defined the notion of {\em guarded universal sets} and {\em enhanced} backdoors which led to a wealth of new FPT-results. Let us now discuss some future research questions.

\paragraph*{Strong backdoors versus deletion backdoors.} In line with the backdoor literature for SAT we defined our backdoor notion as a strong backdoor where each instantiation $\Psi[\tau]$ has to be included in the backdoor class. As remarked, strong backdoors coincide with variable deletion backdoors in the clausal setting, but when constraints can be given via arbitrary relations this is no longer the case. For enhanced backdoors the difference between strong and deletion backdoors diverges even more. We considered both but concentrated on deletion backdoors for the structural classes (e.g., treewidth). Here, the strong backdoor formulation is an intriguing open question, and we remark that the detection problem can be handled with the corresponding tools for SAT~\cite{GaspersSzeider13}. The evaluation problem, however, seems vastly different, and the main complication is that treewidth is not stable under partial assignments, in the sense that each $\Psi[\tau]$ might be in the class even though $\Psi$ is outside. We thus conjecture that enhanced strong backdoor evaluation into classes of bounded treewidth is not FPT, but explicitly proving hardness is interesting, and could lead to new general techniques.

\paragraph*{Broader tractable classes.} We have exhausted the classically tractable classes of $\QBF$. However, if one considers the \emph{quantified constraint satisfaction problem} allowing atoms to use relations over arbitrary finite domains, then one could borrow the tractable classes from the {\em CSP dichotomy theorem}~\cite{bulatov17,zhuk17}. For example, $\AFF$ could conceivably be generalized to arbitrary equations over finite groups, but further generalizing it to e.g.\ {\em Maltsev} constraints appears much more challenging. In a similar vein, $\TCNF$ could first be generalized to relations invariant under {\em majority}, and, if successful, then to {\em near unanimity operations}, or even {\em edge} operations. It should be noted that such generalizations likely need significant new ideas since our proofs make extensive use of properties of Boolean clauses, but could open up large new tractable classes.

\paragraph*{A complete classification.}
We obtained an almost complete classification with respect to the
classical tractable classes of $\QBF$ and only a few open cases
remain. First, backdoor evaluation to $\HORN$ turned out to be
W[1]-hard even in the presence of bounded quantifiers and bounded
arity, but we currently lack complementary upper bounds, such as XP
membership. Second, is it possible to drop the parameter
$\SA(\Phi,Y)$, i.e., the maximum number of universal variables
occurring in any clause, while still maintaining fixed-parameter
tractability for enhanced (deletion) backdoors?

Another question concerns the complexity of $\QBF$ with backdoors of size $1$ to $\CCC \in \{\TCNF, \AFF\}$. The faster squishing procedure allows us to reduce from $\QBF$ on general $\CNF$ formulas to $4$-\textsc{Disjunct}-$QBF(\CCC)$ formulas, proving that the problem is PSpace-hard in the latter case. Combined with \Cref{lem:from-disjunct-to-backdoor}, this implies that $\QBF$ is PSpace-hard on formulas with a backdoor of size $2 = \log_2 4$. The question about backdoors of size $1$ is equivalent to asking about the polynomial-time complexity of $2$-\textsc{Disjunct}-$\QBF(\CCC)$.

\paragraph*{New problems.} We have demonstrated that the classical backdoor notion is applicable to $\QBF$, and it is thus tempting to see whether there are any more PSpace-hard problems which have been omitted, or whether the FPT-results could propagate to even larger complexity classes. For example, Boolean formulas equipped with Henkin quantifiers, {\em dependency quantified Boolean formulas}, is NEXPTIME-complete but could conceivably be attacked with similar methods.

\paragraph*{Analysis of $k$-\textsc{Disjunct}-$\QBF(\CCC)$ formulas.}
Finally, we want to note here that our work (together
with~\cite{DBLP:conf/ijcai/ErikssonLOOPR24}) has laid the foundation
towards a deeper understanding of the parameterized complexity of
$k$-\textsc{Disjunct}-$\QBF(\CCC)$ formulas and such an understanding
seems to be crucial to make any progress on the notoriously difficult
\QBF problem. Apart from the applications for \QBF obtained here and
in~\cite{DBLP:conf/ijcai/ErikssonLOOPR24}, another natural candidate
where a deeper understanding of $k$-\textsc{Disjunct}-$\QBF(\CCC)$ formulas
can likely be useful are the prominent open problems concerning the
parameterized complexity of \QBF with respect to  structural parameters
in between vertex cover and treewidth, e.g., feedback vertex set, vertex integrity, and
treedepth~\cite{DBLP:conf/lics/FichteGHSO23}. In particular, since
both feedback vertex set and vertex integrity are defined via variable
deletion sets (i.e., deletion backdoors), the corresponding problems
can also be naturally formulated on $k$-\textsc{Disjunct}-$\QBF(\CCC)$
formulas, which offers a novel perspective on these problems and
further motivates the study of $k$-\textsc{Disjunct}-$\QBF(\CCC)$ formulas.

\section*{Acknowledgements}

The second author was partially supported by 
the Swedish Research Council (VR)
under grant 2021-04371 and 2025-04487.
The fourth author was supported by VR
under grant 2024-00274.

\newpage

\bibliographystyle{alpha}
\bibliography{references}

\end{document}